%% file: Bruno_Fava_JMP.tex
\documentclass[12pt]{article}

\input{preamble}

\title{\vspace{-65pt}Training and Testing with Multiple Splits: A\\Central Limit Theorem for Split-Sample Estimators}

\author{Bruno Fava\thanks{Department of Economics, Northwestern University. 
Contact: brunofava@u.northwestern.edu. 
I am incredibly grateful to Federico Bugni, Ivan Canay, Joel Horowitz, and Dean Karlan for their unparalleled advising. 
I thank Eric Auerbach, Denis Chetverikov, Federico Crippa, Georgy Egorov, Annie Liang, and Amilcar Velez for helpful discussions. 
All errors are my own.
This research was supported in part through the computational resources and staff contributions provided for the Quest high performance computing facility at Northwestern University which is jointly supported by the Office of the Provost, the Office for Research, and Northwestern University Information Technology.}}
\date{\today}

\begin{document}

\maketitle

\begin{center}
\Large
\vspace{-25pt}
\href{https://bfava.com/files/Bruno_Fava_JMP.pdf}{[Click here for the latest version]}
\vspace{10pt}
\end{center}

\input{abstract}

\clearpage

\input{intro}

\input{setup}

\input{z_estimators}

\input{inference}

\input{inference_normal}

\input{inference_diff_performance}

\input{reproducibility}

\input{application_ghana}

\input{application_gml}

\input{conclusion}

\input{bibliography}

\begin{appendices}

\crefalias{section}{appendix}
\crefalias{subsection}{appendix}
\crefalias{subsubsection}{appendix}

\input{appendix_etabar.tex}

\section{Proofs and Extra Definitions} \label{app.proofs}

The following notation is used throughout the proofs. 
If unspecified, $X \Pto Y$ denotes convergence in probability uniformly in $P \in \cP$, that is, for every $\varepsilon > 0$, $\sup_{P \in \cP} P \lrp{\lrm{X - Y} > \varepsilon} \to 0$. 
$X_n = o_P(a_n) \iff X_n / a_n \Pto 0$. 
$X_n = O_P(a_n) \iff (\forall \varepsilon > 0, \exists M > 0 \text{ and } N > 0 \text{ s.t. } n > N \implies \sup_{P \in \cP} P \lrp{\lrm{\frac{X_n}{a_n}} > M} < \varepsilon)$. 
$\leadsto$ means weak convergence uniformly in $P \in \cP$.

\input{proofs_z_estimators}

\input{proofs_inference}

\input{proofs_reproducibility_basic}

\input{proofs_reproducibility_measure}

\input{appendix_ghana}
\input{appendix_hte}

\input{modeling_power}
\input{averages}
\input{general}

\input{fast_convergence}

\input{appendix_comparison.tex}

\input{appendix_additional_results}

\input{application_rct}

\end{appendices}

\end{document}

%% file: preamble.tex
\usepackage{graphicx}
\usepackage{xcolor}
\usepackage[margin=1.25in, top=1.5in, bottom=1.5in]{geometry}
\usepackage{amsmath}
\usepackage{amssymb}
\usepackage{xparse} 
\usepackage{amsthm} 
\usepackage{enumitem} 
\usepackage{appendix}
\usepackage{mathabx} 
\usepackage[hidelinks]{hyperref} 
\usepackage{cleveref} 
\usepackage{mathrsfs} 
\usepackage{booktabs} 
\usepackage{array} 
\usepackage{threeparttable} 
\usepackage{algorithm} 
\usepackage{algpseudocode} 
\usepackage{makecell} 
\usepackage[bottom]{footmisc} 

\usepackage{hyperref}
\hypersetup{
  pdftitle={Fava - JMP},
  pdfauthor={Bruno Fava},
  pdfsubject={Training and Testing with Multiple Splits}
}

\newtheorem{lemma}{Lemma}

\newtheorem{theoreminner}{Theorem} 
 
\newenvironment{theorem}[1][]{
  \pushQED{\qed}
  \begin{theoreminner}[#1]
}{
  \popQED
  \end{theoreminner}
}
\newtheorem{propositioninner}{Proposition} 
 
\newenvironment{proposition}[1][]{
  \pushQED{\qed}
  \begin{propositioninner}[#1]
}{
  \popQED
  \end{propositioninner}
}
\newtheorem{exampleinner}{Example} 
 
\newenvironment{example}[1][]{
  \pushQED{\qed}
  \begin{exampleinner}[#1]
}{
  \popQED
  \end{exampleinner}
}
\NewDocumentEnvironment{RevExample}{m}{ 
    \begingroup
    \renewcommand{\theexampleinner}{\ref{#1}}%
    \begin{example}[Revisited]%
}{%
    \end{example}%
    \endgroup
}
\newtheorem{assumptioninner}{Assumption} 
 
\newenvironment{assumption}[1][]{
  \pushQED{\qed}
  \begin{assumptioninner}[#1]
}{
  \popQED
  \end{assumptioninner}
}
\newtheorem{remarkinner}{Remark} 
 
\newenvironment{remark}[1][]{
  \pushQED{\qed}
  \begin{remarkinner}[#1]
}{
  \popQED
  \end{remarkinner}
}

\numberwithin{equation}{section}
\numberwithin{theoreminner}{section}
\numberwithin{propositioninner}{section}
\numberwithin{lemma}{section}
\numberwithin{corollary}{section}
\numberwithin{remarkinner}{section}
\numberwithin{definition}{section}
\numberwithin{assumptioninner}{section}

\crefname{equation}{}{}
\Crefname{equation}{}{}
\crefname{theoreminner}{Theorem}{Theorems}
\Crefname{theoreminner}{Theorem}{Theorems}
\crefname{propositioninner}{Proposition}{Propositions}
\Crefname{propositioninner}{Proposition}{Propositions}
\crefname{exampleinner}{Example}{Examples}
\Crefname{exampleinner}{Example}{Examples}
\crefname{assumptioninner}{Assumption}{Assumptions}
\Crefname{assumptioninner}{Assumption}{Assumptions}
\crefname{remarkinner}{Remark}{Remarks}
\Crefname{remarkinner}{Remark}{Remarks}
\crefname{appendix}{Appendix}{Appendices}
\Crefname{appendix}{Appendix}{Appendices}

\newlist{assumptionenum}{enumerate}{1}
\setlist[assumptionenum]{label=\upshape(\roman*), ref=\theassumptioninner(\roman*)}
\crefname{assumptionenumi}{Assumption}{Assumptions}
\Crefname{assumptionenumi}{Assumption}{Assumptions}

\newcommand{\Pto}{\xrightarrow{P}}
\newcommand{\Pnto}{\xrightarrow{P_n}}

\newcommand{\etahat}{{\hat{\eta}}}
\newcommand{\etanot}{{\eta_{0}}}
\newcommand{\etabar}{{\bar{\eta}}}
\newcommand{\betahat}{{\hat{\beta}}}
\newcommand{\deltahat}{{\hat{\delta}}}
\newcommand{\etatilde}{{\tilde{\eta}}}
\newcommand{\etastar}{\eta^*}

\newcommand{\etastarP}{{\eta^{*}_{P}}}
\newcommand{\etastarPn}{{\eta^{*}_{P_n}}}

\newcommand{\thetahat}{\hat{\theta}}
\newcommand{\thetaeta}{\theta_{\eta}}
\newcommand{\thetaetahat}{\theta_{\etahat}}
\newcommand{\thetahatetahat}{\thetahat_{\etahat}}
\newcommand{\sigmahat}{\hat{\sigma}}

\newcommand{\Sigmahat}{\hat{\Sigma}}
\newcommand{\tauhat}{\hat{\tau}}
\newcommand{\taubar}{\bar{\tau}}
\newcommand{\Psihat}{\hat{\Psi}}
\newcommand{\Psidot}{\dot{\Psi}}
\newcommand{\Psihatmin}{\hat{\Psi}_{\rm min}}
\newcommand{\CIalpha}{\widehat{\rm CI}_{\alpha}}
\newcommand{\Mbar}{\bar{M}}

\newcommand{\cP}{\mathcal{P}}
\newcommand{\cW}{\mathcal{W}}
\newcommand{\cF}{\mathcal{F}}
\newcommand{\cX}{\mathcal{X}}
\newcommand{\cQ}{\mathcal{Q}}
\newcommand{\cA}{\mathcal{A}}

\newcommand{\cN}{\mathcal{N}}

\newcommand{\cS}{\mathcal{S}}
\newcommand{\cR}{\mathcal{R}}

\newcommand{\cV}{\mathcal{V}}

\newcommand{\R}{\mathbb{R}}
\newcommand{\G}{\mathbb{G}}

\newcommand{\I}[1]{\mathbb{I}\left(#1\right)}

\newcommand{\bP}{\mathbb{P}}

\newcommand{\s}{\mathsf{s}}
\newcommand{\stilde}{\tilde{\s}}
\newcommand{\xitilde}{\tilde{\xi}}
\newcommand{\bkn}{\lrbk{n}}
\newcommand{\hdot}{\dot{h}}
\newcommand{\bhat}{\hat{b}}
\newcommand{\dhat}{\hat{d}}
\newcommand{\betastarP}{\beta^*_{P}}
\newcommand{\dstarP}{d^*_{P}}

\newcommand{\lrp}[1]{\left( #1 \right)}
\newcommand{\lrbk}[1]{\left[ #1 \right]}
\newcommand{\lrbc}[1]{\left\{ #1 \right\}}
\newcommand{\lrm}[1]{\left| #1 \right|}
\newcommand{\norm}[1]{\left\lVert #1 \right\rVert}
\newcommand{\covering}[1]{N \left( #1 \right)}
\newcommand{\bracketing}[1]{N_{[\, ]} \left( #1 \right)}

\NewDocumentCommand{\E}{o g}{%
\IfNoValueTF{#2}{\operatorname{\mathbb{E}}\IfValueT{#1}{_{#1}}}
{\,{\operatorname{\mathbb{E}}\IfValueT{#1}{_{#1}}}\left[{#2}\right]}}

\NewDocumentCommand{\Var}{o m o}{%
  \IfNoValueTF{#1}{%
    \IfNoValueTF{#3}{%
      \operatorname{Var}\left[#2\right]%
    }{%
      \operatorname{Var}\left[#2 \mid #3\right]%
    }%
  }{%
    \IfNoValueTF{#3}{%
      \operatorname{Var}_{#1}\left[#2\right]%
    }{%
      \operatorname{Var}_{#1}\left[#2 \mid #3\right]%
    }%
  }%
}

\NewDocumentCommand{\Cov}{o g}{%
\IfNoValueTF{#2}{\operatorname{Cov}\IfValueT{#1}{_{#1}}}
{\,{\operatorname{Cov}\IfValueT{#1}{_{#1}}}\left[{#2}\right]}}

\usepackage{natbib}

%% file: abstract.tex
\begin{abstract}
    As predictive algorithms grow in popularity, using the same dataset to both train and test a new model has become routine across research, policy, and industry. 
    Sample-splitting attains valid inference on model properties by using separate subsamples to estimate the model and to evaluate it. 
    However, this approach has two drawbacks, since each task uses only part of the data, and different splits can lead to widely different estimates. 
    Averaging across multiple splits, I develop an inference approach that uses more data for training, uses the entire sample for testing, and improves reproducibility. 
    I address the statistical dependence from reusing observations across splits by proving a new central limit theorem for a large class of split-sample estimators under arguably mild and general conditions. 
    Importantly, I make no restrictions on model complexity or convergence rates. 
    I show that confidence intervals based on the normal approximation are valid for many applications, but may undercover in important cases of interest, such as comparing the performance between two models. 
    I develop a new inference approach for such cases, explicitly accounting for the dependence across splits. 
    Moreover, I provide a measure of reproducibility for p-values obtained from split-sample estimators. 
    Finally, I apply my results to two important problems in development and public economics: predicting poverty and learning heterogeneous treatment effects in randomized experiments. 
    I show that my inference approach with repeated cross-fitting achieves better power than existing alternatives, often enough to reveal statistical significance that would otherwise be missed. 
\end{abstract}

%% file: intro.tex
\section{Introduction} \label{section.intro}

As predictive algorithms transform empirical economics, policy, and industry, it is now routine to use the same dataset to train and evaluate a new model. 
For example, when training a machine learning algorithm to predict treatment effects, create a targeted policy rule, or automate consumer credit scoring, it is essential to evaluate the quality of the predictions and assess whether implementation would generate disparate impact across demographic groups. 
However, standard methods for training and evaluating typically waste part of the data by splitting the sample into training and testing sets. 
I develop a new inference approach that uses the entire sample for both tasks by combining multiple splits, improving statistical power and reproducibility. 
I provide valid confidence intervals under weak conditions for model properties such as accuracy and fairness calculated using cross-fitting, repeated sample-splitting or repeated cross-fitting.

Specifically, I study a setting in which an analyst (a researcher, policymaker, or industry practitioner) wishes to use the same dataset to both: 
\begin{enumerate}[label=(\roman*)]
    \item train a new model, and 
    \item evaluate some of its properties, such as a measure of accuracy or fairness. 
\end{enumerate}
For example, consider a government using machine learning (ML) to target recipients of a poverty alleviation program. 
Step (i) consists of training a model to predict, for example, families at higher risk of falling below the poverty line, while step (ii) consists of constructing a confidence interval for the out-of-sample mean squared error or rate of correct classifications. 

Using the same observations for both steps (i) and (ii) creates a form of statistical dependence that makes inference challenging. 
For example, standard central limit theorems (CLTs) assume independence, which is violated in this setting. 
This difficulty is often overcome by randomly splitting the sample into two, one part to train the model (training sample), and the other to evaluate its properties (evaluation sample). 
Since each task is conducted with separate data, such statistical dependence is not generated, and one can use standard approaches to inference. 
This procedure, however, has three drawbacks: it uses only part of the data for training the model, only part of the data for evaluating its properties, and different random splits can lead to widely different estimates and potentially affect statistical significance. 

My main contribution is an inference approach that averages estimates across multiple sample splits, improving upon a standard 50/50 split by using more data for training, using twice as much data for evaluation, and improving reproducibility. 
In empirical applications and Monte Carlo experiments, I show that these improvements often reveal statistical significance that would otherwise be missed. 
The main challenge of using multiple splits is a new form of statistical dependence due to reusing observations in both training and evaluation roles across different splits. 
I address this challenge by proving a new CLT for a large class of split-sample estimators under weak conditions. 
My CLT builds on previous literature in two key dimensions: (i) it applies to a large class of estimators and split-sample procedures, and (ii) it imposes no restrictions on model complexity or rates of convergence and stability, only requiring that the estimated model converges to an arbitrary limit at any rate. 
This generality is crucial for accommodating popular ML algorithms, such as random forests or neural networks. 
Moreover, I characterize when confidence intervals based on the normal approximation are valid, showing they may fail in important cases such as comparing the performance of two models or learning features of heterogeneous treatment effects. 
I develop a new inference approach for such cases that explicitly accounts for the dependence across splits, leveraging my CLT.
Finally, I develop a reproducibility measure for p-values from split-sample procedures, quantifying whether the number of repetitions is sufficiently large to ensure reproducible inference.

To illustrate the technical challenges and empirical implications of my results, consider the problem of predicting poverty as a simple running example, which is one of my applications. 
Accurate out-of-sample poverty prediction is central to Development Economics for understanding poverty dynamics and designing targeted interventions. 
I focus on assessing predictive accuracy as the natural starting point, though my framework applies more broadly. 
Consider a sample $D = (Y_i, X_i)_{i=1}^n$ of $n$ households, where $X_i$ are covariates and $Y_i$ is a binary indicator equal to one if household $i$ is below the poverty line 13 years after the covariates were measured. 
The goal is to use the sample to (i) train a model $\etahat(x)$ to predict poverty by estimating $P(Y = 1 | X = x)$, for example using a machine learning algorithm, and (ii) evaluate its accuracy, for example by estimating and calculating a confidence interval (CI) for the out-of-sample mean squared error (MSE) 
$$\theta_\etahat = \E{\lrp{Y_{new} - \etahat(X_{new})}^2 | D} = \int \lrp{y - \etahat(x)}^2 dP(y,x),$$
where $(Y_{new},X_{new})$ is an out-of-sample observation drawn from the same distribution as the sample. 
Note that $\theta_\etahat$ is data-dependent, and is thus different from targeting a parameter $\theta_\etanot$ for some fixed $\etanot$. 
In the policy prediction example, the researcher is not interested in the out-of-sample accuracy of an ideal but unknown model $\etanot$. 
Instead, they are interested in the accuracy of the actually estimated model $\etahat$.

In this context, confidence intervals are often constructed using sample-splitting. 
If the entire sample is used for both tasks, standard CLTs do not apply to the average 
$$\frac{1}{n} \sum_{i=1}^{n} (Y_i - \etahat(X_i))^2,$$
since the summands are not independent. 
For example, $Y_1 - \etahat(X_1)$ and $Y_2 - \etahat(X_2)$ are dependent since $\etahat$ is estimated with both $(Y_1, X_1)$ and $(Y_2, X_2)$. 
A standard approach to handle this dependence is to impose complexity restrictions on how $\etahat$ is estimated, such as Donsker conditions. 
These restrictions hold for simple procedures like ordinary least squares, but fail for complex machine learning algorithms frequently used in applied problems \citep{chernozhukov2018double}. 
Sample-splitting avoids this dependence without strong assumptions: randomly split $\{1, \dots, n\}$ into sets $\s_1$ and $\s_2$ of size for example $n/2$, use data in $\s_1$ to estimate $\etahat_{1}$, and data in $\s_2$ to calculate the average 
\begin{equation} \label{eq.setup_ss}
    \thetahat_{\etahat_{1}} = \frac{1}{n/2} \sum_{i \in \s_2} (Y_i - \etahat_{1}(X_i))^2.
\end{equation}
Since the summands in $\thetahat_{\etahat_{1}}$ are independent conditional on $\s_1$, standard CLTs apply, and the normal approximation gives a valid CI for $\theta_{\etahat_1}$. 
However, this procedure uses only half of the data for each task, and different random splits can lead to widely different estimates and potentially different conclusions about statistical significance. 

Using multiple splits can improve upon these drawbacks but introduces a new challenge. 
Consider, for example, two-fold cross-fitting, where the roles of samples $\s_1$ and $\s_2$ are reversed and the final estimator averages the split-specific estimates. 
That is, estimate $\etahat_{1}$ using $\s_1$ and $\etahat_{2}$ using $\s_2$, then calculate the final estimator 
\begin{equation} \label{eq.intro_cf}
    \thetahat_\etahat = \frac{1}{n} \lrbk{\sum_{i \in \s_1} (Y_i - \etahat_{2}(X_i))^2 + \sum_{i \in \s_2} (Y_i - \etahat_{1}(X_i))^2},
\end{equation}
where $\etahat = (\etahat_1, \etahat_2)$. 
The estimand in this case is the MSE of an average model, as discussed in the next paragraph. 
While this estimator averages over all $n$ observations, standard CLTs do not apply due to a different form of statistical dependence: the first sum is not independent of the second since both use the entire dataset. 
My first main contribution is a central limit theorem for a large class of estimators that includes $\thetahat_\etahat$, which I use to construct valid CIs. 
In addition to using the entire sample for evaluation in \cref{eq.intro_cf}, which reduces the variance of the asymptotic distribution compared to that of \cref{eq.setup_ss}, more data can be used for training by increasing the number of folds. 
With 3 folds, for example, three models are trained, each using two-thirds of the data, with the remaining third used to evaluate the MSE. 
Finally, reproducibility is improved by repeating the splitting process multiple times and averaging the estimators over repetitions. 

I show that $\sqrt{n} (\thetahat_\etahat - \theta_\etahat)$ is asymptotically normal under weak conditions, targeting its out-of-sample expectation 
$$\theta_\etahat = \E{\frac{1}{2} \lrp{Y_{new} - \etahat_1(X_{new})}^2 + \frac{1}{2} \lrp{Y_{new} - \etahat_2(X_{new})}^2 | D}.$$
In the example above, $\theta_\etahat$ is mathematically equivalent to the MSE of the average model, that is, $\theta_\etahat = \theta_\etabar$, where 
$$\theta_\etabar = \E{\lrp{Y_{new} - \etabar(X_{new})}^2 | D}, \qquad \etabar(x) = \frac{1}{2} \etahat_1(x) + \frac{1}{2} \etahat_2(x).$$
This happens anytime the outcome is binary, and holds for the MSE, mean absolute deviation, among others, including when averaging over multiple folds and repetitions. 
In the poverty prediction example, this means that a researcher or policymaker can use model $\etabar$ for out-of-sample predictions, which will have MSE $\theta_\etahat$. 
For continuous outcomes, the researcher has two options. 
The first is to use a model $\etatilde(x)$ that predicts a value in $(\etahat_1(x), \etahat_2(x))$ at random. 
This model has an out-of-sample MSE equal to $\theta_\etahat$. 
Alternatively, one could still use $\etabar$, which has the guarantee to perform better or equal than $\etatilde$ in terms of out-of-sample accuracy due to a risk-contraction property (for details, see \Cref{appendix.etabar}). 

I make three main contributions. 
First, I prove a new central limit theorem for a large class of split-sample estimators under mild conditions. 
Specifically, I make no restrictions on the complexity of the models $\etahat$, or on their rates of convergence or algorithmic stability. 
For sample-average estimators, my CLT follows under a standard moments condition and assuming that $\etahat$ converges to an arbitrary limit, at any rate. 
I show that the normal approximation yields a valid CI in many applications, but may fail to do so in important cases of interest, such as comparing the performance between two models or some instances when $\etahat$ converges to a constant. 
My second contribution builds on the CLT to develop a new inference approach that covers such cases, explicitly accounting for the dependence across splits. 
I focus on the case of comparing the performance between two models, and discuss how the arguments apply more broadly to other cases. 
Finally, I develop a reproducibility measure for p-values obtained from split-sample estimators. 
It addresses a common concern: another researcher using the same dataset, but different splits, may reach a different conclusion about statistical significance. 
For a given (large) number of repetitions of sample-splitting/cross-fitting, my measure quantifies p-value reproducibility, assessing whether the number of repetitions is sufficiently large to ensure reproducible inference. 

Other contributions include a central limit theorem for split-sample empirical processes, which I use to prove my main central limit theorem, and may be of independent interest. 
I also apply this CLT to develop a new \textit{ensemble} method for learning features of heterogeneous treatment effects in randomized experiments, following the framework of \citet{chernozhukov2025generic}. 
The ensemble method improves on previous alternatives by using the entire sample for evaluation, more data for training, and combining multiple machine learning predictors, potentially improving power and avoiding issues of multiple hypothesis testing. 

I apply my inference approaches to two important problems in development and public economics: predicting poverty and learning heterogeneous treatment effects in randomized experiments. 
In the first application, using a panel from Ghana \citep{ghanapaneldataset} and Monte Carlo experiments, I show that repeated cross-fitting outperforms previous alternative approaches in detecting predictive power for being below the poverty line 13 years ahead. 
In the second application, I revisit the experiment of \citet{karlan2007does} on charitable giving and conduct Monte Carlo simulations, and show that my ensemble method achieves improved power for detecting heterogeneous treatment effects compared to previous alternatives.

The rest of the paper is structured as follows. 
\Cref{section.literature} summarizes related work, and \Cref{section.setup} establishes the setup and notation. 
\Cref{section.z_estimators} establishes a central limit theorem for split-sample Z-estimators, and \Cref{section.inference} develops inference using the normal approximation and for comparing two models. 
I introduce my measure of reproducibility in \Cref{section.reproducibility}. 
Finally, I implement my inference approaches in two empirical applications: predicting poverty in Ghana in \Cref{section.application_ghana} and heterogeneous treatment effects in charitable giving in \Cref{section.application_hte}. 
\Cref{section.conclusion} concludes. 
Proofs are delayed to \Cref{app.proofs}.

\subsection{Related Work} \label{section.literature}

I contribute to the literature on inference using multiple splits of the data. 
The literature on risk estimation via cross-validation provides related results establishing asymptotic normality for sample-average estimators based on multiple splits. 
Like my approach, these CLTs target data-dependent parameters, but rely on different types of assumptions and focus on the particular case of sample-averages. 
\citet{dudoit2005asymptotics} consider estimators that average over all possible splits or cross-splits of the sample, assume bounded loss function and require $\etahat$ to be loss consistent for a risk minimizing function, whereas I assume that $\etahat$ converges to an arbitrary limit. 
\citet{austern2020asymptotics} and \citet{bayle2020cross} provide CLTs under rate assumptions on the algorithmic stability of $\etahat$. 
\citet{bayle2020cross} provide two CLTs using estimators based on a single repetition of cross-fitting, one relying on rate conditions for algorithmic stability, and the second requires a ``conditional variance convergence'' assumption that they verify using rates for loss stability. 
My result does not require verifying loss stability conditions, which may not be satisfied in some high-dimensional settings \citep{bates2024cross}, and my result allows for any ML algorithm as long as $\etahat$ converges to an arbitrary limit at any rate, in the sense established in \Cref{section.z_estimators}. 
\citet{ledell2015computationally} provide a CLT for the particular case of estimating the area under the curve (AUC) measure via cross-validation, and \citet{andrews2022transfer} derive confidence intervals for the different but related problem of learning the transfer error of a model across domains. 
Moreover, I document that asymptotic normality of split-sample estimators does not immediately lead to valid inference for the important problem of comparing the performance between two models, and I construct a new inference approach for this case that explicitly incorporates the dependence across splits. 

A different class of related results show asymptotic normality using cross-fitting when targeting parameters that are not data-dependent. 
These approaches require stronger conditions on $\etahat$ that may not hold in general for nonparametric models with more than a handful of covariates, such as requiring $\etahat$ to converge in probability at some specified rate \citep{luedtke2016statistical,belloni2017program,chernozhukov2018double,benkeser2020improved,imai2025statistical}. 
Leveraging the data-dependent parameter of interest, my CLT (\Cref{th.clt_z}) requires no complexity restrictions, and assumes $\etahat$ converges in probability to any limit at any rate. 

In the context of learning features of heterogeneous treatment effects in randomized trials, \citet{chernozhukov2025generic} proposed taking the median of estimators, confidence intervals and p-values across splits, similarly focusing on a data-dependent parameter, without relying on complexity or rate assumptions. 
\citet{wager2024sequential} proposed a modified, sequential approach based on \citet{luedtke2016statistical}, and \citet{chernozhukov2025reply} suggested taking the median over repetitions of the sequential approach. 
In the same framework, \citet{imai2025statistical} developed inference using cross-fitting, relying on rate assumptions. 
My results build on this literature in four main dimensions, relying on the mild assumption that the trained models converge to any limit, at any rate. 
First, my estimator uses all observations on the role of evaluation sample, leading to a smaller variance of its asymptotic distribution. 
Second, my approach does not exhibit a tradeoff between training and evaluation sample sizes, allowing for more data to be used to train the models. 
Third, I provide inference for an interpretable estimand under no rate assumptions on the trained models, while \citet{chernozhukov2025generic} require a rate of concentration condition for coverage of their median estimand, which requires for example the training sample to be large relative to the evaluation sample. 
Finally, I introduce a new \textit{ensemble} method that combines predictions from multiple ML algorithms, potentially improving statistical power for detecting HTE, and avoiding issues of multiple hypothesis testing. 

The literature on learning features of heterogeneous treatment effects with multiple splits is a subset of a broader literature on aggregating potentially dependent p-values \citep{ruger1978maximale,ruschendorf1982random,meng1994posterior,meinshausen2009pvalues,gasparin2025combining}. 
These approaches similarly apply to data-dependent parameters under weak conditions, and typically target size control under the worst data generating process, thus being conservative in general. 
My confidence intervals are asymptotically exact and improve statistical power. 

Finally, my work is complementary to \citet{ritzwoller2023reproducible}. 
They provide a stopping algorithm for determining how many times to repeat sample-splitting to ensure reproducibility of averages over split-sample statistics, for example, the average point estimate. 
I take the number of repetitions as given and focus on inference, providing a measure of reproducibility for p-values calculated using multiple splits.
While \citet{ritzwoller2023reproducible} uses an asymptotic framework that takes the sample size fixed and assumes a small threshold for the variability of the average split-sample statistic, my framework uses a growing sample size and number of repetitions.

%% file: setup.tex
\section{Setup} \label{section.setup} 

I consider a setup in which an analyst (a researcher, policymaker, or industry practitioner) wishes to use a dataset to both (i) train a new model and (ii) evaluate some of its properties. 
This is typically the case when one wants to train a new model to automate or assist decision-making, for example using a machine learning algorithm. 
Since these algorithms, despite their potential, may perform poorly in practice or have disparate performance across groups, it is often important to assess their accuracy and fairness. 
I use the term fairness as in the algorithmic fairness literature \citep{chouldechova2018frontiers,cowgill2020algorithmic,barocas2016big} and provide example measures below.
I state the analyst's goals, discuss the parameter of interest with examples, and introduce the split-sample procedures I study. 

The first goal of the analyst is to train a model $\etahat \in H$ using an algorithm and a dataset $D = (W_i)_{i=1}^n$, where each $W_i \in \cW$ is an iid draw from a distribution $P$. 
I use \textit{train} to denote the fitting/estimation of $\etahat$ using $D$, \textit{training algorithm} (or just \textit{algorithm}) for the procedure that maps $D$ to $\etahat$, and \textit{estimated model} (or just \textit{model}) for the realized function $\etahat$. 
For example, one can use the Random Forests algorithm to train a new model $\etahat$. 
The sets $H$ and $\cW$ are in principle unconstrained, and $H$ depends on the choice of training algorithm. 
Typically, the analyst will estimate $\etahat$ by minimizing some loss function. 
My setup, however, is agnostic to the choice of training algorithm, and all results hold for any algorithm as long as $\etahat$ converges to an arbitrary limit at any rate, in the sense defined in \Cref{section.z_estimators}. 

The second goal of the analyst is to use $D$ to evaluate some performance property of $\etahat$, denoted $\theta_\etahat$. 
Specifically, the analyst wishes to construct a confidence interval $\CIalpha$ for $\theta_\etahat$ such that, for $\alpha \in (0,1)$,
\begin{equation} \label{eq.goal}
    \liminf_{n \to \infty} P \lrp{\thetaetahat \in \CIalpha} \ge 1 - \alpha,
\end{equation}
where the probability accounts for the randomness in both $\etahat$ and $\CIalpha$. 
$\theta_\etahat$ can be, for example, a measure of accuracy or fairness of $\etahat$. 
The parameter of interest, $\theta_\etahat$, depends on the data through the estimated model $\etahat$. 
This differs from the standard semiparametric literature, where the parameter of interest takes the form of $\theta_\etanot$ for some nuisance function $\etanot$. 
In the applications I consider, the object of interest is $\theta_\etahat$ since the analyst/policymaker is interested in the accuracy or fairness of the specific estimated model $\etahat$ they can actually implement. 
This is different from evaluating the performance of an ideal but unknown model $\etanot$. 

I provide three examples of such parameters of interest, and then discuss related cases in the literature where the parameter of interest is data-dependent. 

\begin{example}[Mean Squared Error] \label{example.mse}
    The individual observations are $W = (Y,X)$, where $Y \in \R$ is an outcome and $X \in \cX \subseteq \R^d$ is a set of covariates with $d \ge 1$. 
    $\etahat : \cX \to \R$ is a function that predicts $Y$ from $X$. 
    In the poverty prediction example discussed in \Cref{section.intro} and developed in \Cref{section.application_ghana}, $Y$ is a binary indicator for whether a household is below the poverty line, and $\etahat(x)$ is an estimate of $P(Y=1|X=x)$. 
    The mean squared error (MSE) of model $\etahat$ is 
    $$\theta_\etahat = \int \lrp{y - \etahat(x)}^2 dP(y,x).$$

    A related estimand, also covered by my framework, is the difference in MSE between two groups, which is a measure of fairness \citep{auerbach2024testing,liang2024algorithm,liu2024inference}. 
    Let $W = (Y,X,G)$, where $G \in \lrbc{a,b}$ indicates group membership (e.g., racial groups). 
    Then, 
    \begin{equation} \label{eq.mse_fairness}
        \theta'_\etahat = \int \lrp{y - \etahat(x)}^2 dP_{Y,X|G=a}(y,x) - \int \lrp{y - \etahat(x)}^2 dP_{Y,X|G=b}(y,x)
    \end{equation}
    quantifies how much better $\etahat$ performs for one group relative to the other, where $P_{Y,X|G=g}$ is the conditional distribution of $(Y,X)$ given $G=g$. 
\end{example}

\begin{example}[Classification Rate - Binary Classifiers] \label{example.classif_binary}
    The individual observations are $W = (Y,X)$, where $Y$ is a binary outcome and $X \in \cX \subseteq \R^d$ is a set of covariates, for some $d \ge 1$. 
    $\etahat : \cX \to \lrbc{0,1}$ is a function that predicts whether $Y=1$ or $Y=0$. 
    The correct classification rate of model $\etahat$ is 
    $$\theta_\etahat = \int \I{y = \etahat(x)} dP(y,x).$$
    $\theta_\etahat$ is a measure of accuracy, corresponding to the probability that $\etahat$ classifies an observation correctly. 
    
    Similar to \cref{eq.mse_fairness}, the difference in classification rate between two groups is a measure of fairness. 
\end{example}

\begin{example}[Classification Rate - Probabilistic Classifiers] \label{example.classif_prob}
    The previous example can be generalized to accommodate probabilistic classifiers $\etahat : \cX \to \lrbk{0,1}$, with $\etahat(X)$ being the estimated probability that $Y = 1$ given $X$. 
    The correct classification rate is given by 
    $$\theta_\etahat = \int \lrbk{\etahat(x) \I{y = 1} + (1 - \etahat(x)) \I{y = 0}} dP(y,x).$$
    This is equivalent to the probability (taking $\etahat$ fixed) that $a_\etahat(X)=Y$, where $a_\etahat(X) = 1$ with probability $\etahat(X)$, independent of $D$. 
    A measure of fairness can be defined similar to \cref{eq.mse_fairness}. 
\end{example}

There are several examples in the literature where the parameter of interest takes the form of a data-dependent $\theta_\etahat$. 
This occurs anytime the hypothesis of interest is selected only after the data has been used \citep{dawid1994selection}. 
An important case is the approach of \citet{chernozhukov2025generic} to inference on features of heterogeneous effects in randomized trials, which I revisit in \Cref{section.application_hte}. 
Other examples include evaluating the impacts of data-driven algorithms in policy applications \citep{potash2015predictive,kuzmanovic2024causal}, measuring welfare gains generated from data-driven rules \citep{kitagawa2018should,ida2024dynamic}, and the ``inference on winners'' framework of \citet{andrews2024inference}. 

My setup also applies to some cases where the parameter of interest is not data dependent, but is estimated using split-sample techniques. 
For example, in \citet{fava2025predicting} I develop an approach to inference on points of the distribution of treatment effects. 
Although the parameter of interest, $\theta$, is not data dependent, I incorporate covariate-adjustment terms $\etahat$ that yield bounds $\theta_{\etahat,L} \le \theta \le \theta_{\etahat,U}$. 
Inference on $\theta$ can then be derived from the asymptotic distribution of split-sample estimators $(\thetahat_{\etahat,L},\thetahat_{\etahat,U})$, centered around the bounds $(\theta_{\etahat,L},\theta_{\etahat,U})$. 
Other examples where $\theta_\etahat$ is informative about a parameter $\theta$ include learning the mean outcome under an optimal treatment regime \citep{shi2020breaking, fischer2024bridging}, and averages of intersection bounds \citep{ji2024model,semenova2025debiased}. 
Another type of application is when $\theta = \theta_\etahat$ does not depend on $\etahat$, yet estimating $\theta_\etahat$ leads to some better properties. This is the case of adding a covariate-adjustment term for learning the average treatment effect in a randomized trial, as I discuss in Appendix \ref{appendix.application_rct}. 

I consider four split-sample procedures for attaining the analyst's goals: 1) sample-splitting, 2) cross-fitting, 3) repeated sample-splitting, and 4) repeated cross-fitting. 
First, I introduce some notation. 
Let $\bkn = \lrbc{1, \dots, n}$ and the dataset $D = ( W_i)_{i \in \bkn}$ be an iid sample of $W \sim P$. 
I denote the training algorithm by $\cA : \cW^m \to H$, a function that takes a sample of size $m$ and returns a value $\eta \in H$. 
The dependence on $m$ is suppressed in the notation of $\cA$. 
For any subsample $\s \subset \bkn$, $D_{\s} = \{ W_i \}_{i \in \s}$. 

Sample-splitting consists of taking a random subsample $\s \subseteq \bkn$ of size $b$, using its complement $\stilde = \bkn \setminus \s$ to train the model $\etahat_{\stilde} = \cA(D_{\stilde})$, and calculating $\CIalpha$ from $\s$ for the parameter $\theta_{\etahat_{\stilde}}$. 
Cross-fitting consists of partitioning $\bkn$ into $K$ roughly equal-sized subsets (\textit{folds}) $(\s_k)_{k=1}^K$, at random. 
For $k=1,\dots,K$, train a model $\etahat_{\stilde_k} = \cA(D_{\stilde_k})$, that is, using all observations except those in fold $k$. 
Each model $\etahat_{\stilde_k}$ is trained from $n (K-1) / K$ observations when $n$ is a multiple of $K$. 
In \Cref{section.z_estimators}, I discuss different ways to aggregate the $K$ models into an estimand $\theta_\etahat$, where $\etahat = (\etahat_{\stilde_k})_{k=1}^K$, and the construction of a confidence interval $\CIalpha$ for $\theta_\etahat$. 
I consider $K$ fixed as $n \to \infty$. 

Repeated sample-splitting and cross-fitting consist of repeating the procedures above $M$ times. 
That is, for repeated sample-splitting, take $M$ independent, random subsamples of $\bkn$ of size $b$, $(\s_{m,1})_{m=1}^M$, and train $M$ models $(\etahat_{\stilde_{m,1}})_{m=1}^M$. 
For repeated cross-fitting, take $M$ independent, random partitions of $\bkn$ into $K$ roughly equal-sized folds, $\cR = (r_m)_{m=1}^M$, where each $r_m = (\s_{m,k})_{k=1}^K$ forms a partition of $\bkn$. 
For each subsample $\s_{m,k}$, train a model $\etahat_{\stilde_{m,k}}$ using all observations except the ones in $\s_{m,k}$, giving a total of $M K$ models. 
I discuss different ways to aggregate the multiple splits in \Cref{section.z_estimators}. 

I give a unified notation to the four split-sample procedures described above. 
Let $\cR = \lrp{r_m}_{m=1}^M$ denote a collection of $M$ random splits of the sample, where each split can be either:
\begin{itemize}
\item Sample-splitting: $K=1$ and $r_m = (\s_{m,1})$ with $\s_{m,1} \subset \bkn$ of size $b$, or 
\item $K$-fold cross-fitting: $K>1$ and $r_m = (\s_{m,k})_{k=1}^K$ forms a partition of $\bkn$. 
\end{itemize}
I use $K=1$ to denote sample-splitting for convenience. 
$K=1$ means that $r_m$ consists of one subsample, of size $b < n$ chosen by the researcher. 
For cross-fitting, I assume folds are equal-sized if $n$ is a multiple of $K$, and have sizes $\lfloor n/K \rfloor$ and $\lceil n/K \rceil$ otherwise, and define $b = n/K$. 
Define $\pi = \lim_{n \to \infty} b/n$, $\pi \in (0,1)$. 
With this notation, $K=1$ denotes sample-splitting and $K>1$ denotes cross-fitting. 
I allow $M$ to grow as $n$ increases, and denote 
$$\Mbar = \lim_{n \to \infty} M \in \mathbb{N} \cup \lrbc{+\infty}.$$
This notation unifies the four split-sample procedures described previously, as shown in \Cref{table.procedures}.
\begin{table}[ht]
\centering
\caption{Classification of Split-Sample Procedures}
\label{table.procedures}
\begin{tabular}{c|cc}
\toprule
& \multicolumn{2}{c}{Number of folds ($K$)} \\
\makecell{Limit number of\\repetitions ($\Mbar$)} & 1 & $>1$ \\
\midrule
1 & Sample-splitting & Cross-fitting \\
$>1$ & Repeated sample-splitting & Repeated cross-fitting \\
\bottomrule
\end{tabular}
\end{table}
I use the term \textit{multiple splits} to denote any of the three procedures that use more than one split ($\Mbar > 1$ and/or $K>1$).  
In all cases, I assume that the splits are taken at random uniformly over all possible splits or cross-splits. 
Although the number of possible splits is finite for any given $n$, I consider that the $M$ repetitions are taken independently, with repetition. 
This assumption reflects common practice, as the computationally feasible number of repetitions is usually much smaller than the total number of possible splits, so that the probability of taking two identical splits is negligible. 

I compare the four split-sample procedures in terms of statistical power, modeling power, and reproducibility properties in \Cref{section.comparison_procedures}.

%% file: z_estimators.tex
\section{CLT for Split-Sample Z-Estimators} \label{section.z_estimators}

I prove a central limit theorem for split-sample Z-estimators, defined as zeroes of empirical moment equations. 
Z-estimators are a large class of estimators which include averages, linear regressions, and most M-estimators, since the parameter value that maximizes some objective function is the same that sets its partial derivatives to zero. 
This CLT can be used off-the-shelf in many applications, including the poverty prediction application in \Cref{section.application_ghana}. 
First, in \Cref{section.main_result_z}, I define split-sample Z-estimators and Z-estimands, introduce the assumptions used, and state the CLT. 
Finally, in \Cref{section.comparison_procedures}, I compare the four split-sample procedures (sample-splitting, cross-fitting, repeated sample-splitting, and cross-fitting). 

I provide a more accessible exposition for the particular case of sample average estimators, such as the MSE (\Cref{example.mse}), in \Cref{section.avg}. 
I prove a new CLT for split-sample empirical processes in \Cref{section.general}, which I use to prove my CLT for Z-estimators and may be of independent interest. 

\subsection{Main Result} \label{section.main_result_z}

Since Z-estimators can be nonlinear, unlike the mean squared error (\Cref{example.mse}), different approaches to aggregating multiple splits lead to different estimators and estimands. 
I discuss three such approaches. 
Let $\norm{\cdot}$ be the Euclidean norm, $\psi_{\theta,\eta}: \cW \to \R^d$ be measurable functions for $\theta \in \Theta \subseteq \R^d$ and $\eta \in H$ ($H$ is defined as in \Cref{section.setup}), and $d \ge 1$. 
For $\eta \in H$, let $\Psi_{\eta}(\theta) = P \psi_{\theta, \eta}$, $\Psihat_{\s,\eta}(\theta) = |\s|^{-1} \sum_{i \in \s} \psi_{\theta, \eta}(W_i)$, and $\Psidot_{\eta}$ be the Jacobian matrix of $\Psi_{\eta}(\theta_\eta)$. 
As in \Cref{section.setup}, let $\cR$ denote a collection of splits with $M$ repetitions and $K$ folds, and let $\etahat = \etahat_{\cR} = \lrp{\etahat_{\stilde_{m,k}}}_{m \in \lrbk{M}, k \in \lrbk{K}}$. 

The first type of estimand is an average across split-specific estimands:
\begin{equation} \label{eq.theta_z_1}
    \theta_{\etahat}^{(1)} = \frac{1}{M K} \sum_{r \in \cR} \sum_{\s \in r} \theta_{\etahat_{\stilde}}^{(1)},
\end{equation}
where $\theta_{\eta}^{(1)}$ for $\eta \in H$ is the unique solution for $\theta$ in $\Psi_{\eta}(\theta) = 0$, i.e., 
$$\Psi_{\eta}(\theta_{\eta}^{(1)}) = 0.$$
\cref{eq.theta_z_1} consists of solving the moment condition $\Psi_{\etahat_{\stilde}}(\theta) = 0$ for each split $\stilde$, and averaging over the split-specific estimands. 
The Z-estimator for \cref{eq.theta_z_1} is 
\begin{equation} \label{eq.thetahat_z_1}
    \thetahat_{\etahat}^{(1)} = \frac{1}{M K} \sum_{r \in \cR} \sum_{\s \in r} \thetahat_{\etahat_{\stilde}}^{(1)},
\end{equation}
where $\thetahat_{\etahat_{\stilde}}^{(1)} \in \arg\min_{\theta \in \Theta} \norm{\Psihat_{\s,\etahat_{\stilde}}(\theta)}$. 
This approach is analogous to the DML1 estimator in \citet{chernozhukov2018double}. 

The second type of estimand solves the average of the moment conditions. 
That is, $\theta_{\etahat}^{(2)}$ uniquely solves 
\begin{equation} \label{eq.theta_z_2}
    \frac{1}{M K} \sum_{r \in \cR} \sum_{\s \in r} \Psi_{\etahat_{\stilde}}\lrp{\theta_{\etahat}^{(2)}} = 0.
\end{equation}
The associated Z-estimator is given by 
\begin{equation} \label{eq.thetahat_z_2}
    \thetahat_{\etahat}^{(2)} \in \arg\min_{\theta \in \Theta} \norm{\frac{1}{M K} \sum_{r \in \cR} \sum_{\s \in r} \Psihat_{\s,\etahat_{\stilde}}(\theta)}.
\end{equation}
This approach is analogous to the DML2 estimator in \citet{chernozhukov2018double}. 

Finally, the third type of estimand is a hybrid of the previous two approaches. 
It solves the moment condition at each cross-split of the sample, and averages across repetitions. 
That is,
\begin{equation} \label{eq.theta_z_3}
    \theta_{\etahat}^{(3)} = \frac{1}{M} \sum_{r \in \cR} \theta_{\etahat_{r}}^{(2)},
\end{equation}
where $\theta_{\etahat_{r}}^{(2)}$ uniquely solves 
$$\frac{1}{K} \sum_{\s \in r} \Psi_{\etahat_{\stilde}}\lrp{\theta_{\etahat_{r}}^{(2)}} = 0.$$
The associated Z-estimator is given by 
\begin{equation} \label{eq.thetahat_z_3}
    \thetahat_{\etahat}^{(3)} = \frac{1}{M} \sum_{r \in \cR} \thetahat_{\etahat_{r}}^{(3)},
\end{equation}
where 
\begin{equation} \label{eq.thetahat_z_3.1}
    \thetahat_{\etahat_{r}}^{(3)} \in \arg\min_{\theta \in \Theta} \norm{\frac{1}{K} \sum_{\s \in r} \Psihat_{\s,\etahat_{\stilde}}(\theta)}.
\end{equation}
In this approach, each $\thetahat_{\etahat_{r}}^{(3)}$ uses the whole sample both for calculating $\etahat_{r}$ and the average in \cref{eq.thetahat_z_3.1}, and the final estimator $\thetahat_{\etahat}^{(3)}$ is the average of the cross-fitting estimators across repetitions. 
Note that $\theta_{\etahat}^{(1)} = \theta_{\etahat}^{(3)}$ if $K=1$, $\theta_{\etahat}^{(2)} = \theta_{\etahat}^{(3)}$ if $M = 1$, and $\theta_{\etahat}^{(1)} = \theta_{\etahat}^{(2)} = \theta_{\etahat}^{(3)}$ if $K=M=1$. 
The estimators are not assumed to be unique, but I assume the estimands and the limit of the estimators to be unique. 

For a concrete example, I consider below the particular case of sample-averages, as in the example of calculating the MSE for poverty prediction (\Cref{example.mse}). 

\begin{example}[Split-sample averages] \label{example.averages} \; \\
    Let $\psi_{\theta, \eta}(w) = f_\eta(w) - \theta$ for some known $f_\eta$.
    In this case, the three estimators coincide:
    $$\thetahat_{\etahat}^{(j)} = \frac{1}{M K} \sum_{r \in \cR} \sum_{\s \in r} \frac{1}{|\s|} f_{\etahat_{\stilde}}(W_i)$$
    for any $j$, and the estimand is 
    $$\theta_{\etahat}^{(j)} = \frac{1}{M K} \sum_{r \in \cR} \sum_{\s \in r} \int f_{\etahat_{\stilde}}(w) dP(w).$$ 
\end{example}

$\thetahat_{\etahat}^{(2)}$ can be interpreted as the value of $\theta$ that solves the moment condition for a randomized function that takes value across $\lrp{\etahat_{\stilde_{m,k}}}_{m \in \lrbk{M}, k \in \lrbk{K}}$ uniformly at random. 
That is, \cref{eq.theta_z_2} is equivalent to 
$$\int \psi_{\thetahat_{\etahat}^{(2)}, \etahat_{\xi}}(w) dP(w, \xi) = 0,$$
where $dP(w, \xi) = dP(w) dP(\xi)$ and $\xi$ takes value in $\lrp{\stilde_{m,k}}_{m \in \lrbk{M}, k \in \lrbk{K}}$ uniformly at random. 
If, for example, each $\etahat_{\stilde}$ is a probabilistic classifier as in \Cref{example.classif_prob}, $\thetahat_{\etahat}^{(2)}$ can be interpreted as solving the moment condition for a randomized rule $\etabar(x)$ that predicts a positive classification with probability $(MK)^{-1} \sum_{r \in \cR} \sum_{\s \in r} \etahat_{\stilde}(x)$. 

I provide a CLT for the three estimators $(\thetahat_{\etahat}^{(1)},\thetahat_{\etahat}^{(2)},\thetahat_{\etahat}^{(3)})$. 
Below, I establish my main regularity conditions.

\begin{assumption} \label{as.z_estimator}
    For some $\Theta' \subseteq \Theta$, the following conditions hold: 
    \begin{assumptionenum}
        \item \label{as.z_2plusdelta}
        For some $\delta > 0$, 
        $$\sup_{P \in \cP, \eta \in H, \theta \in \Theta'} \E[P]{\lrm{\psi_{\theta,\eta}(W)}^{2+\delta}} < \infty;$$
        \item \label{as.z_etahat} There exists $\etastarP \in H$ such that for $\etatilde = \cA(D)$, $W \perp D$, and every $\theta \in \Theta'$, 
        $$\lrm{\psi_{\theta, \etatilde}(W) - \psi_{\theta, \etastarP}(W)} \Pto 0$$ 
        uniformly in $P \in \cP$. 
    \end{assumptionenum}
\end{assumption}

\Cref{as.z_2plusdelta} is a standard moments condition for CLTs. 
\Cref{as.z_etahat} is a mild stability condition on $\etatilde$.
Importantly, $\etatilde$ is allowed to converge at any rate and to any limit $\etastar_P$, which may depend on $P$. 
It holds, for example, if 
$$\psi_{\theta, \etatilde}(w) \Pto \psi_{\theta, \etastarP}(w)$$
pointwise for every $w \in \cW$ and $\theta \in \Theta'$. 
This condition is more interpretable but stronger than required (see \Cref{as.etahat} in \Cref{section.general}). 
\Cref{as.z_etahat} differs from the typical approach in the double machine learning literature where faster convergence rates (often $n^{-1/4}$) are required for nuisance functions \citep[e.g.,][]{chernozhukov2018double}. 
The key difference between the two approaches is that I target a different, data-dependent parameter. 

My CLT relies on the additional technical regularity conditions \Cref{as.z_estimator_technical}, which I delay to \Cref{appendix.proofs_z_estimators}. 
This assumption adapts standard conditions for consistency and asymptotic normality of Z-estimators to the context of split-sample estimators \citep[e.g.,][]{van2000asymptotic,van2023weak}. 
This is a weak assumption that holds in many settings, and it mostly concerns the choice of $\psi_{\theta,\eta}$. 
First, it assumes that the classes $\cF_\eta = \lrbc{\psi_{\theta, \eta, j} : \theta \in \Theta'}$ are Donsker, which restricts complexity along $\theta \in \Theta'$ but does not restrict the complexity of $H$. 
Second, it requires $\thetahat_{\etahat}^{(j)}$ to nearly solve the moment conditions, and $\theta_{\etahat}^{(j)}$ to be unique and well-separated zeroes of the population moment conditions. 
Finally, it assumes that $\Psi_{\eta}$ is differentiable in $\theta$ for $\eta \in H$, and the Jacobian is continuous in $\eta$ around $\etastarP$. 
\Cref{as.z_estimator_technical} holds, for example, in the case of sample averages (\Cref{example.averages}), or the ``fraction in poverty by tercile'' estimator in the poverty prediction application in \Cref{section.application_ghana}. 

\Cref{th.clt_z} is the first main result of this paper. 

\begin{theorem} \label{th.clt_z}
    \hyperlink{proof.th.clt_z}{(CLT for split-sample Z-estimators)} \, \\
    Let Assumptions \ref{as.z_estimator} and \ref{as.z_estimator_technical} hold. 
    Then, for $j \in \lrbc{1,2,3}$, 
    $$\sqrt{n} \lrp{\thetahat_{\etahat}^{(j)} - \theta_{\etahat}^{(j)}} \leadsto \cN\lrp{0, V_{\etastarP}}$$
    uniformly in $P \in \cP$, where 
    $$V_{\etastarP} = V_{\Mbar, K} \Psidot^{-1}_{\etastarP} \lrp{P \psi_{\theta_{\etastarP},\etastarP} \psi^{T}_{\theta_{\etastarP},\etastarP}} \lrp{\Psidot^{-1}_{\etastarP}}^{T},$$ 
    and 
    $$V_{\Mbar, K} = \begin{cases}
    \Mbar^{-1} \lrp{\pi^{-1} + \Mbar - 1}, & \text{if } K = 1 \text{ and } \Mbar < \infty \\
    1, & \text{otherwise}. 
    \end{cases}$$
\end{theorem}

The limiting variance $V_{\etastarP}$ is the product of two terms, the scalar $V_{\Mbar, K}$ and a positive semidefinite matrix. 
The choice of split-sample procedure only affects $V_{\etastarP}$ through $V_{\Mbar, K}$, which acts as a variance-inflating term since $V_{\Mbar, K} \ge 1$. 
When using a single split ($K=1, \Mbar=1$), the asymptotic variance is inflated by $V_{\Mbar, K} = \pi^{-1}$, where $\pi$ is the fraction of the sample used to evaluate $\thetahat_{\etahat}^{(j)}$ (as opposed to training $\etahat$).
This occurs because $\thetahat_{\etahat}^{(j)}$ is calculated from only $b = \pi n$ observations.
When using repeated sample-splitting ($K = 1$ and $\Mbar > 1$), $V_{\Mbar, K} = \Mbar^{-1} \pi^{-1} + \Mbar^{-1} (\Mbar - 1)$ is an average of $\pi^{-1}$ and $1$ with weights proportional to $1$ and $\Mbar - 1$. 
This occurs since each observation is picked a different number of times across splits for calculating $\thetahat_{\etahat}^{(j)}$. 
A larger number of repetitions leads to more balance in how often each observation is selected, and $V_{\Mbar, K}$ decreases with larger $\Mbar$. 
In fact, if $\Mbar = \infty$, there is perfect balance in large samples and $V_{\Mbar, K} = 1$. 
When using cross-fitting ($K > 1$), all observations are used an equal amount of times, and $V_{\Mbar, K} = 1$. 
For intuition on this result, consider the particular case of sample averages (\Cref{example.averages}). 
In this case, 
$$\thetahat_{\etahat} = \frac{1}{M} \sum_{r \in \cR} \frac{1}{K} \sum_{\s \in r} \frac{1}{b} \sum_{i \in \s} f_{\etahat_{\stilde}}(W_i)$$
is the same for $j=1,2,3$. 
If $\Mbar=K=1$, $\thetahat_{\etahat}$ averages over $b = \pi n$ observations. 
If $\Mbar > 1$ and $K=1$, different observations are picked by splits $\s$ a different, random amount of times, and larger $\Mbar$ leads to more balance. 
If $K > 1$, $\frac{1}{K} \sum_{\s \in r} \frac{1}{b} \sum_{i \in \s} f_{\etahat_{\stilde}}(W_i)$ is an average over all observations, the entire sample is used equally, and the variance-inflation term is minimum. 
Hence, the asymptotic variance is minimized using cross-fitting with any number of folds $K > 1$ and repetitions $\Mbar$. 

\Cref{th.clt_z} appears to be new. 
The literature on risk estimation via cross-validation provides related results establishing asymptotic normality for sample average estimators based on multiple splits. 
Like my approach, these CLTs target data-dependent parameters, though they rely on different types of assumptions, and focus on the specific case of sample-averages. 
\citet{dudoit2005asymptotics} consider estimators that average over all possible splits or cross-splits of the sample, assume bounded loss function and requires $\etahat$ to be loss consistent for a risk minimizing function, whereas I assume $\etahat$ converges to any limit. 
\citep{austern2020asymptotics,bayle2020cross} provide CLTs under rate assumptions on the algorithmic stability of $\etahat$. 
\citet{bayle2020cross} provides two CLTs using estimators based on a single repetition of cross-fitting, one relying on rate condition for algorithmic stability, and the second requires a ``conditional variance convergence'' assumption that they verify using rates for loss stability. 
My result does not require verifying a loss stability condition, which may not be satisfied in some high-dimensional settings \citep{bates2024cross}, and my result allows for any ML algorithm as long as \Cref{as.avg_etahat} holds. 
\citet{ledell2015computationally} provides a CLT for the particular case of estimating the area under the curve (AUC) measure via cross-validation. 

A different class of related results are CLTs with cross-fitting for parameters that are not data-dependent. 
These approaches require stronger conditions on $\etahat$, such as requiring $\etahat$ to converge in probability at some specified rate, typically $n^{-1/4}$ \citep{luedtke2016statistical,chernozhukov2018double,benkeser2020improved,imai2025statistical}. 
\Cref{th.clt_z} requires no complexity restrictions, and assumes $\etahat$ converges in probability to any limit at any rate. 

A central limit theorem for the class of split-sample Z-estimators appears to be new. 
The characterization of the asymptotic variance, specifically how the variance-inflating term $V_{\Mbar, K}$ depends on the number of splits $\Mbar$ when $K=1$, also appears to be new. 
The proof uses a new CLT for split-sample empirical stated in \Cref{section.general}, which also appears to be new and may be of independent interest. 

\begin{remark}
    In the double machine learning context, which targets a different parameter $\theta_{\etanot}$ and uses $M=1$, simulation evidence \citep{chernozhukov2018double} and theoretical results \citep{velez2024asymptotic} suggest using DML2 over DML1.
    It is unclear whether similar arguments hold for comparing $\thetahat_{\etahat}^{(1)}$ and $\thetahat_{\etahat}^{(2)}$, and how they compare with $\thetahat_{\etahat}^{(3)}$. 
    Exploring theoretical and empirical properties of the three methods is an interesting direction for future research. 
\end{remark}

\subsection{Comparison of Split-Sample Procedures} \label{section.comparison_procedures}

I compare the four split-sample procedures (sample-splitting, cross-fitting, repeated sample-splitting, and repeated cross-fitting) in terms of statistical power, modeling power, reproducibility, and computation time. 

Cross-fitting and repeated cross-fitting, as well as repeated sample-splitting with $\Mbar = \infty$, all exhibit the highest statistical power since they all minimize the variance of the asymptotic distribution in \Cref{th.clt_z}. 
Repeated sample-splitting with $1 < \Mbar < \infty$ comes second, and sample-splitting yields the largest variance. 

I say that an estimator has better modeling power than another if the models in $\lrp{\etahat_{\stilde_{m,k}}}_{m \in \lrbk{M}, k \in \lrbk{K}}$ are trained using larger datasets. 
Using more data for training typically leads to models with smaller expected loss, as I formalize in \Cref{appendix.modeling_power}.
For sample-splitting or repeated sample-splitting, modeling power increases by picking a smaller $b$ (and $\pi$), so that more data is used to train each $\etahat_{\stilde_{m,k}}$. 
However, if $\Mbar < \infty$, a smaller $b$ leads to smaller statistical power, since fewer data are used as evaluation sample at each split. 
When using cross-fitting, modeling power increases with $K$, since $b=n/K$. 
In this case, the returns to increasing $K$ are diminishing. 
For example, if $K=2$, $\etahat_{\stilde_{m,k}}$ is calculated with $50$\% of the sample, and this fraction raises to $90$\% with $K=10$. 
If $K=20$, however, the fraction only raises by another $5$\%. 
Although a large value of $K$ or small value of $\pi$ (when $K=1$) lead to better modelling power, my asymptotic framework takes these quantities as fixed. 
This means that the quality of the asymptotic approximation may be poor if $K$ is large (or $\pi$ small) relative to the sample size. 
For example, my asymptotic framework does not accommodate for leave-one-out cross-fitting, that is, $K=n$. 

I formalize the fact that increasing $M$ leads to better reproducibility properties in \Cref{section.reproducibility}. 
For example, as $M$ increases, it becomes more likely that two researchers using the same dataset but different random splits will reach the same conclusion about statistical significance of $\theta_\etahat$.
Although I make no formal comparison between the cases $K=1$ and $K>1$ in terms of reproducibility, I note that \citet{ritzwoller2023reproducible} documented the difference in variance between repeated sample-splitting and cross-fitting in an earlier draft.\footnote{This discussion appears in the second version at \url{https://arxiv.org/abs/2311.14204} (dated December 9, 2023).} 
Comparing cross-fitting with $M$ repetitions to sample-splitting with $KM$ repetitions, they argued that in principle it is possible that split-sample estimators have smaller variance conditional on data when $K=1$ instead of $K>1$, but show empirical evidence that cross-fitting typically leads to better reproducibility. 

\Cref{table.comparison} summarizes the comparison of the four procedures. 
\begin{table}[ht]
\centering
\caption{Comparison of Split-Sample Procedures}
\label{table.comparison}
\begin{tabular}{lcccc}
\toprule
Procedure & \makecell{Statistical\\Power} & \makecell{Modeling\\Power} & Reproducibility & \makecell{Computation\\Time} \\
\midrule
Sample-splitting & Low & Low & Low & Low \\
Cross-fitting & High & High & Medium & Medium \\
Repeated sample-splitting & Med/High$^*$ & Med/High$^*$ & High$^{**}$ & High \\
Repeated cross-fitting & High & High & High$^{**}$ & High \\
\bottomrule
\end{tabular}
\begin{tablenotes}
\footnotesize
\item $^*$High if $\Mbar = \infty$, medium if $\Mbar < \infty$.
\item $^{**}$Whether repeated sample-splitting or cross-fitting dominates depends on application. 
\item Modeling power considers the trade-off with statistical power: 
for sample-splitting and repeated sample-splitting with $\Mbar < \infty$, increasing modeling power requires decreasing statistical power. 
Computation time and reproducibility columns compare repeated cross-fitting with $M$ repetitions to repeated sample-splitting with $K M$ repetitions. 
\end{tablenotes}
\end{table}

The choices of $M$, $K$, and $\pi$ (when $K=1$) involve tradeoffs. 
Statistical power is maximized when $K>1$ or $\Mbar = \infty$ (\Cref{section.avg}), and the reproducibility properties improve with larger $M$ and are ambiguously affected by $K$, despite empirical evidence that $K>1$ usually leads to better properties \citep{ritzwoller2023reproducible}. 
For $K=1$ and $\Mbar < \infty$, there is a tradeoff between statistical and modelling powers, unlike with cross-fitting. 
A larger $M$ is always beneficial in terms of reproducibility (and statistical power when $K=1$), but this comes at the cost of higher computation time. 
Hence, I recommend choosing $M$ as large as computationally convenient, and $K>1$ but small, since that provides valid asymptotic inference, maximum statistical power, and likely better reproducibility properties. 
In \Cref{section.reproducibility}, I provide a measure to assess whether a given $M$ is sufficiently large to ensure reproducibility of p-values calculated from split-sample Z-estimators.

%% file: inference.tex
\section{Inference on Split-Sample Estimands} \label{section.inference}

The CLT in \Cref{th.clt_z} can be directly applied to conduct inference on many split-sample estimands. 
However, confidence intervals based on the normal approximation may fail to cover $\theta_\etahat$ at the nominal level in some important cases of interest. 
First, in \Cref{section.inference_normal}, I consider inference when the normal approximation is asymptotically exact, and discuss why this approximation may not be precise in some contexts. 
Then, in \Cref{section.diff_performance}, I propose a new approach for the particular cases of inference on comparisons between models, which explicitly accounts for the dependence across splits. 

I discuss in \Cref{section.inference_normal} that a typical case when the normal approximation CI may have coverage probability smaller than nominal is when the variance of a moment function is allowed to be zero in the limit. 
I provide a general method for inference that covers this case in \Cref{section.fast_convergence}, by exploring the faster-than-$\sqrt{n}$ convergence rate of the empirical moment functions and introducing a tuning parameter. 
I also discuss in \Cref{section.inference_normal,section.diff_performance} that although \Cref{section.diff_performance} considers the specific case of comparing two models, the arguments developed in that section apply more broadly, covering other cases such as the Generic ML approach of \citet{chernozhukov2025generic} (see \Cref{appendix.test_hte}).

%% file: inference_normal.tex
\subsection{Inference from Normal Approximation} \label{section.inference_normal}

Consider the problem of conducting inference on $h(\theta_\etahat)$, where $\theta_\etahat$ is any of the split-sample Z-estimands in \Cref{section.z_estimators}, and $h$ is any scalar differentiable function with row-vector of partial derivatives $\hdot(\theta_{\etastarP}) \neq 0$. 
This encompasses many cases of interest, for example when $h(\theta_\etahat)$ is a subset of the vector $\theta_\etahat$ or a linear combination of its entries, as in the application of \Cref{section.application_ghana}. 
An application of \Cref{th.clt_z} and the delta-method yields 
\begin{equation} \label{eq.clt_h}
    \sqrt{n} \lrp{h(\thetahat_\etahat) - h(\theta_\etahat)} \leadsto \cN(0, \sigma^2_{\etastarP}),
\end{equation}
where $\thetahat_\etahat$ is a Z-estimator as in \cref{eq.thetahat_z_2}, and  
$$\sigma^2_{\etastarP} = \hdot(\theta_{\etastarP}) V^*_{\etastarP} \hdot(\theta_{\etastarP})^T.$$
If $\sigma^2_{\etastarP} > 0$, one can calculate the plug-in estimator 
\begin{equation} \label{eq.sigmahat}
    \sigmahat^2_{\etahat} = \hdot(\thetahatetahat) \hat{V}_{\etahat} \hdot(\thetahatetahat)^{T}, 
\end{equation}
where $\hat{V}_{\etahat}$ is given in \cref{eq.Vhat_etahat} in \Cref{section.proofs_section.diff_performance}, and the confidence interval  
\begin{equation} \label{eq.ci_h}
    \lrbk{h(\thetahat_\etahat) - n^{-1/2} z_{1-\alpha/2} \sigmahat_{\etahat}, h(\thetahat_\etahat) + n^{-1/2} z_{1-\alpha/2} \sigmahat_{\etahat}}
\end{equation}
contains $h(\theta_\etahat)$ with probability approaching $1 - \alpha$, where $z_\alpha$ is the $\alpha$-th quantile of the standard normal distribution. 

\begin{theorem} \label{th.clt_h}
    \hyperlink{proof.th.clt_h}{(Asymptotic Exactness of Normal Approximation CI)} \, \\
    Let the conditions of \Cref{th.clt_z} hold, $V_{\etastarP}$ be positive definite, assume there exists an estimator $\widehat{\dot{\Psi}}_{\etahat}$ such that 
    $$\norm{\widehat{\dot{\Psi}}_{\etahat} - \dot{\Psi}_{\etastarP}} \Pto 0$$
    uniformly in $P \in \cP$, and that $\inf_{P \in \cP} \norm{\hdot(\theta_{\etastarP})} > 0$. 
    Then, for any sequence $(P_n)_{n \ge 1} \subseteq \cP$, 
    $$P_n \lrp{h(\theta_\etahat) \in \lrbk{h(\thetahat_\etahat) - n^{-1/2} z_{1-\alpha/2} \sigmahat_{\etahat}, h(\thetahat_\etahat) + n^{-1/2} z_{1-\alpha/2} \sigmahat_{\etahat}}} \to 1 - \alpha.$$
\end{theorem}

\Cref{th.clt_h} assumes the existence of a consistent estimator of $\Psidot_{\etastarP}$. 
If $\psi_{\theta, \eta}(w)$ is differentiable in $\theta$, this assumption is satisfied by the plug-in estimator defined in \cref{eq.Psi_dot_hat} in \Cref{section.proofs_section.diff_performance} under a uniform integrability condition on this derivative. 
Otherwise, consistent estimators of $\Psidot_{\etastarP}$ can typically be constructed on a case-by-case basis \citep{hansen2022econometrics}. 
Note that the probability in \Cref{th.clt_h} is taken over both the random estimand $h(\theta_\etahat)$ and the CI. 

\Cref{th.clt_h} implies that \cref{eq.ci_h} contains $h(\theta_\etahat)$ with probability approaching $1 - \alpha$ in many settings. 
However, in some cases, \cref{eq.ci_h} may not cover $h(\theta_\etahat)$ with nominal probability, as illustrated in the two examples below. 

\begin{example} \label{ex.cov_prediction}
    Consider a dataset with covariates $X$, outcome $Y \in \R$, and moment function $\psi_{\theta, \eta}(y, x) = y \eta(x) - \theta$, so 
    $$\thetahat_\etahat = \frac{1}{M K} \sum_{r \in \cR} \frac{1}{n} \sum_{\s \in r} \sum_{i \in \s} Y_i \etahat_{\stilde}(X_i).$$ 
    The limit variance in \cref{eq.clt_h} is 
    $$\sigma^2_{\etastarP} = \Var[P]{Y \etastarP(X)}.$$ 
    If $\sigma^2_{\etastarP} > 0$, \cref{eq.ci_h} contains $h(\theta_\etahat)$ with probability approaching $1 - \alpha$. 
    However, if $\etastarP(x) = 0$ for all $x$, $\sigma^2_{\etastarP} = 0$, 
    $$\sqrt{n} \lrp{h(\thetahat_\etahat) - h(\theta_\etahat)} \Pto 0,$$
    and \cref{eq.ci_h} may not contain $h(\theta_\etahat)$ with nominal probability. 
\end{example}

\begin{example} \label{ex.cov_reg}
    Consider a dataset with covariates $X$, outcome $Y \in \R$, and moment function 
    $$\psi_{\theta, \eta}(y, x) = \begin{pmatrix}
    y - \theta_0 - \theta_1 \eta(x) \\
    (y - \theta_0 - \theta_1 \eta(x)) \eta(x)
    \end{pmatrix},$$
    that is, for each subsample $\s$, $\thetahat_{\stilde}$ is the OLS estimator for $(\theta_{0,\stilde}, \theta_{1,\stilde})$ in the regression 
    $$Y_i = \theta_{0,\stilde} + \theta_{1,\stilde} \etahat_{\stilde}(X_i) + \varepsilon_i$$
    using observations $i \in \s$. 
    Focusing on the slope coefficient, the final estimator can be, for example, 
    $$\thetahat_{1,\etahat}^{(1)} = \frac{1}{M K} \sum_{r \in \cR} \sum_{\s \in r} \thetahat_{1,\stilde}.$$ 
    If $E[Z^T Z]$ is positive definite, where $Z = (1, \etastarP(X))^T$, the conditions of \Cref{th.clt_z} are met, and \cref{eq.ci_h} contains $\theta_{1,\etahat}^{(1)}$ with probability approaching $1 - \alpha$. 
    However, if $\etastarP(x)$ is constant in $x$, $E[Z^T Z]$ is not invertible, the conditions of \Cref{th.clt_z} are not met, and \cref{eq.ci_h} may contain $h(\theta_\etahat)$ with probability below the nominal level. 
\end{example}

\Cref{ex.cov_prediction,ex.cov_reg} have two features in common: the normal approximation CI may undercover $\theta_\etahat$ only when $\etastarP(x)$ is constant, and one of the empirical moment equations evaluated at the true parameter converges to zero at a rate faster than $n^{-1/2}$: 
\begin{equation} \label{eq.fast_moment}
    \min_{j \in \{1, \dots, d \}} \lrm{\frac{1}{\sqrt{n}} \sum_{i \in \s} \psi_{\theta_{\etahat_{\stilde}},\etahat_{\stilde},j}(W_i)} \Pto 0,
\end{equation}
where $\psi_{\theta,\eta,j}$ is the $j$-th entry of the vector $\psi_{\theta,\eta}$. 
In \Cref{section.diff_performance}, I develop an approach that can be used to test whether $\etastarP(x)$ is constant, and although I focus on the particular case of comparing the performance between two models, the arguments apply more broadly and could be used to provide a valid CI for the problems in \Cref{ex.cov_prediction,ex.cov_reg} under the same conditions of \Cref{th.clt_h}. 
In \Cref{section.fast_convergence}, I establish a general approach to inference on $\theta_\etahat$ that allows \cref{eq.fast_moment} to happen. 
The approach explores the faster-than-$\sqrt{n}$ convergence rate to provide an asymptotically uniformly valid CI by introducing a tuning parameter.

%% file: inference_diff_performance.tex
\subsection{Inference on Model Comparisons} \label{section.diff_performance}

In several applications, the goal is not only to create a new model $\etahat$ and assess some property $\theta_\etahat$, but to compare such properties between two models. 
For example, if $\theta_\etahat$ is a measure of accuracy such as the mean squared error (\Cref{example.mse}), one might want to infer if $\etahat$ has better performance than a baseline model that predicts the sample mean of $Y$ for all observations. 
This is the case in the application of \Cref{section.application_ghana}, where the goal is to assess whether a random forest model has predictive power for poverty, that is, whether it achieves smaller MSE than using the sample average. 
Alternatively, one might want to compare the performance of using different machine learning algorithms, such as training $\etahat$ with neural networks versus random forests. 
I show that the CLTs of the previous sections give a valid inference approach when both models do not have similar performances in large samples. 
However, if the models have similar performance, the asymptotic distribution of the difference in performance is degenerate at the $\sqrt{n}$ rate, and CIs based on the asymptotic approximation may fail to control size. 
In this section, I build on the CLT of \Cref{section.z_estimators} to develop an inference method that is valid for both cases. 
Although this section focuses on the particular case of comparing two models, I discuss in the end of \Cref{section.inference_normal} that the arguments developed in this section apply more broadly. 

The setting is as follows. 
$\thetahat_\etahat$ denotes any of the estimators $(\thetahat_\etahat^{(1)},\thetahat_\etahat^{(2)},\thetahat_\etahat^{(3)})$ of \Cref{section.z_estimators}, assumed to be a scalar ($d=1$) (alternatively, one could consider a scalar transformation $h(\thetahat_\etahat)$ as in \Cref{section.inference_normal}). 
I refer to the parameter $\theta_\etahat$ (defined analogously) as a \textit{performance} measure for expositional convenience, though the results apply more generally. 
I focus on comparing $\theta_\etahat$ to the performance $\theta_{\bhat}$ of a baseline model $\bhat \in H$ computed using the entire sample, that is, without forms of sample-splitting. 
$\bhat$ is assumed to come from a parametric model, and it can be, for example, the sample average $\bhat(x) = n^{-1} \sum_{i=1}^{n} Y_i$ in \Cref{example.classif_prob,example.mse}. 
Following the notation of \Cref{section.z_estimators}, $\theta_{\bhat}$ is the unique solution for $\theta$ in $\Psi_{\bhat}(\theta) = 0$, i.e., 
$$\Psi_{\bhat}(\theta_{\bhat}) = 0.$$
Similarly, the estimator $\thetahat_{\bhat}$ is a (near) zero of the empirical moment condition, 
$$\thetahat_{\bhat} \in \arg\min_{\theta \in \Theta} \norm{\frac{1}{n} \sum_{i=1}^{n} \psi_{\theta,\bhat}(W_i)}.$$
In \Cref{appendix.comparison}, I discuss how to extend the current setting for comparing $\theta_\etahat$ to the performance of another model $\etahat'$ computed with the same split-sample approach as $\etahat$. 
Let 
$$\cS = \lrp{\s_{m,k}}_{m \in \lrbk{M}, k \in \lrbk{K}}$$
be a collection of $M K$ splits of the sample, that is, a vectorization of $\cR$ defined in \Cref{section.setup}. 
Notice that each $\s \in \cS$ is associated with a model $\etahat_{\stilde}$, as in \cref{eq.theta_z_1}. 

To see the challenge of conducting inference based on $\thetahat_\etahat - \thetahat_{\bhat}$, consider a simplified setting where each $\thetahat_{\etahat_{\stilde}}$ (as in \cref{eq.thetahat_z_1}) is a sample average, that is, $\psi_{\theta,\eta}(w) = f_\eta(w) - \theta$ for some $f_\eta$ and $\thetahat_{\etahat_{\stilde}} = \lrm{\s}^{-1} \sum_{i \in \s} f_{\etahat_{\stilde}}(W_i)$. 
The CLT in \Cref{th.clt_z} gives 
$$\sqrt{n} \lrp{\thetahat_\etahat - \theta_\etahat} = \frac{1}{\sqrt{n}} \sum_{i=1}^{n} \lrp{f_{\etastarP}(W_i) - Pf_{\etastarP}} + o_P(1),$$
and the normal approximation gives an asymptotically valid CI for $\theta_\etahat$. 
Similarly, if $\bhat$ converges to some model $b_P \in H$, 
$$\sqrt{n} \lrp{\thetahat_{\bhat} - \theta_{\bhat}} = \frac{1}{\sqrt{n}} \sum_{i=1}^{n} \lrp{f_{b_P}(W_i) - Pf_{b_P}} + o_P(1),$$
and these two results can be combined to construct a CI for $\theta_\etahat - \theta_{\bhat}$ based on a normal approximation.
However, if the baseline model $b_P$ is the same as $\etastarP$, 
both estimators have the same limit, the difference 
\begin{equation} \label{eq.diff_degenerate}
    \sqrt{n} \lrbk{\lrp{\thetahat_{\etahat} - \thetahat_{\bhat}} - \lrp{\theta_{\etahat} - \theta_{\bhat}}} = o_P(1)
\end{equation}
has a degenerate limit in probability, and the CLT of \Cref{section.z_estimators} does not inform how to compute a CI for $\theta_{\etahat} - \theta_{\bhat}$. 

First, I develop a test for whether any of the models $\etahat_{\stilde}$ perform better than $\bhat$, then show how this test can be used to construct a CI for $\theta_{\etahat} - \theta_{\bhat}$. 
Both results build on my CLT for Z-estimators. 

\subsubsection{A Multivariate One-sided Test for Model Differences} \label{section.diff_onesided_test}

From \cref{eq.diff_degenerate}, the asymptotic distribution of $\thetahat_{\etahat} - \thetahat_{\bhat}$ centered around the parameter of interest $\theta_{\etahat} - \theta_{\bhat}$ is degenerate at the $n^{-1/2}$ rate if $b_P = \etastarP$. 
Yet, for each split $\s \in \cS$, 
\begin{align}
    & \sqrt{n} \lrbk{\lrp{\thetahat_{\etahat_{\stilde}} - \thetahat_{\bhat}} - \lrp{\theta_{\etahat_{\stilde}} - \theta_{\bhat}}} \nonumber \\
    & = \frac{\sqrt{n}}{|\s|} \sum_{i \in \s} \lrp{f_{\etastarP}(W_i) - Pf_{\etastarP}} - \frac{\sqrt{n}}{n} \sum_{i=1}^{n} \lrp{f_{\etastarP}(W_i) - Pf_{\etastarP}} + o_P(1) \label{eq.diff_decomp}
\end{align}
has a non-degenerate limit since the first average does not include observations $i \in \stilde$. 
I explore this fact to construct a test of whether any model $\etahat_{\stilde}$ has better performance than $\bhat$, then develop a CI for $\theta_{\etahat} - \theta_{\bhat}$ in the following subsection. 

Consider the hypothesis test 
\begin{equation} \label{eq.diff_onesided_test}
    \begin{cases}
    H_{0,\etahat}: \theta_{\etahat_{\stilde}} - \theta_{\bhat} \ge 0 & \text{ for all } \s \in \cS, \\
    H_{A,\etahat}: \theta_{\etahat_{\stilde}} - \theta_{\bhat} < 0 & \text{ for some } \s \in \cS. 
\end{cases}
\end{equation}
If $\theta_{\eta}$ is a measure of performance such as the mean squared error, having $\theta_{\etahat_{\stilde}} - \theta_{\bhat} < 0$ means that $\etahat_{\stilde}$ performs better than $\bhat$. 
The hypotheses $H_{0,\etahat}$ and $H_{A,\etahat}$ depend on $\etahat$ due to the data-dependent parameter of interest $\theta_\etahat$. 
Testing such hypotheses is analogous to constructing a confidence interval for a data-dependent parameter as in \cref{eq.goal}. 
Let 
$$\delta_\etahat = \lrp{\theta_{\etahat_{\stilde}} - \theta_{\bhat}}_{\s \in \cS},$$
and similarly define 
$$\deltahat_\etahat = \lrp{\thetahat_{\etahat_{\stilde}} - \thetahat_{\bhat}}_{\s \in \cS}.$$
An application of \Cref{th.clt_z} gives 
$$\sqrt{n} \lrp{\deltahat_\etahat - \delta_\etahat} \leadsto \cN \lrp{0, \Sigma},$$
for some nonzero $\Sigma$ that can be consistently estimated with $\Sigmahat$ (see equation \ref{eq.diff_Sigmahat} in \Cref{section.proofs_section.diff_performance}). 
Since splits reuse observations, the off-diagonal terms of $\Sigma$ explicitly incorporate the dependence across splits. 

Denote by $\sigmahat^2_\s$ the entry of the main diagonal of $\Sigmahat$ associated with $\s \in \cS$, that is, with the term $\thetahat_{\etahat_{\stilde}} - \thetahat_{\bhat}$. 
I propose computing the test-statistic 
$$T(\deltahat_\etahat, n^{-1} \Sigmahat) = \sum_{\s \in \cS} \lrp{\min \lrbc{\sqrt{n} \frac{\thetahat_{\etahat_{\stilde}} - \thetahat_{\bhat}}{\sigmahat_{\s}}, 0}}^2.$$
This type of test statistic has been considered for example in \citet{chernozhukov2007estimation,romano2008inference,andrews2009validity,romano2010inference} in the context of moment inequalities. 
Critical values $\hat{c}_{1 - \alpha}$ can be computed via Monte Carlo: simulate $Z \sim \cN(0,\Sigmahat)$ and estimate $\hat{c}_{1 - \alpha}$ as the $1-\alpha$ quantile of $T(Z, n^{-1} \Sigmahat)$. 
I note that, alternatively, one could use the likelihood ratio test statistic. 

Asymptotic exactness of this test under the least favorable null follows from similar conditions to \Cref{th.clt_z}, established below. 

\begin{assumption} \label{as.diff_onesided_test} 
    Assumptions \ref{as.z_estimator} and \ref{as.diff_onesided_test_technical} hold with scalar $\thetahat_\etahat$ $(d=1)$. 
    Additionally,
    \begin{assumptionenum}
        \item \label{as.diff_onesided_test_pos_var} $V_{\etastarP} > 0$;
        \item \label{as.diff_onesided_test_bP} For some $b_P \in H$, 
        $$\sqrt{n} \lrp{\thetahat_{\bhat} - \theta_{\bhat}} - \sqrt{n} \lrp{\thetahat_{b_P} - \theta_{b_P}} \Pto 0$$
        and
        $$\sqrt{n} \lrp{\theta_{\bhat} - \theta_{b_P}} \Pto 0$$
        uniformly in $P \in \cP$. 
    \end{assumptionenum}
\end{assumption}

\Cref{as.diff_onesided_test_technical} consists of more technical conditions, which are delayed to the appendix for ease of exposition. 
For example, they extend the Z-estimator assumptions on $\etastarP$ to $b_P$. 
\Cref{as.diff_onesided_test_pos_var} requires the limiting variance of $\sqrt{n} (\thetahat_\etahat - \theta_\etahat)$ to be positive, and \Cref{as.diff_onesided_test_bP} defines the requirements on the baseline (parametric) estimator $\bhat$. 
It holds, for example, if $\bhat$ belongs to a Donsker class with probability approaching one, which typically happens for parametric models such as the sample average $\bhat(x) = n^{-1} \sum_{i=1}^{n} Y_i$. 

\begin{theorem} \label{th.diff_onesided_test}
    \hyperlink{proof.th.diff_onesided_test}{(Size control of multivariate one-sided test for model differences)} \, \\
    Let \Cref{as.diff_onesided_test} hold. 
    Then, for any $\bar{c}_2 > 0$, 
    $$\limsup_{n \to \infty} \sup_{P \in \lrbc{P \in \cP : P(\delta_\etahat \ge 0) > \bar{c}_2}} P \lrp{T(\deltahat_\etahat, n^{-1} \Sigmahat) > \hat{c}_{1 - \alpha} \Bigm| \delta_\etahat \ge 0} = \alpha.$$
    For any sequence $(P_n)_{n \ge 1} \subseteq \cP$ with $\lim_{n \to \infty} P_n(\delta_\etahat = 0) > 0$, 
    $$\lim_{n \to \infty} P_n \lrp{T(\deltahat_\etahat, n^{-1} \Sigmahat) > \hat{c}_{1 - \alpha} \Bigm| \delta_\etahat = 0} = \alpha.$$
\end{theorem}

\Cref{th.diff_onesided_test} appears to be new. 
It establishes size control: the probability of rejecting the null hypothesis, conditional on it being true, does not exceed $\alpha$ in large samples. 
Note that the probabilities in \Cref{th.diff_onesided_test} are not random objects, they integrate over the distribution of the data conditional on the events $\delta_\etahat \ge 0$ or $\delta_\etahat = 0$. 
Alternative approaches for testing across multiple splits of the sample typically aggregate p-values or confidence intervals computed separately for each split, without accounting for the dependence structure across splits \citep[see, e.g.,][]{chernozhukov2025generic,gasparin2025combining}. 
For example, \citet{chernozhukov2025generic} propose aggregating the median of p-values or CIs across splits. 
Because these methods do not incorporate the correlation across splits, they are conservative in most data-generating processes, as they guard against the worst-case dependence structure. 
In contrast, my approach explicitly accounts for the dependence across splits, which enables the test to achieve exactness under the least favorable null $\delta_\etahat = 0$ in a uniform sense across DGPs. 
The proof is made possible by the decomposition in \cref{eq.diff_decomp}, which follows from the new CLT in \Cref{section.z_estimators}. 

The result above requires the probability of the conditioning event to be bounded away from zero using the constant $\bar{c}_2 > 0$. 
This could lead to an apparent uniformity issue for sequences of DGPs $(P_n)_{n \ge 1}$ with $P_n(\delta_\etahat \ge 0) \to 0$, for example. 
For such sequences, the probability of rejecting the null conditional on the null being true could be greater than $\alpha$. 
This is not, however, an issue for empirical practice: for such sequences the probability of being under the null itself converges to zero. 
Incorrectly rejecting the null is not a concern when the probability of the null being true is zero.

\subsubsection{A Confidence Interval for the Average Performance}

I construct a new confidence interval for $\theta_\etahat - \theta_{\bhat}$ based on two insights from the previous subsections.
The first is that a CI based on the normal approximation using \Cref{th.clt_z} is asymptotically exact if $\theta_\etahat - \theta_{\bhat}$ converges in probability to a value different from zero, since in this case the terms in \cref{eq.diff_degenerate} do not cancel out.
The second insight is that the case $\theta_\etahat - \theta_{\bhat} \Pto 0$ is closely connected with the null hypothesis of the one-sided test developed in the previous subsection. 
Hence, my CI consists of using the normal approximation if the one-sided test is rejected, and an extended CI in case it is not. 

Define the normal approximation CI 
$$\widehat{\rm CI}_{\alpha,\cN} = \lrbk{\lrp{\thetahat_\etahat - \thetahat_{\bhat}} - z_{1 - \alpha / 2} \frac{\sigmahat_\deltahat}{\sqrt{n}}, \lrp{\thetahat_\etahat - \thetahat_{\bhat}} + z_{1 - \alpha / 2} \frac{\sigmahat_\deltahat}{\sqrt{n}}},$$
where $\sigmahat_\deltahat$ is a standard error for $\thetahat_\etahat - \thetahat_{\bhat}$ (see equation \ref{eq.diff_sigmahat_deltahat} in \Cref{section.proofs_section.diff_performance}), 
and an extended CI 
$$\widehat{\rm CI}_{\alpha,{\rm ext}} = \operatorname{Conv}\lrp{\widehat{\rm CI}_{\alpha,\cN} \cup \{0\}},$$
where $\operatorname{Conv}(\cdot)$ denotes the convex hull, that is, $\widehat{\rm CI}_{\alpha,{\rm ext}}$ has all the elements in $\widehat{\rm CI}_{\alpha,\cN}$, $0$, and all elements in between. 
The final CI is given by 
$$\CIalpha = \begin{cases}
    \widehat{\rm CI}_{\alpha,\cN}, \text{ if } T(\deltahat_\etahat, \Sigmahat) > \hat{c}_{1 - \alpha} \\
    \widehat{\rm CI}_{\alpha,{\rm ext}}, \text{ otherwise}. 
\end{cases}$$
$\CIalpha$ is based on a pre-test, using different inference approaches depending on whether the one-sided test is rejected or not. 
This construction is motivated by the following facts, which are formalized in \Cref{th.diff_pointwise,th.diff_ci1,th.diff_ci1_power}. 
If $\theta_\etahat - \theta_{\bhat}$ converges in probability to a negative value, $P(T(\deltahat_\etahat, \Sigmahat) > \hat{c}_{1 - \alpha}) \to 1$, and $\widehat{\rm CI}_{\alpha,\cN}$ is used, which is asymptotically exact. 
If $\theta_\etahat - \theta_{\bhat}$ converges in probability to a positive value, $P(T(\deltahat_\etahat, \Sigmahat) > \hat{c}_{1 - \alpha}) \to 0$, $\widehat{\rm CI}_{\alpha,\cN}$ is asymptotically exact but $\widehat{\rm CI}_{\alpha,{\rm ext}}$ is used, which is valid since it is wider than $\widehat{\rm CI}_{\alpha,\cN}$, although conservative. 
This asymmetric construction is a choice, which reflects the motivating problem of this section of learning whether the new model $\etahat$ performs better (instead of worse) than the baseline model $\bhat$. 
Finally, if $\theta_\etahat - \theta_{\bhat} \Pto 0$, intuitively $P(T(\deltahat_\etahat, \Sigmahat) > \hat{c}_{1 - \alpha})$ should be close to $\alpha$ given \Cref{th.diff_onesided_test}. 
If that happens, $\widehat{\rm CI}_{\alpha,{\rm ext}}$ covers $\theta_\etahat - \theta_{\bhat}$ with high probability since it includes $0$, the limit of $\theta_\etahat - \theta_{\bhat}$. 
However, this guarantee depends on additional conditions as I discuss next, since $P (\delta_\etahat \ge 0)$ may not converge to one even if $\delta_\etahat \Pto 0$. 

First, I show that $\CIalpha$ is valid pointwise in $P \in \cP$, assuming that if $\etastarP = b_P$, then the parametric model is well-specified in the sense that it minimizes the error $\theta_\eta$ in $\eta$, that is, $\theta_\eta \ge \theta_{b_P}$ for all $\eta \in H$. 
Then, I establish conditions under which $\CIalpha$ is valid uniformly in $P \in \cP$. 

\begin{theorem} \label{th.diff_pointwise}
    \hyperlink{proof.th.diff_pointwise}{(Pointwise Asymptotic Validity of $\CIalpha$)} \, \\
    Let \Cref{as.diff_onesided_test} hold. 
    Then, for any $P \in \cP$ such that either
    \begin{itemize}
        \item[(i)] $\theta_\etastarP \neq \theta_{b_P}$, or 
        \item[(ii)] $\theta_{b_P} \le \inf_{\eta \in H} \theta_\eta$, 
    \end{itemize} 
    $$\liminf_{n \to \infty} P \lrp{\lrp{\theta_\etahat - \theta_{\bhat}} \in \CIalpha} \ge 1 - \alpha.$$
\end{theorem}

Further, I show that $\CIalpha$ is asymptotically valid for most sequences of $\theta_{\etahat} - \theta_{\bhat}$, and discuss why it may fail for specific sequences.  
Then, I establish that the additional condition \Cref{as.diff_ci1} is sufficient for $\CIalpha$ to be asymptotically uniformly valid in $P \in \cP$. 
Later, I propose a modification to $\CIalpha$ that gives uniform validity under only \Cref{as.diff_onesided_test}. 

\begin{assumption} \label{as.diff_ci1} 
    For any sequence $(P_n)_{n \ge 1} \subseteq \cP$ such that $\theta_\etastarPn - \theta_{b_{P_n}} \to 0$, 
    $$\sqrt{n} \lrp{\theta_\etahat - \theta_{\bhat}} \Pnto 0.$$
\end{assumption}

\begin{theorem} \label{th.diff_ci1}
    \hyperlink{proof.th.diff_ci1}{(Uniform Asymptotic Validity of $\CIalpha$)} \, \\
    Let \Cref{as.diff_onesided_test} hold. 
    For any $\bar{c}_3 > 0$ and $\bar{c}_4 > 0$, define 
    $$\cP_{\bar{c}_3,\bar{c}_4} = \lrbc{P \in \cP: P \Bigl( \lrp{\theta_\etahat - \theta_{\bhat}} \ge 0 \lor \lrp{\theta_\etahat - \theta_{\bhat}} \le \bar{c}_3 \Bigr) > \bar{c}_4}.$$
    Then, 
    $$\liminf_{n \to \infty} \inf_{P \in \cP_{\bar{c}_3,\bar{c}_4}} P \lrp{\lrp{\theta_\etahat - \theta_{\bhat}} \in \CIalpha \Bigm| \lrp{\theta_\etahat - \theta_{\bhat}} \ge 0 \lor \lrp{\theta_\etahat - \theta_{\bhat}} \le \bar{c}_3} = 1 - \alpha.$$
    Moreover, if \Cref{as.diff_ci1} holds, 
    $$\liminf_{n \to \infty} \inf_{P \in \cP} P \lrp{\lrp{\theta_\etahat - \theta_{\bhat}} \in \CIalpha} = 1 - \alpha.$$
\end{theorem}

Under \Cref{as.diff_onesided_test}, $\CIalpha$ covers $\theta_\etahat - \theta_{\bhat}$ when this difference is positive or ``sufficiently'' negative, with $\cP_{\bar{c}_3,\bar{c}_4}$ only requiring this event to happen with positive probability. 
If $\theta_\etahat - \theta_{\bhat}$ converges to any negative value, coverage is asymptotically exact (\Cref{th.diff_ci1_power}). 
If it converges to a positive value, similarly, the normal approximation CI is exact, and the extended $\widehat{\rm CI}_{\alpha,{\rm ext}}$ is conservative. 
Failure of coverage may happen only if $\theta_\etahat - \theta_{\bhat} \Pto 0^-$, that is, it converges to zero ``from the left''. 
For such sequences, the components of $\delta_{\etahat}$, $\theta_{\etahat_{\stilde}} - \theta_{\bhat}$, may be enough negative so that the one-sided test rejects the null with high probability, but since they converge to zero, the terms in \cref{eq.diff_degenerate} cancel out, and the normal approximation CI may undercover. 
Importantly, $\CIalpha$ is valid in the case of interest $\theta_\etahat - \theta_{\bhat} \ge 0$, that is, when $\etahat$ performs equally or worse than the baseline model $\bhat$. 
This is the case, for example, when the parametric model is well-specified, as in \Cref{th.diff_pointwise}, since $\sqrt{n} \lrp{\theta_{\bhat} - \theta_{b_P}} \Pto 0$ from \Cref{as.diff_onesided_test}. 
Hence, $\CIalpha$ may overstate the advantage of $\etahat$ when it slightly outperforms $\bhat$, but not when their performances are equal or when $\etahat$ performs worse.

\Cref{as.diff_ci1} rules out the problematic sequences by ensuring that if $\theta_\etahat - \theta_{\bhat} \Pto 0^-$, $\theta_\etahat$ is close enough to $\theta_{\bhat}$ in large samples so that the one-sided test does not reject with probability higher than $\alpha$. 
It is motivated by the fact that machine learning algorithms typically penalize deviations from the mean. 
If there is little signal to be learned by $\etahat$, that is, $\theta_\etastarPn$ is close to $\theta_{b_{P_n}}$, it may be reasonable to expect that regularization will make the estimates $\etahat$ closer to $\bhat$ than to $\etastarPn$. 
For example, in the case of estimating a linear model with the Lasso, if the true coefficients are very small, penalization leads to estimated coefficients exactly equal to $0$ with high probability \citep{zhao2006model,zhang2008sparsity,wuthrich2023omitted}. 
However, this assumption may not lead to a good approximation for the behavior of DGPs where $\theta_\etastarPn$ is sufficiently distant from $\theta_{b_{P_n}}$ and $\etahat$ is estimated with no or little regularization. 

Next, I provide an alternative, more conservative CI that gives uniform coverage without relying on \Cref{as.diff_ci1}. 
It deals with sequences with $\theta_\etahat - \theta_{\bhat} \Pto 0^{-}$ by modifying $\CIalpha$ to be more conservative in the one-sided test. 
For any $\bar{c}_5 > 0$, consider the modified version of the test in \cref{eq.diff_onesided_test}: 
\begin{equation*}
    \begin{cases}
    H_{0,\etahat}: \theta_{\etahat_{\stilde}} - \theta_{\bhat} \ge - \bar{c}_5 & \text{ for all } \s \in \cS, \\
    H_{A,\etahat}: \theta_{\etahat_{\stilde}} - \theta_{\bhat} < - \bar{c}_5 & \text{ for some } \s \in \cS. 
\end{cases}
\end{equation*}
$\bar{c}_5$ represents a degree of slackness on how large $-(\theta_{\etahat_{\stilde}} - \theta_{\bhat})$ has to be to reject the null hypothesis. 
The final CI is given by  
$$\CIalpha' = \begin{cases}
    \widehat{\rm CI}_{\alpha,\cN}, \text{ if } T(\deltahat_\etahat + \bar{c}_5, \Sigmahat) > \hat{c}_{1 - \alpha} \\
    \widehat{\rm CI}_{\alpha,{\rm ext}}, \text{ otherwise}, 
\end{cases}$$
and the critical value $\hat{c}_{1 - \alpha}$ is the same as before. 
A large $\bar{c}_5$ gives more robustness in finite samples in the sense that 
$$P \lrp{\lrp{\theta_\etahat - \theta_{\bhat}} \in \CIalpha'}$$
is (weakly) increasing in $\bar{c}_5$. 
On the other hand, a large $\bar{c}_5$ leads to less power. 
Importantly, this approach is not necessary if the goal is to test the null $H_{0,\etahat}: \theta_{\etahat} - \theta_{\bhat} = 0$, since this cased is covered by \Cref{th.diff_ci1}. 
The modified confidence interval $\CIalpha'$ is intended for researchers who may want to be careful not to overestimate the magnitude of $\theta_\etahat - \theta_{\bhat}$ when it is small but negative.

\begin{theorem} \label{th.diff_ci2}
    \hyperlink{proof.th.diff_ci2}{(Uniform Asymptotic Validity of $\CIalpha'$)} \, \\
    Let \Cref{as.diff_onesided_test} hold and fix any $\bar{c}_5 > 0$. 
    Then, 
    $$\liminf_{n \to \infty} \inf_{P \in \cP} P \lrp{\lrp{\theta_\etahat - \theta_{\bhat}} \in \CIalpha'} \ge 1 - \alpha.$$
\end{theorem}

%% file: reproducibility.tex
\section{Reproducibility} \label{section.reproducibility}

The split-sample estimators and estimands defined in \Cref{section.z_estimators} depend not only on the algorithm used to estimate $\etahat$, but also on the specific splits of the sample $\cR$. 
In applications, this may lead to the undesirable phenomenon that different researchers with the same dataset, using different random splits $\cR$, may reach different conclusions in terms of statistical significance. 
Intuitively, by averaging over multiple splits, this phenomenon becomes less likely. 
In this section, I first formalize this intuition by establishing basic reproducibility properties of split-sample Z-estimators. 
Then, I develop a measure that quantifies the reproducibility of p-values from hypothesis tests based on Z-estimators for a given number of repetitions $M$.

\subsection{Basic Reproducibility Properties} \label{section.reproducibility_basic}

I establish two basic reproducibility properties for the three versions of split-sample Z-estimators defined in \Cref{section.z_estimators}. 
The two results formalize the notion that, for fixed $n$, choosing to use a larger number of repetitions $M$ improves reproducibility of the estimators. 
The results exploit the fact that $\thetahat_\etahat^{(1)}$ and $\thetahat_\etahat^{(3)}$ are averages over $M$ independent repetitions $r \in \cR$. 
For the second estimator, I use the fact that, conditional on the data $D$, $\thetahat_\etahat^{(2)}$ is still a Z-estimator where the ``observations'' are the splits $r \in \cR$ and the target parameter is the value of $\theta$ that solves the moment condition averaged over all possible splits. 
This allows me to explore large $M$ properties of $\thetahat_\etahat^{(2)}$ using arguments similar to those applied to Z-estimators (e.g., Theorem 5.9 in \citealp{van2000asymptotic}). 
For $\thetahat_\etahat^{(2)}$, I require an additional technical condition which I delay to \Cref{appendix.proofs_reproducibility_basic}. 
This assumption is analogous to standard conditions for proving consistency of Z-estimators, and holds, for example, if $\Theta'$ is bounded, $\psi_{\theta,\eta}$ is Lipschitz in $\theta$ with a Lipschitz constant that does not depend on $\eta$ or $w$, and if the solution to the moment condition averaged over all possible splits is unique. 

The first reproducibility property is that, for fixed $n$, the variance of the Z-estimators conditional on the data converge to zero as $M$ grows. 
Conditional on the data, the estimators vary only due to the random partitioning. 
This approximates the behavior of the estimators when the number of repetitions $M$ is chosen to be large. 
This guarantees that two researchers with the same dataset and different splits will calculate estimators that are arbitrarily close to each other with high probability for large enough $M$. 

\begin{proposition} \label{prop.reproducibility_basic}
    Let \ref{as.z_estimator} hold, $\pi, K$ be arbitrary, and $j \in \lrbc{1,2,3}$. 
    Additionally, let \ref{as.reproducibility_thetahat2} hold if $j=2$. 
    Then, 
    $$\Var[P]{\thetahat_{\etahat}^{(j)} \bigm| D} \Pto 0$$
    as $M \to \infty$ with $n$ fixed.
\end{proposition}

For the estimators $\thetahat_{\etahat}^{(1)}$ and $\thetahat_{\etahat}^{(3)}$, I show that the conditional variance is strictly decreasing in $M$. 
This establishes a stronger property than the asymptotic result in \Cref{prop.reproducibility_basic}: not only does reproducibility improve as $M \to \infty$, but every increase in $M$ strictly reduces variance and thus improves reproducibility. 

\begin{proposition} \label{prop.reproducibility_monotonic}
    Let \ref{as.z_estimator} hold, $n$ be fixed, $M,\pi, K$ be arbitrary, and $j \in \lrbc{1,3}$. 
    Then, if 
    $$\Var[P]{\thetahat_{\etahat}^{(j)} \bigm| D} > 0,$$
    $\Var[P]{\thetahat_{\etahat}^{(j)} \bigm| D}$ is strictly decreasing in $M$. 
\end{proposition}

\subsection{A Reproducibility Measure} \label{section.reproducibility_measure}

I propose a reproducibility measure for p-values from hypothesis tests based on transformations of split-sample Z-estimators.
Specifically, I study reproducibility of the p-value for testing $H_{0,\etahat}: h(\theta_{\etahat}) = \tau$ versus $H_{A,\etahat}: h(\theta_{\etahat}) \neq \tau$ (and its one-sided versions) for $h: \Theta \to \R$ differentiable. 
The hypotheses $H_{0,\etahat}$ and $H_{A,\etahat}$ depend on $\etahat$ since the parameter of interest, $\theta_\etahat$, depends on $\etahat$. 
Testing this hypothesis is analogous to constructing a CI for $\theta_\etahat$: in fact, inverting this test for all values of $\tau$ at significance level $\alpha$ gives the confidence interval of \Cref{section.inference_normal}.

I begin by defining the reproducibility measure, then describe the asymptotic framework I use and the technical challenges involved. 
Finally, I establish the limit distribution of the difference of t-statistics constructed from different random splits, and apply this result to construct the reproducibility measure. 
As in \Cref{section.z_estimators}, I consider $M$ repetitions of sample-splitting with $K$ folds ($K=1$ denotes repeated sample-splitting).

The goal of this section is to construct a measure $\hat{\delta}(\beta)$, for $\beta \in (0, 0.5)$, that satisfies 
$$P \lrp{p_2 > p_1 + \hat{\delta}(\beta) \Biggm| D} = \beta + o_P(1),$$
where $p_1$ and $p_2$ are p-values for $H_{0,\etahat}$ calculated with separate, independent splits. 
This measure provides the following guarantee: if a researcher calculates a p-value $p_1$ using one set of random splits, then a second researcher using the same dataset, but different splits, will obtain a p-value exceeding $p_1 + \hat{\delta}(\beta)$ with probability approximately $\beta$. 
This allows researcher 1 to assess whether their result would remain statistically significant without the computational cost of re-running the analysis. 
For example, if $p_1 < 0.05$ but $p_1 + \hat{\delta}(\beta) > 0.05$ for some small $\beta$, the researcher may need to increase $M$ to guarantee reproducibility of their finding. 

I consider an asymptotic regime where both the number of repetitions $M$ and the sample size $n$ grow to infinity, which is the main technical challenge for proving validity of my reproducibility measure. 
An alternative framework is to consider the data $D$ fixed, let $M \to \infty$, and treat each repetition as an independent observation. 
Although this alternative regime facilitates statistical analysis, it provides asymptotic guarantees only when $M$ is large relative to $n$. 
In practice, choosing $M$ much larger than $n$ is often computationally intractable. 
My asymptotic framework better reflects much of empirical practice by allowing $M$ to grow slower than $n$, so that $M$ can be, for instance, a small fraction of $n$. 
The proofs of my results under this asymptotic regime rely on the CLT of \Cref{section.z_estimators}. 

I focus on the estimator $\thetahatetahat = \thetahat_{\etahat}^{(2)}$ from \Cref{section.z_estimators}, and similar results can be extended to $\thetahat_{\etahat}^{(1)}$ and $\thetahat_{\etahat}^{(3)}$ using similar techniques. 
The $\thetahat_{\etahat}^{(2)}$ case is much more challenging because, unlike $\thetahat_{\etahat}^{(1)}$ and $\thetahat_{\etahat}^{(3)}$, $\thetahat_{\etahat}^{(2)}$ is not an average of $M$ independent terms conditional on the data. 

The setting follows \Cref{section.z_estimators}. 
Additionally, let $\cR_1$ and $\cR_2$ be independent collections of $M$ splits of the data with $K$ folds (uniformly at random). 
Let $\etahat_1$ and $\etahat_2$ be calculated with $\cR_1$ and $\cR_2$ respectively, which leads to analogous definitions of $\thetahat_{\etahat_{j}}$, $\theta_{\etahat_{j}}$, and $\sigmahat_{\etahat_j}$ for $j=1,2$. 
Under the null hypothesis and the conditions of \Cref{th.clt_z}, the t-statistic 
$$t_{\etahat_{j}} = \frac{\sqrt{n} (h(\thetahat_{\etahat_{j}}) - \tau)}{\sigmahat_{\etahat_{j}}} \leadsto \cN(0,1),$$
where $\sigmahat_{\etahat_{j}}$ is given as in \Cref{eq.sigmahat}, $\hdot(\theta)$ is a row vector with the partial derivatives of $h(\theta)$ evaluated at $\theta$, and $\hat{V}_{\etahat}$ is a plug-in estimator for $V_{\etastarP}$ defined in \Cref{appendix.proofs_z_estimators}.  Based on this result, one can calculate p-values 
$$p_j^{\pm} = 2 \Phi \lrp{- \lrm{\frac{\sqrt{n} (h(\thetahat_{\etahat_j}) - \tau)}{\sigmahat_{\etahat_j}}}},$$
$$p_j^+ = \Phi \lrp{\frac{\sqrt{n} (h(\thetahat_{\etahat_j}) - \tau)}{\sigmahat_{\etahat_j}}}, \quad p_j^- = \Phi \lrp{-\frac{\sqrt{n} (h(\thetahat_{\etahat_j}) - \tau)}{\sigmahat_{\etahat_j}}},$$
for $H_{0,\etahat_j}: h(\theta_{\etahat_j}) = \tau$ versus $H_{A,\etahat}: h(\theta_{\etahat}) \neq \tau$ and its one-sided versions. 

The asymptotic regime assumes $M^{-1} n \sigma^2_D = O_P(1)$, where $\sigma^2_D$ (defined in \cref{eq.sigma2D} in the appendix) reflects the variance of $t_{\etahat_{1}}$ conditional on the data. 
Since $h(\thetahat_{\etahat})$ and $\sigmahat_{\etahat}$ converge to non-random quantities $h(\theta_{\etastarP})$ and $\sigma_{\etastarP}$ respectively, $\sigma^2_D \Pto 0$. 
Hence, the asymptotic regime requires $M \to \infty$ at a rate slower than $n$. 
The rate of convergence of $\sigma^2_D$ depends on the rate at which $\etatilde = \cA(D)$ converges to $\etastarP$, and may be slow especially when $\etatilde$ is estimated nonparametrically. 
In \Cref{th.reproducibility_m_fast}, I show that a safe guideline for achieving the reproducibility guarantees established below is to choose $M$ of comparable magnitude to $n$.

I characterize below a central limit theorem for the difference of t-statistics constructed using different splits, which is the main ingredient for deriving my reproducibility measure in \Cref{th.reproducibility_pvalue}. 
Both results rely on the fairly technical \Cref{as.reproducibility}, stated in \Cref{section.proofs_reproducibility_measure}. 
The key condition is a Donsker-type requirement on $\lrbc{\Psi_{\etahat_{\stilde}} : \s \subseteq \bkn}$ and $\psi_{\theta, \eta, i} \psi_{\theta, \eta, j}$. 
This condition holds, for example, if $\Theta'$ and $\psi_{\theta,\eta}$ are bounded and the cross products of the entries of $\psi_{\theta,\eta}$ are Lipschitz. 
Importantly, \Cref{as.reproducibility} does not restrict the complexity of $\etahat$, it only restricts the complexity of the function classes over $\theta \in \Theta'$, and not over $\eta \in H$. 

\begin{theorem} \label{th.reproducibility_clt}
    \hyperlink{proof.th.reproducibility_clt}{(Reproducibility of t-statistics based on Z-estimators)} \, \\
    Let Assumptions \ref{as.z_estimator} and \ref{as.reproducibility} hold. 
    Then, for any $\tau \in \R$, 
    $$\lrp{\frac{\sqrt{n} \hat{\sigma}_{D}}{\sqrt{M}}}^{-1} \lrp{\frac{\sqrt{n} (h(\thetahat_{\etahat_{1}}) - \tau)}{\sigmahat_{\etahat_{1}}} - \frac{\sqrt{n} (h(\thetahat_{\etahat_{2}}) - \tau)}{\sigmahat_{\etahat_{2}}}} \leadsto \cN (0, 1)$$
    conditional on $D$ with probability approaching one. 
\end{theorem}

I introduce my reproducibility measure for each of the three tests (two-sided and both one-sided tests), where $\Phi$ is the standard normal cdf, and formalize their guarantees in \Cref{th.reproducibility_pvalue}. 
$$\hat{\delta}^\pm(\beta) = 2 \Phi \lrp{-\lrm{\frac{\sqrt{n} (h(\thetahat_{\etahat_1}) - \tau)}{\sigmahat_{\etahat_1}}} - \frac{\sqrt{n} \hat{\sigma}_{D}}{\sqrt{M}} \Phi^{-1}(\beta/2)} - 2 \Phi \lrp{-\lrm{\frac{\sqrt{n} (h(\thetahat_{\etahat_1}) - \tau)}{\sigmahat_{\etahat_1}}}},$$
$$\hat{\delta}^+(\beta) = \Phi \lrp{\frac{\sqrt{n} (h(\thetahat_{\etahat_1}) - \tau)}{\sigmahat_{\etahat_1}} - \frac{\sqrt{n} \hat{\sigma}_{D}}{\sqrt{M}} \Phi^{-1}(\beta)} - \Phi \lrp{\frac{\sqrt{n} (h(\thetahat_{\etahat_1}) - \tau)}{\sigmahat_{\etahat_1}}},$$
$$\hat{\delta}^-(\beta) = \Phi \lrp{-\frac{\sqrt{n} (h(\thetahat_{\etahat_1}) - \tau)}{\sigmahat_{\etahat_1}} - \frac{\sqrt{n} \hat{\sigma}_{D}}{\sqrt{M}} \Phi^{-1}(\beta)} - \Phi \lrp{-\frac{\sqrt{n} (h(\thetahat_{\etahat_1}) - \tau)}{\sigmahat_{\etahat_1}}}.$$

\begin{theorem} \label{th.reproducibility_pvalue}
    \hyperlink{proof.th.reproducibility_pvalue}{(Reproducibility of p-values based on Z-estimators)} \, \\
    Let Assumptions \ref{as.z_estimator} and \ref{as.reproducibility} hold, and $\tau \in \R$. 
    For any $\beta \in (0, 0.5)$ and 
    $$(p_j, \hat{\delta}(\beta)) \in \lrbc{(p_j^+, \hat{\delta}^+(\beta)),(p_j^-, \hat{\delta}^-(\beta)),(p_j^\pm, \hat{\delta}^\pm(\beta))},$$
    it follows that 
    \begin{equation} \label{eq.reproducibility_guarantee}
        P \lrp{p_2 > p_1 + \hat{\delta}(\beta) \Biggm| D} \le \beta + o_P(1),
    \end{equation}
    with equality if $p_j \in \lrbc{p_j^+, p_j^-}$. 
\end{theorem}

\Cref{th.reproducibility_pvalue} gives a novel measure of reproducibility for p-values based on split-sample Z-estimators. 
The guarantee of reproducibility in \cref{eq.reproducibility_guarantee} is inspired by the definition of $(\xi, \beta)$-reproducibility in \citet{ritzwoller2023reproducible}. 
They provide an algorithm for deciding how many repetitions $M$ of the sample-splitting procedure are necessary to guarantee reproducibility of the average across split-sample statistics. 
This covers, for example, the estimators $\thetahat_{\etahat}^{(1)}$ and $\thetahat_{\etahat}^{(3)}$. 
My approach complements theirs by focusing on reproducibility of inference, examining p-value rather than average statistics. 
My results hold for $\thetahat_{\etahat}^{(2)}$, and the arguments can easily be extended to $\thetahat_{\etahat}^{(1)}$ and $\thetahat_{\etahat}^{(3)}$. 
\citet{ritzwoller2023reproducible}'s procedure takes as input the desired level of reproducibility, and outputs the required number of repetitions $M$ that guarantees such reproducibility. 
My approach takes $M$ as input (assumed ``large''), and outputs a measure of how much reproducibility is guaranteed by such $M$. 
The asymptotic regimes also differ: \citet{ritzwoller2023reproducible} takes the data as fixed and considers that the desired threshold for the variability of the average split-sample statistic is small, while my framework considers $n$ and $M$ large.

The result in \Cref{th.reproducibility_pvalue} relies on choosing $M$ such that $M^{-1} n \sigma^2_D = O_P(1)$. 
In practice, it may be hard to choose $M$ that satisfies this condition since the rate at which $\sigma^2_D \Pto 0$ is in general unknown. 
I show that if $M$ grows too fast, i.e., if $M^{-1} n \sigma^2_D \Pto 0$, the distribution in \Cref{th.reproducibility_clt} collapses and the guarantees in \Cref{th.reproducibility_pvalue} hold conservatively. 
This gives a safe guideline for empirical implementation: choose $M$ to be at least a small fraction of $n$, such as $M= 0.1 n$, and the guarantee in \Cref{th.reproducibility_pvalue} will hold conservatively. 

\begin{theorem} \label{th.reproducibility_m_fast}
    \hyperlink{proof.th.reproducibility_m_fast}{(Reproducibility under $M^{-1} n \sigma^2_D \Pto 0$)} \, \\
    Let Assumptions \ref{as.z_estimator} and \ref{as.reproducibility} hold, replacing \ref{as.reproducibility_rate} with $M^{-1} n \sigma^2_D \Pto 0$. 
    Then, for any $\tau \in \R$,
    $$\lrp{\frac{\sqrt{n} \hat{\sigma}_{D}}{\sqrt{M}}}^{-1} \lrp{\frac{\sqrt{n} (h(\thetahat_{\etahat_{1}}) - \tau)}{\sigmahat_{\etahat_{1}}} - \frac{\sqrt{n} (h(\thetahat_{\etahat_{2}}) - \tau)}{\sigmahat_{\etahat_{2}}}} \Pto 0.$$
    For 
    $$(p_j, \hat{\delta}(\beta)) \in \lrbc{(p_j^+, \hat{\delta}^+(\beta)),(p_j^-, \hat{\delta}^-(\beta)),(p_j^\pm, \hat{\delta}^\pm(\beta))},$$
    and $\beta \in (0, 0.5)$,
    $$P \lrp{p_2 > p_1 + \hat{\delta}(\beta) \Biggm| D} \Pto 0.$$ 
\end{theorem}

%% file: application_ghana.tex
\section{Application 1: Poverty Prediction in Ghana} \label{section.application_ghana}

Understanding the drivers of poverty is at the root of much of Development Economics. For research, being able to better predict poverty dynamics is of first-order importance to both form hypotheses and then validate theories that explain poverty and poverty dynamics. For policy, accurate predictions of current or future poverty could enable better targeting of interventions (ideally then combined with causal inference on policies and interventions).

Using a sample of 319 households in urban Accra from the ISSER-Northwestern-Yale Long Term Ghana Socioeconomic Panel Survey (GSPS) \citep{ghanapaneldataset}, I examine how well I can predict which households will be below the poverty line 13 years ahead. 
The outcome of interest is an indicator for whether a household is below the poverty line in the fourth wave of GSPS (2022/2023), and I use covariates measured in wave 1 (2009/2010), that is, 13 years before. 
Of the 319 households, 22 were below the poverty line in wave 4 (around 7\%). 
I use predictive covariates including household demographics, parental education, religion, political and traditional leadership experience, asset holdings, and financial indicators (see \Cref{appendix.ghana} for details). 
Although I focus on the binary indicator of below the poverty line, the approach applies more broadly and could use other outcomes such as level of consumption or assets. 

I estimate two quantities: the mean squared error (MSE) and the fraction in poverty by tercile of predicted probability of being below the poverty line. 
In both cases, I use repeated cross-fitting with $K=3$ and $M=200$, and fit random forest models using the \texttt{R} package \texttt{ranger} implemented through \texttt{mlr3}.
Let $i \in \lrbc{1,\dots,319}$, $Y_i$ denote the indicator of whether household $i$ is below the poverty line in wave 4 of GSPS and $X_i$ the set of covariates measured in wave 1. 
The estimated MSE is given by 
$$\thetahat_{\etahat, {\rm MSE}} = \frac{1}{M} \sum_{r \in \cR} \frac{1}{n} \sum_{i=1}^{n} \lrp{Y_i - \etahat_{\stilde}(X_i)}^2.$$
For $j \in \lrbc{1,2}$, let $\hat{t}_{j,\stilde}$ be the first and second terciles of $(\etahat_{\stilde}(X_i))_{i=1}^n$, that is, 
$$\hat{t}_{j,\stilde} = \inf\left\{t : \frac{1}{|\s|} \sum_{i \in \s} \mathbf{1}\{\etahat_{\stilde}(X_i) \le t\} \ge \frac{j}{3}\right\},$$
and let $\hat{t}_{0,\stilde} = -\infty$, $\hat{t}_{3,\stilde} = \infty$. 
For $j \in \lrbc{1,2,3}$, the fraction in poverty in tercile $j$ of predicted probability of being below the poverty line is given by 
$$\thetahat_{\etahat, {\rm Frac} j} = \frac{1}{M K} \sum_{r \in \cR} \sum_{\s \in r} \frac{\sum_{i \in \s} Y_i \I{\hat{t}_{j - 1,\stilde} < \etahat_{\stilde}(X_i) \le \hat{t}_{j,\stilde}}}{\sum_{i \in \s} \I{\hat{t}_{j - 1,\stilde} < \etahat_{\stilde}(X_i) \le \hat{t}_{j,\stilde}}}.$$
I show in \Cref{appendix.ghana} that $\thetahat_{\etahat, {\rm Frac} j}$ is a Z-estimator. 

I also compare the MSE of the models estimated with random forests to the MSE of using the sample average, as described in \Cref{section.diff_performance}. 
In particular, I report p-values for the test of \Cref{section.diff_onesided_test}. 

I calculate the MSE estimators and the one-sided test both in the real data and in two Monte Carlo designs, described in \Cref{appendix.ghana}. 
The data generating processes are designed to be similar to the original dataset, preserving the empirical marginals and rank-based dependence structure of the observed data. 
In the first design, denoted \textit{Correlated}, the outcome $Y$ is correlated to the covariates $X$. 
In the second design, denoted \textit{Uncorrelated}, the outcome is independent of the covariates. 
I run around 5,000 Monte Carlo iterations for each of the three designs -- real data, ``correlated'' and ``uncorrelated'' simulated data --, drawing 200 new random splits of the sample at each Monte Carlo iteration. 
For the real data, the only source of randomness are the 200 splits, while for the simulation designs I draw a new dataset at each iteration (with 200 splits for each dataset). 
For each simulated dataset and split, I also calculate the difference between top and bottom terciles, $\thetahat_{\etahat, {\rm Frac} 3} - \thetahat_{\etahat, {\rm Frac} 1}$. 

I compare the estimates and p-values of using repeated cross-fitting (RCF) with three alternatives. 
The first is the standard ``twice the median'' (TTM) rule \citep{ruger1978maximale,gasparin2025combining,chernozhukov2025generic}: calculate the p-value (for difference in MSE or ``top minus bottom'' estimator) separately for each fold, that is, using a third of the data, take the median of the 600 p-values (200 repetitions, 3 folds) and multiply it by 2. 
The second is the Sequential Aggregation (Seq) approach of \citet{luedtke2016statistical} and \citet{wager2024sequential}: train a random forest using only fold 1, compute the t-statistic using fold 2, then train a random forest using folds 1 and 2 and compute the t-statistic in fold 3. 
The p-value for each repetition of cross-fitting uses as final t-statistic $\sqrt{2}$ times the average of the two t-statistics. 
Finally, the final p-value for each Monte Carlo iteration is twice the median over the 200 p-values coming from the 200 repetitions, similar to \citet{chernozhukov2025reply}. 
The third method is standard sample-splitting (SS): train a random forest using two thirds of the data, calculate p-value in the excluded third, not aggregating across repetitions.

\begin{figure}[!ht]
{\centering
\includegraphics[width=.99\textwidth]{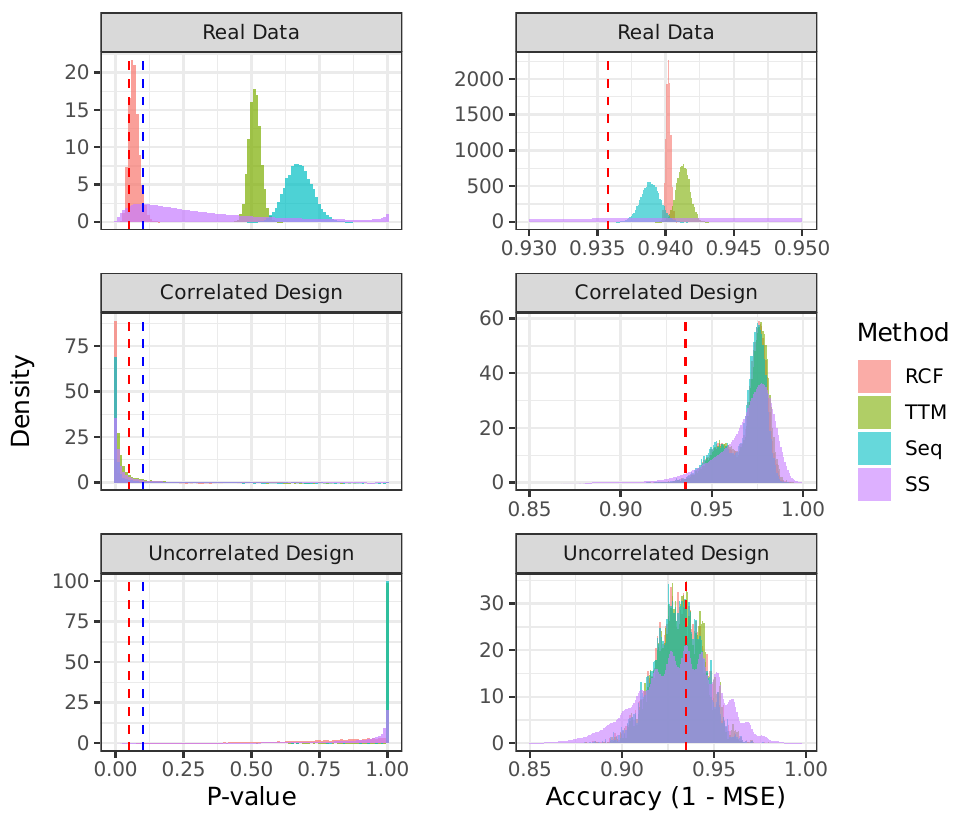}
\caption{Accuracy Comparison Across Methods and Datasets}
\label{fig.mse_combined}}
\vspace{0.3cm}
\footnotesize
Notes: Left panels show distribution across Monte Carlo iterations of p-values for testing whether random forest MSE is lower than sample average MSE. 
Vertical red and blue lines are respectively 0.05 and 0.10. 
Right panels show distribution of accuracy ($1 - MSE$) of the random forest, and vertical red lines are the accuracy of the sample average. 
Rows show results for real data (top), simulations from correlated design (middle), and simulations from uncorrelated design (bottom). 
Methods: RCF (repeated cross-fitting), TTM (twice-the-median), Seq (sequential aggregation), SS (standard sample-splitting).
Top-right panel excludes sample-splitting observations higher than 0.95 or smaller than 0.93 to improve visualization. 
The number of iterations for the real dataset, correlated design and uncorrelated design are, respectively, 11841, 28335, and 7485. 
SS uses the same number of iterations, multiplied by 600 (200 repetitions, 3 folds). 
\end{figure}

\Cref{fig.mse_combined} presents the p-values for whether random forest MSE is lower than sample average MSE, and accuracy ($1 - MSE$) point estimates across Monte Carlo iterations for the two simulation designs as well as for the real data. 
In the uncorrelated design, all methods exhibit similar accuracy on average, with sample-splitting having larger variance since it does not aggregate across multiple splits. 
All methods are conservative: the p-values concentrate around 1. 
For sample-splitting, this happens since the sample average is the best predictor of $Y$ in this design, and the random forests are noisy estimates that have larger MSE. 
The other methods are conservative for the same reason, and TTM and Seq are more conservative since they take twice the median p-value, which guards against the worst DGP. 
For the correlated design, all methods correctly give small p-values, with RCF being more concentrated around zero. 

In the real dataset, Seq often has the smallest accuracy, and TTM the highest, while RCF stands in between. 
Seq has smaller accuracy since one the two models that it averages over is trained with only a third of the data. 
RCF and TTM, on the other hand, always use two thirds of the data for training. 
The only difference between the two numbers is that RCF averages model performances over 200 repetitions while TTM takes the median. 
Hence, the higher accuracy of TTM reflects the distribution of model accuracies being left-skewed. 
Only RCF manages to consistently reject the null, concluding that poverty can be predicted from the observed covariates using a random forest model. 
TTM and Seq are more conservative, with Seq having larger p-values than TTM due to its lower accuracy. 

A comparison similar to \Cref{fig.mse_combined} for the top minus bottom estimator $\thetahat_{\etahat, {\rm Frac} 3} - \thetahat_{\etahat, {\rm Frac} 1}$ is presented in \Cref{fig.tmb_combined}. 

\begin{table}[ht]
\caption{Poverty Prediction by Tercile in Real Dataset}
\label{table.poverty_terciles}
{\centering
\begin{tabular}{lccccc}
\hline\hline
Method & Variable & Estimate & 95\% CI & p-value \\
\hline
RCF & Bottom tercile & 0.046 & [0.007, 0.085] & \\
    & Top tercile & 0.122 & [0.059, 0.184] & \\
    & Top minus bottom & 0.076 & [0.002, 0.150] & 0.023 \\
\addlinespace
TTM & Bottom tercile & 0.056 & [$-0.021$, 0.133] & \\
    & Top tercile & 0.114 & [0.006, 0.223] & \\
    & Top minus bottom & 0.083 & [$-0.052$, 0.208] & 0.228 \\
\addlinespace
Seq & Top minus bottom & -- & -- & 0.150 \\
\hline\hline
\end{tabular}
\par}

\vspace{0.3cm}
\footnotesize
Notes: Estimates of fraction below poverty line by tercile of predicted probability of being below the poverty line. 
Fraction below poverty line in entire sample is around 7\%. 
Bottom tercile corresponds to $\thetahat_{\etahat, {\rm Frac} 1}$ and top tercile to $\thetahat_{\etahat, {\rm Frac} 3}$. 
RCF, TTM and Seq correspond respectively to repeated cross-fitting, twice-the-median, and sequential aggregation. 
All estimates aggregate over 7,104,600 splits. 
\end{table} 

\Cref{table.poverty_terciles} shows the point estimates and CIs for the estimators $\thetahat_{\etahat, {\rm Frac} 1}$, $\thetahat_{\etahat, {\rm Frac} 3}$, and their difference using RCF and TTM in the real dataset, as well as p-values for testing whether the difference between top and bottom groups is positive (that is, top tercile has a larger fraction below the poverty line than the bottom tercile). 
These final estimates aggregate over all the 7,104,600 splits displayed in \Cref{fig.mse_combined}, averaging for RCF and taking the median for TTM. 
I do not display the point estimates for Seq since \citet{wager2024sequential} focuses on testing, but the p-value indicates that the difference between top and bottom groups is not significant. 
\Cref{table.poverty_terciles} shows that the difference between top and bottom terciles is statistically significant only for RCF.

%% file: application_gml.tex
\section{Application 2: Heterogeneous Treatment Effects in Charitable Giving} \label{section.application_hte}

There has been growing interest in the literature for learning features of heterogeneous treatment effects using machine learning (\citealp{chernozhukov2025generic,wager2024sequential,imai2025statistical}; for applications, see, e.g., \citealp{bryan2024big,athey2025machine,johnson2023improving}). 
I revisit the Generic Machine Learning framework of \citet{chernozhukov2025generic} (henceforth CDDF), and propose a new \textit{ensemble} estimator that uses the entire sample for calculating confidence intervals, more data for training machine learning algorithms, and aggregates predictions over multiple ML predictors into an ensemble. 
I first revisit CDDF's approach, and second introduce my ensemble estimator. 
Theoretical properties are delayed to \Cref{appendix.hte}. 
Finally, I compare my estimator to the approaches of CDDF and of \citet{wager2024sequential} in a Monte Carlo design and in an empirical application using data from \citet{karlan2007does}. 
The simulation exercise shows gains in power using the ensemble method, and the ensemble approach is the only to detect statistically significant treatment effect heterogeneity in the empirical application. 

\subsection{The Generic ML Approach of \texorpdfstring{\citet{chernozhukov2025generic}}{Chernozhukov et al.\ (2025b)}} \label{section.revisit_cddf}

CDDF proposed a method for learning features of treatment effect heterogeneity in randomized trials. 
In this section, I focus on their Sorted Group Average Treatment Effects (GATES) estimand. 
This approach consists of using a machine learning (ML) algorithm and pre-treatment covariates to find groups of individuals with larger and smaller average treatment effects (ATEs). 
If such groups exist, this means that treatment effect is heterogeneous and that this heterogeneity can be explained at least in part by observable characteristics. 
Moreover, one can explore how these groups differ in terms of these characteristics. 
They call this last step Classification Analysis (CLAN), and although I focus on GATES to simplify exposition, my results also hold for CLAN. 

First, I define some notation. 
Let $D = \lrp{Y_i, T_i, X_i}_{i=1}^n$ denote the data, where $Y$ is a scalar outcome, $T$ is the treatment assignment indicator, and $X$ is a vector of pre-treatment covariates. 
I assume that $(Y_i, T_i, X_i)$ are drawn i.i.d. from a distribution $P \in \cP$. 
Let $\cA$ denote an ML algorithm, a function that takes a dataset as input, and outputs an estimate of the Conditional Average Treatment Effect (CATE) function,
$$\eta_P(x) = \E[P]{Y(1) - Y(0) | X = x}.$$
For example, $\cA$ could be Causal Forests \citep{wager2018estimation}, or based on Random Forests, Neural Networks, or Gradient Boosting.\footnote{For example, one could use any of these three algorithms to estimate separately the functions $\E[P]{Y(1) | X = x}$ and $\E[P]{Y(0) | X = x}$, and use the difference of the two estimated functions as an estimate of the CATE.}
For any subsample $\s \subseteq \lrbc{1,\dots,n}$, let $D_{\s} = \lrbc{Y_i, T_i, X_i}_{i \in s}$, $\stilde = \lrbc{1,\dots,n} \setminus \s$, and $\etahat_{\stilde} = \cA(D_{\stilde})$, that is, $\etahat_{\stilde}$ is the model trained with algorithm $\cA$ using the subsample $D_{\stilde}$. 

The procedure is given as follows. 
First, take $M$ random subsets of $\lrbc{1,\dots,n}$ of size $\pi n$. 
For each $m=1,\dots,M$, denote the subsample by $\s_m$, where $\s_m \subseteq \lrbc{1,\dots,n}$ and $\lrm{\s_m} = \pi n$. 
For each repetition $m$, call $\s_m$ the main sample, and $\stilde_m = \lrbc{1,\dots,n} \setminus \s_m$ the auxiliary sample. 
For $m=1,\dots,M$, train the model 
\begin{equation} \label{eq.gates_etahat}
    \etahat_{\stilde_m} = \cA(D_{\stilde_m}) 
\end{equation}
using data from the auxiliary sample. 
In the main sample, calculate predicted individual treatment effects (ITEs) $\tauhat_i = \etahat_{\stilde_m}(X_i)$. 
Sort $\lrp{\tauhat_{i}}_{i \in \s}$ into $J$ quantile groups $G_1, \dots, G_J$, where 
\begin{equation} \label{eq.def_gates_groups}
    G_{j} = \lrbc{i \in \lrbc{1,\dots,n} : \tauhat_{i} \in I_j}, 
\end{equation}
with $I_j = [\hat{d}_{j-1}, \hat{d}_{j})$, $-\infty = \hat{d}_0 < \hat{d}_1 < \dots < \hat{d}_J = \infty$, and $(\hat{d}_j)_{j=0}^J$ are calculated such that the number of observations in $(G_j)_{j=1}^J$ is balanced or nearly balanced. 
For example, with $J=4$, $(G_j)_{j=1}^J$ is a partition of the sample into quartiles of $\lrp{\tauhat_{i}}_{i \in s}$.
Calculate the split-specific GATES estimator by running the weighted regression 
\begin{equation} \label{eq.gates_reg}
    Y_i = \alpha Z_i + \sum_{j=1}^{J} \gamma_{j}^{(m)} \lrbk{T_i - p(X_i)} \I{i \in G_j} + \varepsilon_i, \qquad i \in \s_m,
\end{equation}
with weights $\omega_i = \lrbc{p(X_i) \lrbk{1 - p(X_i)}}^{-1}$, where $p(x) = P(T = 1 | X = x)$ is the (known) propensity score. 
These weights guarantee correct identification of ATEs when the propensity score is not constant, that is, it ensures 
$$\gamma_{j}^{(m)} = \E[P]{Y_i(1) - Y_i(0) | i \in G_j}.$$
Denote the estimates by $(\hat{\gamma}_{j}^{(m)})_{j=1}^J$. 
A frequent parameter of interest is 
$$\delta^{(m)} = \gamma_{J}^{(m)} - \gamma_{1}^{(m)},$$
the difference in ATEs between the top and bottom groups of predicted ITEs. 
This parameter can be estimated with the analogue 
$$\hat{\delta}^{(m)} = \hat{\gamma}_{J}^{(m)} - \hat{\gamma}_{1}^{(m)},$$
and a CI can be calculated as usual, 
\begin{equation} \label{eq.gates_ci}
    (L^{(m)}, U^{(m)}) = (\hat{\delta}^{(m)} - z_{1-\alpha/2} \sigmahat^{(m)} / \sqrt{\pi n}, \hat{\delta}^{(m)} + z_{1-\alpha/2} \sigmahat^{(m)} / \sqrt{\pi n}),
\end{equation}
where $\sigmahat^{(m)} / \sqrt{\pi n}$ is a heteroscedasticity-robust standard error for $\hat{\delta}^{(m)}$ calculated as usual from the OLS regression \cref{eq.gates_reg}, and $z_{1-\alpha/2}$ is the $1-\alpha/2$ quantile of the standard normal distribution. 
Finally, the final estimators and CIs are given by 
$$\hat{\delta} = {\rm Med}(\hat{\delta}^{(m)})$$
and
$$(L, U) = \lrp{{\rm Med}(L^{(m)}), {\rm Med}(U^{(m)})},$$
where ${\rm Med}$ denotes the median across repetitions $m$. 
Conditions for the validity of this CI are established in Theorem 4.3 of CDDF. 

This approach carries a tradeoff that's not present in my method, and it considers a single ML algorithm $\cA$. 
The tradeoff regards the choice of $\pi$: a larger $\pi$ means more data is used to estimate the regression \cref{eq.gates_reg}, leading to narrower CIs in \cref{eq.gates_ci}; but fewer data are used to train the ML model in \cref{eq.gates_etahat}, likely yielding a worse estimate of the CATE. 
Moreover, regularity condition R3 in CDDF requires $\pi$ to be relatively small to guarantee that the CI $[L, U]$ covers the median of $\delta^{(m)}$ across all possible splits. 
My ensemble approach presented next avoids this tradeoff since it uses the entire sample for estimation and a larger sample for training. 
The ensemble estimator also incorporates more than one ML algorithm, which is important if one does not want to commit beforehand to any specific algorithm. 
Although CDDF's approach can be repeated with different algorithms, that comes with potential issues of multiple hypothesis testing. 

In the next subsection I propose a new GATES estimator that (i) uses the entire sample to calculate $(\hat{\gamma}_{j})_{j=1}^J$ in \cref{eq.gates_reg}, and (ii) combines predictions from multiple ML algorithms to form an \textit{ensemble}, eliminating the need for algorithm selection. 

\subsection{An Ensemble Estimator} \label{section.ensemble_estimator}

Before defining my ensemble estimator, I introduce some additional notation. 
Theoretical properties are delayed to \Cref{appendix.hte}. 
Let $A$ denote the number of machine learning algorithms that will be used for predicting ITEs. 
For $a=1,\dots,A$, let $\cA_a$ denote an ML algorithm, that is, a function that takes a dataset as input, and outputs an estimate of the CATE. 
For example, one could choose $\cA_1$ to use Random Forests, $\cA_2$ Neural Nets, and $\cA_3$ Gradient Boosting. 
For $\s \subseteq \lrbk{1,\dots,n}$ and $a=1,\dots,A$, let 
$$\etahat_{\s,a} = \cA_a(D_{\s}),$$
that is, $\etahat_{\s,a}$ is the model trained with algorithm $\cA_a$ using the subsample $D_\s$.

The ensemble approach is summarized in \Cref{algo.ensemble_gates}. 
The first difference is that instead of splitting the sample into two sets, I split it into $K$ roughly equal-sized folds $(\s_k)_{k=1}^K$, again repeating the process $M$ times. 
I calculate $A$ predicted ITEs for each individual using the $A$ ML algorithms, trained using all folds except the one that contains observation $i$. 
I denote the predicted ITEs by $\tauhat_{i,a} = \etahat_{\stilde_{k(i)}}(X_i)$, where $k(i)$ is such that $i \in \s_{k(i)}$. 
Then, to calibrate the weights for combining the multiple ML predictions into one, I split the sample again into $L$ different folds, for each repetition $m=1,\dots,M$. 
Let $\{\s'_{\ell}\}_{\ell=1}^L$ denote the $L$ folds ($m$ is not incorporated in the notation to simplify exposition). 
For $\ell=1,\dots,L$, estimate the weighted regression 
\begin{equation} \label{eq.ensemble_weights}
    Y_i = \alpha_1 + \sum_{a = 1}^A \beta_a (\tauhat_{i,a} - \taubar_{a}) \lrbk{T_i - p(X_i)} + \alpha_2 Z_i + \varepsilon_i, \qquad i \in \stilde'_\ell, 
\end{equation}
with weights $\omega_i = \lrbc{p(X_i) \lrbk{1 - p(X_i)}}^{-1}$. 
In \cref{eq.ensemble_weights}, $\taubar_{a} = \frac{1}{n - \lrm{\s'_{\ell}}} \sum_{i \not\in \s'_{\ell}} \tauhat_{i,a}$, $p(X_i)$ is the propensity score, and $Z_i$ is a vector of functions of $X_i$, for example $Z_i = (X_{1,i}, p(X_i))'$, where $X_{1,i}$ is a subset of $X_i$. 
The role of $Z_i$ is only reducing noise in estimation, so this term can be omitted if desired. 
Denote the estimates of $(\beta_{\ell, a})_{a=1}^A$ by $(\betahat_{\ell, a})_{a=1}^A$. 
The final predicted ITE is then given by 
$$\tauhat_{i} = \sum_{a = 1}^A \betahat_{\ell,a} \tauhat_{i,a}, i \in \s'_\ell.$$
Repeating this process for $\ell=1,\dots,L$ gives $\tauhat_{i}$ for every observation. 
I sort $\lrp{\tauhat_{i}}_{i \in \s}$ into groups separately by fold. 
That is, for $k=1,\dots,K$, 
\begin{equation*}
    G_{j,k} = \lrbc{i \in \s_k : \tauhat_{i} \in I_{j,k}}, 
\end{equation*}
with $I_{j,k} = [\hat{d}_{j-1,k}, \hat{d}_{j,k})$, $-\infty = \hat{d}_{0,k} < \hat{d}_{1,k} < \dots < \hat{d}_{J,k} = \infty$, and $(\hat{d}_{j,k})_{j=0}^J$ are calculated such that the number of observations in $(G_{j,k})_{j=1}^J$ is balanced or nearly balanced. 
Finally, the split-specific GATES estimator uses the whole sample, defining 
\begin{equation} \label{eq.def_gates_groups_ensemble}
    G_{j} = \bigcup_{k=1}^K G_{j,k}, 
\end{equation}
and running the weighted regression 
\begin{equation} \label{eq.gates_reg_ensemble}
    Y_i = \alpha Z_i + \sum_{j=1}^{J} \gamma_{j}^{(m)} \lrbk{T_i - p(X_i)} \I{i \in G_j} + \varepsilon_i, \qquad i \in \lrbc{1,\dots,n},
\end{equation}
with weights $\omega_i = \lrbc{p(X_i) \lrbk{1 - p(X_i)}}^{-1}$. 

\begin{algorithm}[!ht]
\caption{Ensemble Method for GATES}
\label{algo.ensemble_gates}
\textbf{Input:} Dataset $D = (Y_i, T_i, X_i)_{i=1}^n$, ML algorithms $(\mathcal{A}_a)_{a=1}^A$, repetitions $M$, number of folds $K$ (training) and $L$ (calibration), number of groups $J$. \\
\textbf{Output:} GATES estimates $(\hat{\gamma}_j)_{j=1}^J$ and standard errors $(\sigmahat_j)_{j=1}^J$
\begin{algorithmic}[1]
\For{$m = 1, \ldots, M$}
    \State \textbf{Train ML models:} Split $D$ into $K$ folds $(\s_{k})_{k=1}^K$
    \For{$k = 1, \ldots, K$ and $a = 1, \ldots, A$}
        \State Train $\etahat_{\stilde_{k},a} = \mathcal{A}_a(D_{\stilde_{k}})$; compute $\tauhat_{i,a} = \etahat_{\stilde_{k},a}(X_i)$ for $i \in \s_k$
    \EndFor
    \State \textbf{Calibrate ensemble:} Split $D$ into $L$ different folds $(\s'_{\ell})_{\ell=1}^L$
    \For{$\ell = 1, \ldots, L$}
        \State Estimate $(\betahat_{\ell, a})_{a=1}^A$ using $D_{\stilde'_{\ell}}$ as in \cref{eq.ensemble_weights}
        \State Compute $\tauhat_{i} = \sum_{a=1}^{A} \betahat_{\ell, a} \tauhat_{i,a}$ for $i \in \s'_\ell$
    \EndFor
    \State \textbf{Compute GATES:} Sort $(\tauhat_i)_{i=1}^n$ into $(G_j)_{j=1}^J$ as in \cref{eq.def_gates_groups_ensemble}
    \State Estimate $(\hat{\gamma}_j^{(m)}, \sigmahat_j^{(m)})_{j=1}^J$ with \cref{eq.gates_reg_ensemble}
\EndFor
\State Compute: $(\hat{\gamma}_j)_{j=1}^J = \frac{1}{M} \sum_{m=1}^M (\hat{\gamma}_j^{(m)})_{j=1}^J$, $(\sigmahat_j)_{j=1}^J = \frac{1}{M} \sum_{m=1}^M (\sigmahat_j^{(m)})_{j=1}^J$
\State \Return $(\hat{\gamma}_j, \sigmahat_j)_{j=1}^J$
\end{algorithmic}
\end{algorithm}

\cref{eq.ensemble_weights} is very close to the Best Linear Predictor (BLP) regression of CDDF, except that it uses the $A$ predicted ITEs instead of just one. 
The intuition behind \cref{eq.ensemble_weights} is that $(\beta_a)_{a=1}^A$ are the best linear predictor coefficients of a regression where the true CATE $\eta_P(X_i)$ is the response variable, and $(\tauhat_{i,a})_{a=1}^A$ are the independent variables (see Theorem 3.1 of CDDF). 
Hence, $\sum_{a = 1}^A \beta_a \tauhat_{i,a}$ is the best linear approximation of $\eta_P(X_i)$ given $(\tauhat_{i,a})_{a=1}^A$. 

The final estimator averages over repetitions, 
$$\hat{\delta}_\etahat = \hat{\delta} = \frac{1}{M} \sum_{m=1}^M \hat{\delta}^{(m)},$$ 
where, as before, $\hat{\delta}^{(m)} = \hat{\gamma}_{J}^{(m)} - \hat{\gamma}_{1}^{(m)}$, with $\hat{\gamma}_{J}^{(m)}$ and $\hat{\gamma}_{1}^{(m)}$ being the estimates from \cref{eq.gates_reg_ensemble}. 
The final standard error is 
\begin{equation} \label{eq.ensemble_se}
    \sigmahat_\etahat = \sigmahat = \frac{1}{M} \sum_{m=1}^M \frac{\sigmahat^{(m)}}{\sqrt{n}},
\end{equation}
where $\sigmahat^{(m)} / \sqrt{n}$ is a heteroscedasticity-robust standard error for $\hat{\delta}^{(m)}$ calculated as usual from the OLS regression \cref{eq.gates_reg_ensemble}. 
The parameter of interest is 
$$\delta_\etahat = \delta = \frac{1}{M} \sum_{m=1}^M \gamma_{J}^{(m)} - \gamma_{1}^{(m)},$$
where $\gamma_{J}^{(m)}$ and $\gamma_{1}^{(m)}$ are defined in \cref{eq.gates_reg_ensemble}. 

\subsection{Application to Charitable Giving and Monte Carlo Experiments}

I compare my new ensemble approach to two alternative methods in an empirical application and in Monte Carlo experiments. 
I revisit \citet{karlan2007does}, which sent fundraising letters to prior donors of a liberal nonprofit organization in the United States, randomizing the match ratio offered (1:1, 2:1, or 3:1) versus no match for a control group. 
I pool all match treatments into a single treatment group, focusing on the binary treatment of receiving any match offer versus none.
The outcome of interest is the amount donated in dollars. 
The predictive covariates I use include individual donation history (frequency, recency, amount), gender, state-level political variables (Bush vote share, count of court cases in which the organization was either a party to or filed a brief), and zip code-level demographics and economics (race, age, household size, income, homeownership, education, urbanization) (see \Cref{appendix.hte} for details).
I focus on the subset of 6,419 donors who donated within the last two months, as they were more responsive to the solicitation and the smaller sample facilitates computation of the Monte Carlo experiments. 

I compare the ensemble with CDDF's approach, described in \Cref{section.revisit_cddf}, and the sequential aggregation approach of \citet{luedtke2016statistical}, \citet{wager2024sequential}, and \citet{chernozhukov2025reply}. 
Sequential aggregation (Seq) consists of splitting the sample into $K$ folds, for $k=2,\dots,K$ train an ML model using folds $1$ through $k-1$, and compute GATES in the $K$-th fold. 
The final estimator is the average over the $K-1$ estimates, and the p-value uses the final t-statistic equal to $\sqrt{K-1}$ times the average of the fold-specific t-statistics. 
This approach uses more data for calculating GATES and p-values ($n (K-1)/K$ observations), but trains some ML models using fewer data (the first model uses $n/K$ observations). 
I aggregate the final estimates and p-values taking the median over $M$ repetitions as in \citet{chernozhukov2025reply}. 

I compute the three approaches across four designs: (i) using the real data (real), (ii) using the real data but shuffling the treatment assignment indicator at random (so there is no treatment effect heterogeneity) (real-shuffled), (iii) drawing from a DGP where treatment effect is partially predictable using covariates (mc-hte), (iv) drawing from a DGP where treatment effect heterogeneity is independent of covariates (mc-nohte). 
The two DGPs are meant to be similar to the real data, preserving the marginal distributions of covariates and rank-correlation structure, as described in \Cref{appendix.hte}. 
Across all methods and datasets, at each Monte Carlo iteration I use 100 repetitions of sample-splitting, take random samples (without replacement) of sizes $n=500,1000,2000,6419$ (entire dataset), and compare the number of folds $K=2,3,5,10$ (for CDDF, the ML is trained with $n(K-1)/K$ observations and GATES calculated in the remaining sample). 
For Ensemble, I draw at random between 1 and 4 ML algorithms among 10 popular algorithms available in \texttt{R}'s \texttt{mlr3verse}: XGBoost (\texttt{xgboost}), Random Forest (\texttt{ranger}), Neural Networks (\texttt{nnet}), Elastic Net (\texttt{glmnet}), k-Nearest Neighbors (\texttt{kknn}), Linear Regression (\texttt{lm}), Decision Trees (\texttt{rpart}), Fast Nearest Neighbors (\texttt{fnn}), Multivariate Adaptive Regression Splines (\texttt{earth}), and Gradient Boosting (\texttt{gbm}).
For CDDF and Seq, I draw one of the same ten algorithms at random, for each Monte Carlo iteration. 
I show the number of iterations used for each specification in \Cref{tab.iterations} in the appendix. 

\begin{figure}[!ht]
\includegraphics[width=\textwidth]{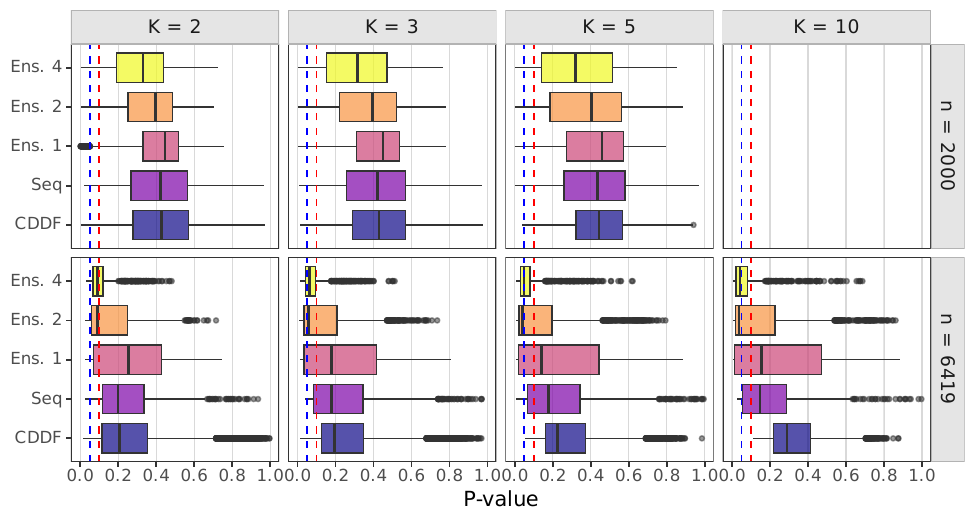}
\caption{Distribution of p-values for Top - Bottom GATES Groups -- Real Dataset}
\label{fig.hte_real}
\vspace{0.3cm}
\footnotesize
Notes: Distribution of one-sided p-values for testing whether the top tercile has a larger ATE than the bottom tercile across Monte Carlo iterations using the real dataset. 
Rows show different sample sizes ($n = 2000, 6419$), columns show different numbers of folds ($K = 2, 3, 5, 10$). 
Each box represents the distribution across Monte Carlo iterations with 100 repetitions of sample-splitting per iteration. 
Sources of randomness are the subsample when $n=2000$, which ML algorithms are used, and how the data are split. 
Red dashed line at 0.1, blue dashed line at 0.05. 
Specifications with $K = 10, n = 2000$ are excluded. 
\end{figure}

\Cref{fig.hte_real} shows the gains in power of using the ensemble method in the real dataset. 
It displays boxplots of one-sided p-values for testing whether the top tercile of predicted treatment effects has a larger ATE than the bottom tercile. 
A small p-value means rejecting the null hypothesis of no detectable treatment effect heterogeneity. 
With $n=6419$ (the entire dataset), Ensemble with 4 algorithms detects treatment effect heterogeneity at the 10\% level in more than 75\% of the iterations. 
Seq and CDDF give p-values above 10\% in most iterations. 
None of the methods are powered enough to reject the null consistently with $n = 2000$. 

\begin{figure}[!ht]
\includegraphics[width=\textwidth]{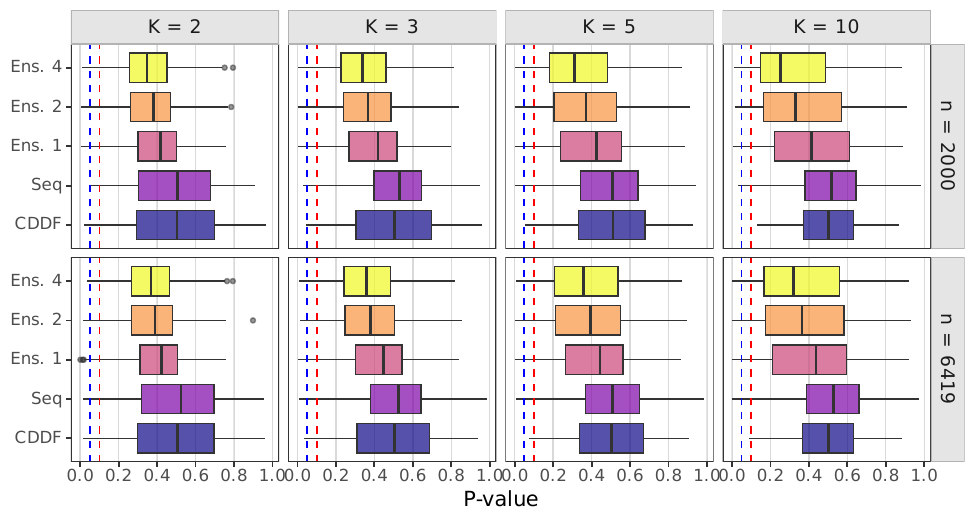}
\caption{Distribution of p-values for Top - Bottom GATES Groups -- Synthetic DGP with no Heterogeneity}
\label{fig.hte_fake_nohte}
\vspace{0.3cm}
\footnotesize
Notes: Distribution of one-sided p-values for testing whether the top tercile has a larger ATE than the bottom tercile across Monte Carlo iterations using the real dataset. 
Rows show different sample sizes ($n = 2000, 6419$), columns show different numbers of folds ($K = 2, 3, 5, 10$). 
Ens. 1, Ens. 2, and Ens. 4 represent the Ensemble method using respectively 1, 2, and 4 algorithms. 
Each box represents the distribution across Monte Carlo iterations with 100 repetitions of sample-splitting per iteration. 
Boxplots show the median (center line), interquartile range (box), and whiskers extending to 1.5 times the IQR, with points beyond shown as outliers. 
Data is generated from a synthetic DGP where there is no explainable treatment effect heterogeneity (\Cref{appendix.hte}). 
Red dashed line at 0.1, blue dashed line at 0.05. 
Specifications with $K = 10, n = 2000$ are excluded.
\end{figure}

\begin{figure}[!ht]
\includegraphics[width=\textwidth]{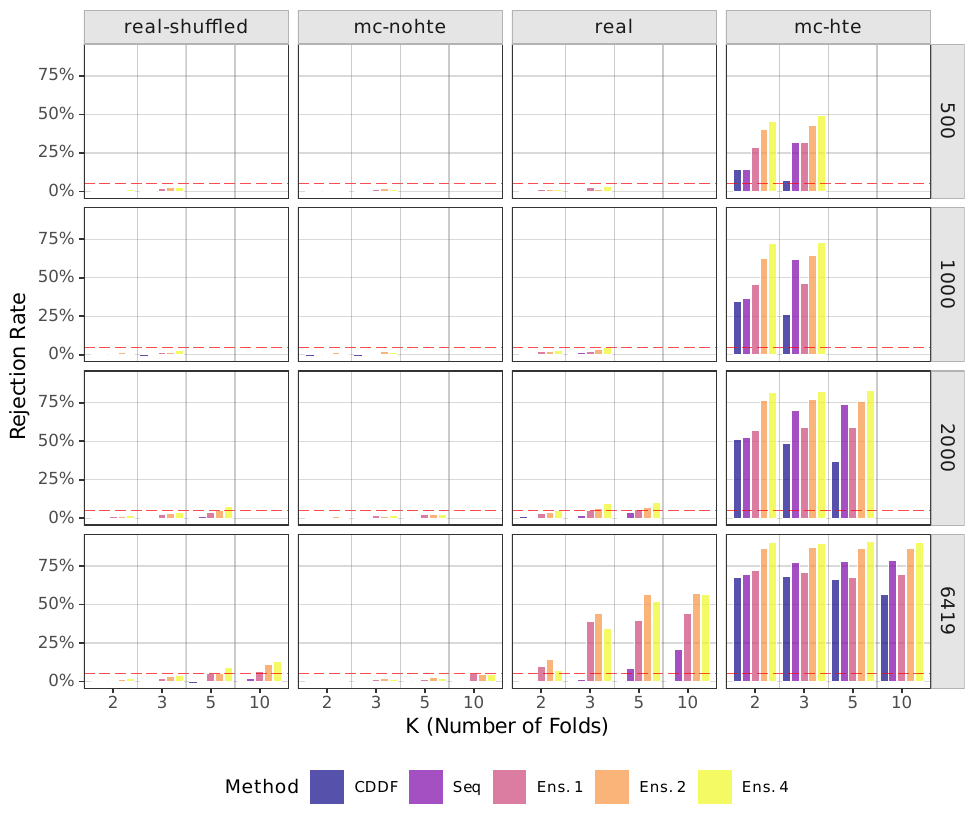}
\caption{Rejection probabilities for Top - Bottom GATES Groups at 5\% Significance Level}
\label{fig.hte_rej}
\vspace{0.3cm}
\footnotesize
Notes: Percentage of Monte Carlo iterations with p-value below 5\% for testing whether the top tercile has a larger ATE than the bottom tercile. 
Rows show different sample sizes ($n = 500, 1000, 2000, 6419$), columns show what simulation design is used. 
Specifications with $K \ge 5, n \le 1000$ and $K = 10, n = 2000$ are excluded.
\end{figure}

\Cref{fig.hte_fake_nohte} is similar to \Cref{fig.hte_real}, except that it uses the synthetic DGP where there is no detectable heterogeneity. 
It shows that all methods correctly fail to reject the null in most iterations. 
Similar figures for designs real-shuffled and mc-hte are presented in \Cref{appendix.hte}. 

\Cref{fig.hte_rej} shows the rejection probabilities at the 5\% significance level, that is, the percentage of iterations with p-value below 5\%. 
For the two datasets with no detectable heterogeneity, real-shuffled and mc-nohte, all methods are conservative when $K=2$ or $K=3$, they yield rejection probabilities below the nominal level. 
In the real-shuffle design with $n=6419$ and $K=5$ or $K=10$, the ensemble methods reject the null with probability slightly higher than nominal, but smaller than 10\%. 
With $n=2000$, only Ensemble 4 rejects the null with probability higher than nominal with $K \ge 5$ in the real-shuffled design. 
In the real dataset, CDDF almost never detects HTE, and Seq detects in less than 20\% of iterations with $K=10$ and $n=6419$. 
The ensemble methods have higher power especially in the specifications using the entire dataset. 
For example, Ensemble 2 detects heterogeneity in around 50\% of iterations with $K=3$ folds. 
In the synthetic dataset where there is detectable heterogeneity, mc-hte, as well as in the real data, Ensemble 2 and 4 have higher power across all specifications. 

As I discuss in \Cref{appendix.hte}, the rejection probability under the null of no detectable heterogeneity could in principle be above the nominal level when using the normal approximation CI. 
In \Cref{appendix.test_hte}, I propose an alternative CI that controls size under the null, at the expense of being more conservative and requiring more computational time. 
However, I note that extensive simulation experiments, including but not limited to the design of \Cref{fig.hte_rej}, suggest that Ensemble 4 is conservative for relatively small values of $K$. 
Hence, my recommendation for empirical practice is to use the normal approximation CI with Ensemble 4 and $K=3$.

%% file: conclusion.tex
\section{Conclusion} \label{section.conclusion}

As predictive algorithms become increasingly popular, using the same dataset to both train and test a new model has become routine across research, policy, and industry. 
I derived a new inference approach on model properties that averages across several splits of the sample, where at each split one part is used to train a model and the remaining to evaluate it. 
Compared to a standard 50-50 sample-splitting, my approach improves statistical and modeling power by using more data for training and evaluating, and improves reproducibility, so two researchers using different splits are more likely to reach the same conclusion about statistical significance.
Although the practice of averaging over multiple splits is not new, the confidence intervals and establishing their validity appears to be new. 

I addressed the main technical challenge, the dependence created by reusing observations across splits, by proving a central limit theorem for the large class of split-sample Z-estimators. 
Leveraging the data-dependent parameter of interest, my CLT does not require restricting the complexity of the model or its convergence rate, unlike in the classic semiparametrics problem that used cross-fitting and focused on a different parameter that is not data-dependent. 
This generality is important as it allows the model to be learned with potentially complex machine learning algorithms, as is commonly done across research, policy, and industry. 

Using the CLT, I constructed CIs based on the normal approximation that are valid in a large class of problems, and documented cases where this approximation may fail to cover the parameter of interest at nominal rate. 
I provided a new approach to inference for such problems, focusing on the particular case of inference when comparing the performance between two models.
The approach builds on my CLT, and I discussed how the arguments can be extended to other problems.
I also provided a general approach that allows the moment functions to have zero limit variance in \Cref{section.fast_convergence}, by exploring the faster-than-$\sqrt{n}$ convergence of the empirical moment equations and a tuning parameter. 

In \Cref{section.reproducibility}, I derived a new reproducibility measure for p-values calculated with split-sample Z-estimators. 
This measure is especially useful when computational resources are limited, quantifying whether a given number of split-sample repetitions suffices for two researchers using different splits to reach similar conclusions about statistical significance with high probability. 

Finally, I illustrated the empirical implications of my results by revisiting two important problems in development and public economics: predicting poverty and learning heterogeneous treatment effects in randomized experiments. 
Using a panel from Ghana \citep{ghanapaneldataset} and Monte Carlo experiments, repeated cross-fitting performed better than previous alternatives in detecting predictive power for being below the poverty line 13 years ahead. 
For the heterogeneous treatment effects application, I developed a new \textit{ensemble} method that uses the entire sample for evaluation, more data for training, and combines multiple machine learning predictors. 
I revisited \citet{karlan2007does}'s experiment on charitable giving and conducted Monte Carlo simulations. 
In both cases, my ensemble method achieved improved power for detecting heterogeneous treatment effects compared to previous alternatives.

%% file: bibliography.tex
\clearpage
\renewcommand{\baselinestretch}{1}
\bibliography{refs.bib}
\clearpage

%% file: appendix_etabar.tex
\section{Bounding the Performance of Average Model} \label{appendix.etabar}

Let $Y$ be a scalar outcome, $X$ a set of covariates, and $(\etahat_{\stilde})_{\s \in \cS}$ be a collection of models estimated through multiple splits of the sample, where $\stilde$ is the complement of $\s$, as in \Cref{section.setup}. 
For example, $\cS$ can be a vectorization of $\cR$ defined in \Cref{section.setup}, $\cS = \lrp{\s_{m,k}}_{m \in \lrbk{M}, k \in \lrbk{K}}$. 
Denote $\etabar(x) = \frac{1}{|\cS|} \sum_{\s \in \cS} \etahat_{\stilde}(x)$. 
If $Y$ is binary, some algebra manipulation gives the following equalities:
\begin{align*}
    \theta_{\etabar,1} = \int \lrm{y - \etabar(x)} dP(y,x) & = \frac{1}{|\cS|} \sum_{\s \in \cS} \int \lrm{y - \etahat_{\stilde}(x)} dP(y,x) = \theta_{\etahat,1}, \\
    \theta_{\etabar,2} = \int \lrp{y - \etabar(x)}^2 dP(y,x) & = \frac{1}{|\cS|} \sum_{\s \in \cS} \int \lrp{y - \etahat_{\stilde}(x)}^2 dP(y,x) = \theta_{\etahat,2}.
\end{align*}
Hence, one can use either $\etabar$ or a model $\etatilde(x)$ that takes value in $(\etahat_{\stilde}(x))_{\s \in \cS}$ uniformly at random, and both will yield the same out-of-sample mean absolute deviation and mean squared error. 

For the general case, if $Y$ is continuous, an application of the triangle inequality establishes a risk-contraction property for $\etabar$:
\begin{align*}
    \theta_{\etabar,1} = \int \lrm{y - \etabar(x)} dP(y,x) & \le \frac{1}{|\cS|} \sum_{\s \in \cS} \int \lrm{y - \etahat_{\stilde}(x)} dP(y,x) = \theta_{\etahat,1}, \\
    \theta_{\etabar,2} = \sqrt{\int \lrp{y - \etabar(x)}^2 dP(y,x)} & \le \frac{1}{|\cS|} \sum_{\s \in \cS} \sqrt{\int \lrp{y - \etahat_{\stilde}(x)}^2 dP(y,x)} = \theta_{\etahat,2}.
\end{align*}
Similar results hold for other distance-based functional forms where the triangle inequality applies. 
Although my framework does not cover the parameters $\theta_{\etabar,1}$ and $\theta_{\etabar,2}$, it covers $\theta_{\etahat,1}$ and $\theta_{\etahat,2}$, which are upper bounds on the error rate of using model $\etabar$. 
Hence, if one uses model $\etabar$ for out-of-sample prediction, they have the guarantee that its accuracy will be at least as large (error at least as small) as the error they can estimate, $\theta_{\etahat,1}$ or $\theta_{\etahat,2}$. 
Note that the root mean squared error estimand $\theta_{\etahat,2}$ is similar although different from the one discussed in \Cref{section.intro}. 
In this case, the estimator is also covered by \Cref{section.z_estimators} and given by 
$$\thetahat_{\etahat, 2} = \frac{1}{|\cS|} \sum_{\s \in \cS} \sqrt{\frac{1}{|\s|} \sum_{i \in \s} (Y_i - \etahat_{\stilde}(X_i))^2}.$$

%% file: proofs_z_estimators.tex
\subsection{Proofs and Extra Definitions of Section \ref{section.z_estimators}} \label{appendix.proofs_z_estimators}

Define 
$$\Psi_{\etahat}(\theta) = \frac{1}{M K} \sum_{r \in \cR} \sum_{\s \in r} \Psi_{\etahat_{\stilde}}(\theta),$$
$$\Psihat_{\etahat}(\theta) = \frac{1}{M K} \sum_{r \in \cR} \sum_{\s \in r} \Psihat_{\s,\etahat_{\stilde}}(\theta),$$
$$\Psihat_{\etastarP}(\theta) = \frac{1}{M K} \sum_{r \in \cR} \sum_{\s \in r} \Psihat_{\s,\etastarP}(\theta),$$
$$\Psidot_{\etahat}(\theta) = \frac{1}{M K} \sum_{r \in \cR} \sum_{\s \in r} \Psidot_{\etahat_{\stilde}}(\theta),$$
$$\Psidot_{\etahat} = \Psidot_{\etahat}(\theta_{\etahat}),$$
where $\Psidot_{\etahat_{\stilde}}(\theta)$ is the Jacobian matrix of $\Psi_{\eta}(\theta)$, its derivative in $\theta$. 

\begin{assumption} \label{as.z_estimator_technical}
    For some $\Theta' \subseteq \Theta$, the following conditions hold: 
    \begin{assumptionenum}
        \item \label{as.z_donsker} $\lrbc{\theta_{\etastarP} \in \Theta: P \in \cP} \subseteq \text{int}\lrp{\Theta'}$, and the classes $\cF_\eta = \lrbc{\psi_{\theta, \eta, j} : \theta \in \Theta'}$ are P-Donsker uniformly in $P \in \cP$ and $\eta \in H$ in the sense defined in \Cref{as.emp_proc} with $T = \Theta'$, where $j = 1, \dots, d$, and $\psi_{\theta, \eta, j}$ is the $j$-th coordinate of $\psi_{\theta, \eta}$;
        \item \label{as.z_approx_0} The estimators $\thetahat_{\etahat}^{(1)},\thetahat_{\etahat}^{(2)},\thetahat_{\etahat}^{(3)}$ satisfy 
        
        $$\sqrt{n} \norm{\Psihat_{\s,\etahat_{\stilde}}(\thetahat_{\etahat_{\stilde}}^{(1)})} \Pto 0 \;\; \forall \s \in \lrp{\s_{m,k}}_{m \in \lrbk{M}, k \in \lrbk{K}},$$
        $$\sqrt{n} \norm{\frac{1}{M K} \sum_{r \in \cR} \sum_{\s \in r} \Psihat_{\s,\etahat_{\stilde}}(\thetahat_{\etahat}^{(2)})} \Pto 0,$$
        $$\sqrt{n} \norm{\frac{1}{K} \sum_{\s \in r} \Psihat_{\s,\etahat_{\stilde}}(\thetahat_{\etahat_{r}}^{(3)})} \Pto 0 \;\; \forall r \in \cR,$$
        uniformly in $P \in \cP$; 
        \item \label{as.z_unique_0} For every $\varepsilon > 0$, 
        $$\sup_{P \in \cP} \sup_{\norm{\theta - \theta_{\etastarP}} > \varepsilon} - \norm{\Psi_{\etastarP}(\theta)} < 0 = \norm{\Psi_{\etastarP}(\theta_{\etastarP})};$$
        \item \label{as.z_psidot_etastar} For $\etatilde = \cA(D)$, 
        $$\norm{\Psidot_{\etatilde} - \Psidot_{\etastarP}} \Pto 0$$
        uniformly in $P \in \cP$;
        \item \label{as.z_jacobian} $\Psi_{\eta}$ is differentiable at $\theta_{\eta}$ for $\eta \in H$, and for some $\bar{c}_1 > 0$, 
        $$\inf_{P \in \cP} \lrm{\det \lrp{\Psidot_{\etastarP}}} \ge \bar{c}_1.$$
    \end{assumptionenum}
\end{assumption}

\Cref{as.z_donsker} is a Donsker condition for a subset $\Theta'$ that contains $\theta_{\etastarP}$ in its interior. 
Importantly, \Cref{as.emp_proc}, defined in \Cref{section.general}, does not restrict the complexity of the class of trained models $H$, and it allows $\etahat$ to be estimated with any machine learning algorithm as long as \Cref{as.z_etahat} holds. 
It restricts the complexity of $\psi_{\theta,\eta}$ only along $\theta \in \Theta'$, and not along $\eta \in H$. 
\Cref{as.z_donsker} holds, for example, if $\Theta'$ is bounded and $\psi_{\theta,\eta}$ is Lipschitz in $\theta$ with a Lipschitz constant that does not depend on $\eta$ or $w$. 
\Cref{as.z_approx_0} allows for approximate Z-estimators which nearly solve the moment condition, and is immediately satisfied for exact Z-estimators, for example when 
$$\frac{1}{M K} \sum_{r \in \cR} \sum_{\s \in r} \Psihat_{\s,\etahat_{\stilde}}(\thetahat_{\etahat}^{(2)})=0$$
in the case of $\thetahat_{\etahat}^{(2)}$. 
\Cref{as.z_unique_0} requires $\theta_{\etastarP}$ to be a unique and well-separated zero of $\Psi_{\etastarP}$, and can be replaced by the higher-level condition that $\lVert \thetahat_{\etahat}^{(j)} - \theta_{\etahat}^{(j)} \rVert \Pto 0$ uniformly in $P \in \cP$ for $j \in \lrbc{1,2,3}$. 
\Cref{as.z_psidot_etastar} holds under the condition that $\Psidot_{\etastarP}$ is continuous in $\eta$ around $\etastarP$. 
Finally, \Cref{as.z_jacobian} requires the absolute determinant of the Jacobian to be bounded away from zero, which guarantees its invertibility in a uniform sense over $P \in \cP$.

\begin{lemma} \label{lemma.z_unif_conv}
    Let Assumptions \ref{as.z_estimator} and \ref{as.z_estimator_technical} hold. 
    Then, uniformly in $P \in \cP$,  
    \begin{align}
        \sup_{\theta \in \Theta'} \norm{\Psihat_{\etahat}(\theta) - \Psi_{\etahat}(\theta)} & \Pto 0 \label{eq.z_unif_conv.1} \\
        \sup_{\theta \in \Theta'} \norm{\Psihat_{\etastarP}(\theta) - \Psi_{\etastarP}(\theta)} & \Pto 0 \label{eq.z_unif_conv.2} \\
        \sup_{\theta \in \Theta'} \norm{\Psi_{\etahat}(\theta) - \Psi_{\etastarP}(\theta)} & \Pto 0 \label{eq.z_unif_conv.3}
    \end{align}
\end{lemma}

\begin{proof}[\hypertarget{proof.th.clt_z}{Proof of \Cref{lemma.z_unif_conv}}] 

    \cref{eq.z_unif_conv.1,eq.z_unif_conv.2} follow from asymptotic equicontinuity established in \Cref{th.clt_general}. 
    \cref{eq.z_unif_conv.3} follows from asymptotic equicontinuity of $\Psi_{\etahat}(\theta) - \Psi_{\etastarP}(\theta)$ (follows from \Cref{as.equicontinuity}) and pointwise in $\theta$ convergence (\Cref{as.etahat}). 
\end{proof}

\begin{lemma} \label{lemma.z_consistency}
    Let Assumptions \ref{as.z_estimator} and \ref{as.z_estimator_technical} hold. 
    Then, 
    $$\norm{\thetahat_{\etahat} - \theta_{\etahat}} \Pto 0.$$
\end{lemma}

\begin{proof}[\hypertarget{proof.th.clt_z}{Proof of \Cref{lemma.z_consistency}}] 
    By \Cref{as.z_unique_0}, for any $\varepsilon > 0$, there is $\gamma > 0$ such that 
    $$\norm{\theta - \theta_{\etastarP}} > \varepsilon \implies \norm{\Psi_{\etastarP}(\theta)} > \gamma.$$
    Hence, 
    $$\sup_{P \in \cP} P \lrp{\norm{\thetahat_{\etahat} - \theta_{\etastarP}} > \varepsilon} \le \sup_{P \in \cP}  P \lrp{\norm{\Psi_{\etastarP}(\thetahat_{\etahat})} > \gamma} \to 0,$$
    since 
    \begin{align*}
        \norm{\Psi_{\etastarP}(\thetahat_{\etahat})} 
        & \le \norm{\Psi_{\etastarP}(\thetahat_{\etahat}) - \Psi_{\etahat}(\thetahat_{\etahat})} + \norm{\Psihat_{\etahat}(\thetahat_{\etahat}) - \Psi_{\etahat}(\thetahat_{\etahat})} + o_P(1) \\
        & = o_P(1),
    \end{align*}
    by \Cref{as.z_approx_0}, \cref{eq.z_unif_conv.3}, and \cref{eq.z_unif_conv.1}. 
    This implies $\norm{\thetahat_{\etahat} - \theta_{\etastarP}} \Pto 0$ uniformly in $P \in \cP$. 

    Similar happens for $\norm{\theta_{\etahat} - \theta_{\etastarP}}$. 
    For any $\varepsilon > 0$, there is $\gamma > 0$ such that 
    $$\sup_{P \in \cP} P \lrp{\norm{\theta_{\etahat} - \theta_{\etastarP}} > \varepsilon} \le \sup_{P \in \cP}  P \lrp{\norm{\Psi_{\etastarP}(\theta_{\etahat})} > \gamma} \to 0,$$
    since $\Psi_{\etahat}(\theta_{\etahat}) = 0$ and 
    $$\norm{\Psi_{\etastarP}(\theta_{\etahat})} = \norm{\Psi_{\etastarP}(\theta_{\etahat}) - \Psi_{\etahat}(\theta_{\etahat})} \Pto 0$$
    uniformly in $P \in \cP$ by \cref{eq.z_unif_conv.3}. 

    The result follows from the triangle inequality.  
\end{proof}

\begin{proof}[\hypertarget{proof.th.clt_z}{Proof of \Cref{th.clt_z}}] 

    I first show the result for the case of $\theta_{\etahat} = \theta_{\etahat}^{(2)}$ (and $\thetahat_{\etahat} = \thetahat_{\etahat}^{(2)}$). 
    Differentiability of $\Psi_{\etahat}$ and \Cref{as.z_psidot_etastar} gives 
    \begin{align}
        \Psi_{\etahat}(\thetahat_{\etahat}) - \Psi_{\etahat} \lrp{\theta_{\etahat}} \nonumber
        & = \Psidot_{\etahat} \lrp{\thetahat_{\etahat} - \theta_{\etahat}} + o_P \lrp{\norm{\thetahat_{\etahat} - \theta_{\etahat}}} \nonumber \\
        & = \Psidot_{\etastarP} \lrp{\thetahat_{\etahat} - \theta_{\etahat}} + o_P \lrp{\norm{\thetahat_{\etahat} - \theta_{\etahat}}}. \label{eq.proof.th.clt_z.differentiability}
    \end{align}

    Asymptotic equicontinuity gives 
    \begin{align}
        \sqrt{n} \lrp{\Psi_{\etahat}(\thetahat_{\etahat}) - \Psi_{\etahat}(\theta_{\etahat})}
        & = - \sqrt{n} \lrp{\Psihat_{\etahat}(\thetahat_{\etahat}) - \Psi_{\etahat}(\thetahat_{\etahat})} + o_P(1) \label{eq.proof.th.clt_z.asympequic_1} \\
        & = - \sqrt{n} \lrp{\Psihat_{\etahat}(\theta_{\etahat}) - \Psi_{\etahat}(\theta_{\etahat})} + o_P(1) \label{eq.proof.th.clt_z.asympequic_2} \\
        & = - \sqrt{n} \Psihat_{\etahat}(\theta_{\etahat}) + o_P(1) \label{eq.proof.th.clt_z.asympequic_3} \\
        & = O_P(1), \label{eq.proof.th.clt_z.asympequic_4}
    \end{align}
    where \cref{eq.proof.th.clt_z.asympequic_1} uses $\sqrt{n} \Psihat_{\etahat}(\thetahat_{\etahat}) = o_P(1)$ (\Cref{as.z_approx_0}) and $\Psi_{\etahat}(\theta_{\etahat}) = 0$, and \cref{eq.proof.th.clt_z.asympequic_2} uses \Cref{as.z_donsker} and \Cref{th.clt_general}, and 
    $$\norm{\thetahat_{\etahat} - \theta_{\etahat}} \Pto 0$$
    uniformly in $P \in \cP$, established in \Cref{lemma.z_consistency}. 
    Note that \Cref{as.z_etahat}, used for \Cref{th.clt_general}, is stronger than \Cref{as.etahat} (see proof of \Cref{th.clt_avg}). 

    By invertibility of $\Psidot_{\etastarP}$, 
    \begin{equation*}
        \sqrt{n} \norm{\thetahat_{\etahat} - \theta_{\etahat}} \le \norm{\Psidot^{-1}_{\etastarP}} \norm{\Psidot_{\etastarP} \sqrt{n} \lrp{\thetahat_{\etahat} - \theta_{\etahat}}}.
    \end{equation*}
    Plugging \cref{eq.proof.th.clt_z.differentiability} in the right-hand side gives 
    $$\sqrt{n} \norm{\thetahat_{\etahat} - \theta_{\etahat}} \norm{\Psidot^{-1}_{\etastarP}}^{-1} \le \norm{\sqrt{n} \lrp{\Psi_{\etahat}(\thetahat_{\etahat}) - \Psi_{\etahat} \lrp{\theta_{\etahat}}} + o_P\lrp{\sqrt{n} \norm{\thetahat_{\etahat} - \theta_{\etahat}}}},$$
    which implies 
    $$\sqrt{n} \norm{\thetahat_{\etahat} - \theta_{\etahat}} \lrp{\norm{\Psidot^{-1}_{\etastarP}}^{-1} + o_P(1)} \le \norm{\sqrt{n} \lrp{\Psi_{\etahat}(\thetahat_{\etahat}) - \Psi_{\etahat} \lrp{\theta_{\etahat}}}} = O_P(1),$$
    where the equality follows from \cref{eq.proof.th.clt_z.asympequic_4} and \Cref{as.z_jacobian}. 
    As a consequence, 
    \begin{equation} \label{eq.proof.th.clt_z.little_op}
        o_P \lrp{\sqrt{n} \norm{\thetahat_{\etahat} - \theta_{\etahat}}} = o_P(1). 
    \end{equation}

    Finally, combining \cref{eq.proof.th.clt_z.differentiability} and \cref{eq.proof.th.clt_z.asympequic_3} gives 
    \begin{align*}
        \Psidot_{\etastarP} \sqrt{n} \lrp{\thetahat_{\etahat} - \theta_{\etahat}}
        & = - \sqrt{n} \Psihat_{\etahat}(\theta_{\etahat}) + o_P \lrp{\sqrt{n} \norm{\thetahat_{\etahat} - \theta_{\etahat}}} + o_P(1) \\
        & = - \sqrt{n} \Psihat_{\etahat}(\theta_{\etahat}) + o_P(1). 
    \end{align*}
    Hence, 
    \begin{align*}
        \sqrt{n} \lrp{\thetahat_{\etahat} - \theta_{\etahat}} & = - \Psidot_{\etastarP}^{-1} \sqrt{n} \Psihat_{\etahat}(\theta_{\etahat}) + o_P(1), \\
        & = - \Psidot_{\etastarP}^{-1} \sqrt{n} \Psihat_{\etastarP}(\theta_{\etastarP}) + o_P(1)
    \end{align*}
    by applying \Cref{th.clt_general}, and the result follows for $\theta_{\etahat}^{(2)}$. 
    Note that \Cref{as.z_etahat} is stronger than \Cref{as.etahat} (see proof of \Cref{th.clt_avg}). 

    The results for $\theta_{\etahat}^{(1)}$ and $\theta_{\etahat}^{(3)}$ follow similarly. 
    For $\theta_{\etahat}^{(1)}$, applying the same arguments above with $K=1$ and $M=1$ gives
    $$\sqrt{n} \lrp{\thetahat_{\etahat_{\stilde}} - \theta_{\etahat_{\stilde}}} = - \Psidot_{\etastarP}^{-1} \sqrt{n} \Psihat_{\etahat_{\stilde}}(\theta_{\etahat_{\stilde}}) + o_P(1)$$
    for any $\stilde \in \lrp{\stilde_{m,k}}_{m \in \lrbk{M},k \in \lrbk{K}}$, and the result follows for $j=1$ by summing over $\s \in \lrp{\s_{m,k}}_{m \in \lrbk{M},k \in \lrbk{K}}$: 
    $$\sqrt{n} \lrp{\thetahat_{\etahat}^{(1)} - \theta_{\etahat}^{(1)}} = - \Psidot_{\etastarP}^{-1} \sqrt{n} \Psihat_{\etastarP}(\theta_{\etastarP}) + o_P(1).$$ 
    Similar holds for $j=3$ applying the arguments above with $M=1$ and $K>1$ and summing over $r \in \cR$. 
\end{proof}

%% file: proofs_inference.tex
\subsection{Proofs and Extra Definitions of Section \ref{section.inference}} \label{section.proofs_section.diff_performance}

If $\psi_{\theta,\eta}$ is differentiable in $\theta$, let $\dot{\psi}_{\theta,\eta}(w)$ be the Jacobian matrix of $\psi_{\theta,\eta}(w)$, where the derivatives are taken in respect to $\theta$. 
In that case, $\widehat{\dot{\Psi}}_{\etahat}$ can be given by 
\begin{equation} \label{eq.Psi_dot_hat}
    \widehat{\dot{\Psi}}_{\etahat} = \frac{1}{M K} \sum_{r \in \cR} \sum_{\s \in r} \frac{1}{b} \sum_{i \in \s} \dot{\psi}_{\thetahat_{\etahat},\etahat_{\stilde}}(W_i).
\end{equation}
Define
$$V_{M, K} = \begin{cases} M^{-1} \lrp{n/b + M - 1}, 
    & \text{if } K = 1 \\ 1, & \text{otherwise}, \end{cases}$$
\begin{equation} \label{eq.Vhat_etahat}
    \hat{V}_{\etahat} = V_{M, K} \widehat{\dot{\Psi}}_{\etahat}^{-1} \lrp{\frac{1}{M K} \sum_{r \in \cR} \sum_{\s \in r} \frac{1}{b} \sum_{i \in \s} \psi_{\thetahat_{\etahat},\etahat_{\stilde}}(W_i) \psi^{T}_{\thetahat_{\etahat},\etahat_{\stilde}}(W_i)} \lrp{\widehat{\dot{\Psi}}_{\etahat}^{-1}}^{T}.
\end{equation}

\begin{proof}[\hypertarget{proof.th.clt_h}{Proof of \Cref{th.clt_h}}] 
    Under the conditions of the theorem, for $j \in \lrbc{1,2,3}$, 
    $$\sqrt{n} \hat{V}_{\etahat}^{-1/2} \lrp{\thetahat_{\etahat}^{(j)} - \theta_{\etahat}^{(j)}} \leadsto \cN\lrp{0, I_d}$$
    uniformly in $P \in \cP$, where $I_d$ is the identity matrix. 
    Consistency of the inner term to $\lrp{P \psi_{\theta_{\etastarP},\etastarP} \psi^{T}_{\theta_{\etastarP},\etastarP}}$ follows similarly to the proof of \Cref{th.avg_ci}, and the result follows from the continuous mapping theorem, \Cref{th.clt_z} and the delta method. 
\end{proof}

\begin{assumption} \label{as.diff_onesided_test_technical} 
    The following conditions hold: 
    \begin{enumerate}[label=\upshape(\roman*), ref=\Cref{as.diff_onesided_test_technical}(\roman*)]
        \item \label{as.diff_onesided_test_technical1} There exists a consistent estimator $\hat{V}_{\etahat} \Pto V_{\etastarP}$ uniformly in $P \in \cP$; 
        \item \label{as.diff_onesided_test_technical2} $\norm{\Psidot_{\bhat} - \Psidot_{b_P}} \Pto 0$ uniformly in $P \in \cP$; 
        \item \label{as.diff_onesided_test_technical3} $\lrm{\thetahat_{\bhat} - \theta_{\bhat}} \Pto 0$ uniformly in $P \in \cP$; 
        \item \label{as.diff_onesided_test_technical4} $\sup_{P \in \cP} \Psidot_{b_P}^{-1} < \infty$.
    \end{enumerate}
\end{assumption}
\Cref{as.diff_onesided_test_technical1} requires $V_{\etastarP}$ to be consistently estimable, which can typically be verified as in \Cref{th.clt_h}. 
\Cref{as.diff_onesided_test_technical2} through \Cref{as.diff_onesided_test_technical4} adapt conditions \Cref{as.z_psidot_etastar} through \Cref{as.z_jacobian} to $b_P$ instead of $\etastarP$. 

I give below a formula for $\Sigmahat$ for the case of sample averages, that is, $\psi_{\theta,\eta}(w) = f_\eta(w) - \theta$. 
Analogous estimators can be defined for the general case using the fact that $\thetahat_{\etahat} - \theta_{\etahat}$ is asymptotically linear:
$$\sqrt{n} \lrp{\thetahat_{\etahat} - \theta_{\etahat}} = - \Psidot_{\etastarP}^{-1} \sqrt{n} \Psihat_{\etastarP}(\theta_{\etastarP}) + o_P(1),$$
from \Cref{th.clt_z}. 

\begin{equation} \label{eq.diff_Sigmahat}
\Sigmahat = \lrp{\hat{\Sigma}_{j,\ell}}_{j,\ell=1}^{MK},
\end{equation}
where for splits $\s_j, \s_\ell \in \cS$ with complements $\stilde_j, \stilde_\ell$,
\begin{align*}
\hat{\Sigma}_{j,j} & = \frac{1}{n^2} \sum_{i \in \stilde_j} \lrp{f_{\bhat}(W_i) - \bar{f}_{\bhat,\stilde_j}}^2 + \frac{1}{n^2} \sum_{i \in \s_j} \lrp{\tilde{f}_{j}(W_i) - \bar{\tilde{f}}_{j}}^2, \\
\hat{\Sigma}_{j,\ell} & = \frac{1}{n^2} \sum_{i \in \stilde_j \cap \stilde_\ell} \lrp{f_{\bhat}(W_i) - \bar{f}_{\bhat,\stilde_j \cap \stilde_\ell}}^2 \\
& \quad + \frac{1}{n^2} \sum_{i \in \stilde_j \cap \s_\ell} \lrp{f_{\bhat}(W_i) - \bar{f}_{\bhat,\stilde_j \cap \s_\ell}} \lrp{\tilde{f}_{\ell}(W_i) - \bar{\tilde{f}}_{\ell,\stilde_j \cap \s_\ell}} \\
& \quad + \frac{1}{n^2} \sum_{i \in \s_j \cap \stilde_\ell} \lrp{\tilde{f}_{j}(W_i) - \bar{\tilde{f}}_{j,\s_j \cap \stilde_\ell}} \lrp{f_{\bhat}(W_i) - \bar{f}_{\bhat,\s_j \cap \stilde_\ell}} \\
& \quad + \frac{1}{n^2} \sum_{i \in \s_j \cap \s_\ell} \lrp{\tilde{f}_{j}(W_i) - \bar{\tilde{f}}_{j,\s_j \cap \s_\ell}} \lrp{\tilde{f}_{\ell}(W_i) - \bar{\tilde{f}}_{\ell,\s_j \cap \s_\ell}} \quad \text{for } j \neq \ell,
\end{align*}
where $\tilde{f}_{j}(W_i) = f_{\bhat}(W_i) - \frac{n}{|\s_j|} f_{\etahat_{\stilde_j}}(W_i)$, and for any set $\s \subseteq \{1,\ldots,n\}$, $\bar{f}_{\bhat,\s} = |\s|^{-1} \sum_{i \in \s} f_{\bhat}(W_i)$ and $\bar{\tilde{f}}_{j,\s} = |\s|^{-1} \sum_{i \in \s} \tilde{f}_{j}(W_i)$, with $\bar{\tilde{f}}_{j} = \bar{\tilde{f}}_{j,\s_j}$.

Again, I give a standard error for the case of sample averages, and analogous estimators can be constructed for the general case following, e.g., \Cref{th.clt_h}. 
\begin{equation} \label{eq.diff_sigmahat_deltahat}
    \sigmahat_\deltahat^2 = \sigmahat_{\etahat}^2 + \sigmahat_{\bhat}^2 - 2 \frac{1}{M K} \sum_{r \in \cR} \sum_{\s \in r} \frac{1}{|\s|} \sum_{i \in \s} \lrp{f_{\etahat_{\stilde}}(W_i) - \thetahat_{\stilde}} \lrp{f_{\etahat_{\bhat}}(W_i) - \thetahat_{\bhat}},
\end{equation}
where $\sigmahat_{\etahat}$ is defined as in \Cref{eq.sigmahat} and 
$$\sigmahat_{\bhat}^2 = \frac{1}{n} \sum_{i=1}^{n} \lrp{f_{\bhat}(W_i) - \thetahat_{\bhat}}^2.$$

\begin{proposition} \label{prop.diff_Sigmahat}
    $\Sigmahat \Pto \Sigma$ uniformly in $P \in \cP$. 
\end{proposition}

\begin{proposition} \label{prop.diff_sigmahat_deltahat}
    $\sigmahat_\deltahat \Pto \sigma_\deltahat$ uniformly in $P \in \cP$. 
\end{proposition}

The two propositions above follow from a law of large numbers and \Cref{as.z_2plusdelta} (assumed in \Cref{as.diff_onesided_test}). 

Coverage of $\CIalpha$ is exact along any sequences where $\theta_\etastarPn < \theta_{b_{P_n}}$ in the limit, without relying on \Cref{as.diff_ci1}. 

\begin{theorem} \label{th.diff_ci1_power}
    \hyperlink{proof.th.diff_ci1_power}{(Asymptotic exactness of $\CIalpha$)} \, \\
    Let \Cref{as.diff_onesided_test} hold. 
    Then, for any sequence $(P_n)_{n \ge 1} \subseteq \cP$ such that $\lim_{n \to \infty} \theta_\etastarPn - \theta_{b_{P_n}} < 0$, 
    $$\lim_{n \to \infty} P_n \lrp{\lrp{\theta_\etahat - \theta_{\bhat}} \in \CIalpha} = 1 - \alpha.$$
\end{theorem}

\begin{proof}[\hypertarget{proof.th.diff_ci1_power}{Proof of \Cref{th.diff_ci1_power}}] 
    Follows from \cref{eq.proof.th.diff_ci1.normality} and \Cref{prop.diff_sigmahat_deltahat}. 
\end{proof}

For the proof of \Cref{th.diff_onesided_test}, define 
$$\delta_\etastarP = \lrp{\theta_{\etastarP} - \theta_{b}}_{\s \in \cS},$$
and 
$$\deltahat_\etastarP = \lrp{\thetahat_{\etastarP} - \thetahat_{b}}_{\s \in \cS}.$$

\begin{proof}[\hypertarget{proof.th.diff_onesided_test}{Proof of \Cref{th.diff_onesided_test}}] 
    I first show the result for the case $\delta_\etahat = 0$. 
    Let $C_2 > 0$, $(P_n)_{n \ge 1} \subseteq \cP$ arbitrary such that $P_n(\delta_\etahat = 0) > C_2$. 
    For any $\varepsilon > 0$ and $\s \in \cS$, denote the event 
    $$E_{\s} = \lrm{\sqrt{n} \lrbk{\lrp{\thetahat_{\etahat_{\stilde}} - \thetahat_{\bhat}} - \lrp{\theta_{\etahat_{\stilde}} - \theta_{\bhat}}} - 
    \sqrt{n} \lrbk{\lrp{\thetahat_{\etastarP} - \thetahat_{b}} - \lrp{\theta_{\etastarP} - \theta_{b}}}} > \varepsilon.$$ 
    By \Cref{th.clt_z} and \Cref{as.diff_onesided_test_bP}, 
    $$P_n \lrp{E_{\s}} \to 0,$$
    which implies 
    \begin{align*}
        P_n \lrp{E_{\s} | \delta_\etahat = 0} & \le P_n \lrp{E_{\s} | \delta_\etahat = 0} P_n(\delta_\etahat = 0) C_2^{-1} \\
        & \le \lrp{P_n \lrp{E_{\s} | \delta_\etahat = 0} P_n(\delta_\etahat = 0) + P_n \lrp{E_{\s} | \delta_\etahat \neq 0} P_n(\delta_\etahat \neq 0)} C_2^{-1} \\
        & P_n \lrp{E_{\s}} C_2^{-1} \to 0.
    \end{align*}
    Hence, 
    $$P_n \lrp{\lrm{\sqrt{n} \lrp{\deltahat_\etahat - \delta_\etahat} - \sqrt{n} \lrp{\deltahat_\etastarP - \delta_\etastarP}} > \varepsilon \Bigm| \delta_\etahat = 0} \to 0,$$
    and 
    $$\sqrt{n} \lrp{\deltahat_\etahat - \delta_\etahat} \leadsto \cN(0, \Sigma)$$
    conditional on $\delta_\etahat = 0$. 
    Together with \Cref{prop.diff_Sigmahat} and the continuous mapping theorem, this implies 
    $$T(\deltahat_\etahat, n^{-1} \Sigmahat) \leadsto T(Z, \Sigma),$$
    where $Z \sim \cN(0,\Sigma)$. 
    The result follows since the quantiles of $\cN(0,\Sigmahat)$ converge to those of $\cN(0,\Sigma)$ by the continuous mapping theorem and \Cref{prop.diff_Sigmahat}. 

    Similar happens for the case $\delta_\etahat \ge 0$. 
    The inequality comes from the fact that 
    $$\sqrt{n} \deltahat_\etahat \ge \sqrt{n} \lrp{\deltahat_\etahat - \delta_\etahat} \leadsto \cN(0, \Sigma).$$
\end{proof}

\begin{proof}[\hypertarget{proof.th.diff_pointwise}{Proof of \Cref{th.diff_pointwise}}] 
    Follows from \Cref{th.diff_onesided_test}, using 
    $$\sqrt{n} \lrp{\thetahat_{\bhat} - \theta_{\bhat}} = \sqrt{n} \lrp{\thetahat_{b_P} - \theta_{b_P}} + o_P(1)$$ 
    from \Cref{as.diff_onesided_test_bP}, so that $\sqrt{n} \lrp{\thetahat_{\etahat} - \theta_{\bhat}} \ge o_P(1)$ when $\theta_\etastarP = \theta_{b_P}$. 
\end{proof}

\begin{proof}[\hypertarget{proof.th.diff_ci1}{Proof of \Cref{th.diff_ci1}}] 
    For the first result, an argument similar to the proof of \Cref{th.diff_onesided_test} conditional on 
    \begin{equation} \label{eq.proof.th.diff_ci1.condition}
        \lrp{\theta_\etahat - \theta_{\bhat}} \ge 0 \lor \lrp{\theta_\etahat - \theta_{\bhat}} \le \bar{c}_3
    \end{equation}
    implies 
    \begin{equation} \label{eq.proof.th.diff_ci1.deltahat}
        \sqrt{n} \lrp{\deltahat_\etahat - \delta_\etahat} \leadsto \cN(0, \Sigma)
    \end{equation}
    and 
    \begin{equation} \label{eq.proof.th.diff_ci1.normality}
        \sqrt{n} \lrp{\thetahat_{\etahat} - \theta_{\etahat}} - \sqrt{n} \lrp{\thetahat_{\bhat} - \theta_{\bhat}} = - \Psidot_{\etastarP}^{-1} \sqrt{n} \Psihat_{\etastarP}(\theta_{\etastarP}) + \Psidot_{b_{P}}^{-1} \sqrt{n} \Psihat_{b_{P}}(\theta_{b_{P}}) + o_{P}(1)
    \end{equation}
    conditional on \cref{eq.proof.th.diff_ci1.condition}, uniformly in $P \in \cP_{\bar{c}_3,\bar{c}_4}$. 
    \cref{eq.proof.th.diff_ci1.normality} uses \Cref{as.diff_onesided_test_bP}, \Cref{th.clt_z}, and \Cref{prop.diff_sigmahat_deltahat}. 
    If $(P_n)_{n \ge 1} \subseteq \cP_{\bar{c}_3,\bar{c}_4}$ is such that $\lim_{n \to \infty} \theta_{\etastarPn} - \theta_{b_{P_n}} \le \bar{c_3}$, the result follows \Cref{prop.diff_sigmahat_deltahat} since \cref{eq.proof.th.diff_ci1.normality} is asymptotically normal (nondegenerate). 
    $\bar{c_3} < \lim_{n \to \infty} \theta_{\etastarPn} - \theta_{b_{P_n}} < 0$ is ruled out since that implies  
    $$P_n \lrp{\lrp{\theta_\etahat - \theta_{\bhat}} \ge 0 \lor \lrp{\theta_\etahat - \theta_{\bhat}} \le \bar{c}_3} \to 0.$$ 
    If $\lim_{n \to \infty} \theta_{\etastarPn} - \theta_{b_{P_n}} \ge 0$, the result follows from \cref{eq.proof.th.diff_ci1.deltahat} and \Cref{prop.diff_Sigmahat}. 

    For the second result, note that \cref{eq.proof.th.diff_ci1.deltahat} and \cref{eq.proof.th.diff_ci1.normality} also hold unconditionally. 
    For any sequence with $\lim_{n \to \infty} \theta_{\etastarPn} - \theta_{b_{P_n}} < 0$, the result follows from \cref{eq.proof.th.diff_ci1.normality}, and for sequences with $\lim_{n \to \infty} \theta_{\etastarPn} - \theta_{b_{P_n}} > 0$ it holds from \cref{eq.proof.th.diff_ci1.deltahat}. 
    If $\lim_{n \to \infty} \theta_{\etastarPn} - \theta_{b_{P_n}} = 0$, \Cref{as.diff_ci1} implies 
    $$\sqrt{n} \delta_\etahat \Pnto 0,$$
    and the result follows from \cref{eq.proof.th.diff_ci1.deltahat}. 
\end{proof}

\begin{proof}[\hypertarget{proof.th.diff_ci2}{Proof of \Cref{th.diff_ci2}}] 
    Follows as in the proof of \Cref{th.diff_ci1}, except for sequences with $-\bar{c}_{5} \le \lim_{n \to \infty} \theta_{\etastarPn} - \theta_{b_{P_n}} \le 0$, where the result follows from using \cref{eq.proof.th.diff_ci1.deltahat} and \Cref{prop.diff_Sigmahat}. 
\end{proof}

%% file: proofs_reproducibility_basic.tex
\subsection{Proofs and Extra Definitions of Section \ref{section.reproducibility}} 

\subsubsection{Proofs and Extra Definitions of Section \ref{section.reproducibility_basic}} \label{appendix.proofs_reproducibility_basic}

\begin{assumption} \label{as.reproducibility_thetahat2}
    The following conditions hold: 
    \begin{enumerate}[label=\upshape(\roman*), ref=\Cref{as.reproducibility_thetahat2}(\roman*)]
        \item For every $\varepsilon > 0$,
        \begin{equation*} 
            P \lrp{\sup_{\theta \in \Theta'} \norm{\frac{1}{M K} \sum_{r \in \cR} \sum_{\s \in r} \Psihat_{\s,\etahat_{\stilde}}(\theta) - \Psihat_{D}(\theta)} > \varepsilon \Biggm| D} \Pto 0,
        \end{equation*}
        where 
        $$\Psihat_{D}(\theta) = \E[P]{\Psihat_{\xi, \etahat_{\xitilde}}(\theta) \bigm| D},$$
        and $\xi$ is a random subset of $\bkn$ of size $b$ (as defined in \Cref{section.setup});
        \item For every $\varepsilon > 0$, 
    \begin{equation*} 
        \sup_{\norm{\theta - \thetahat_{D}} > \varepsilon} - \norm{\Psihat_{D}(\theta)} < 0 = \norm{\Psihat_{D}(\thetahat_{D})}
    \end{equation*}
    with probability $1$, where $\thetahat_{D}$ uniquely solves $\norm{\Psihat_{D}(\thetahat_{D})} = 0$. 
    \end{enumerate}
\end{assumption}

\begin{proof}[\hypertarget{proof.prop.reproducibility_basic}{Proof of \Cref{prop.reproducibility_basic}}]
    For $j \in \lrbc{1,3}$, the result follows from a Law of Large Numbers since $\thetahat_{\etahat}^{(j)}$ is an average of $M$ iid observations (conditional on data). 
    Note that convergence in probability to a point implies convergence of the variance to zero given uniform square integrability (\Cref{as.unif_integrability}). 
    For $j=2$, consistency follows from consistency of M-estimators (for example, Theorem 5.9 in \citet{van2000asymptotic}). 
\end{proof}

\begin{proof}[\hypertarget{proof.prop.reproducibility_monotonic}{Proof of \Cref{prop.reproducibility_monotonic}}]
    
    Let $X(r) = \frac{1}{K} \sum_{\s \in r} \thetahat_{\etahat_{\stilde}}^{(1)}$ if $j=1$ and $X(r) = \thetahat_{\etahat_r}^{(2)}$ if $j=3$. 
    Then, 
    $$\thetahat_{\etahat}^{(j)} = \frac{1}{M} \sum_{r \in \cR} X(r),$$
    and $X(r) \perp X(r')$ conditional on $D$ for $r \neq r'$. 
    It follows that 
    $$\Var[P]{\thetahat_{\etahat}^{(j)} \bigm| D} = \frac{1}{M} \Var[P]{X(r) \bigm| D}$$
    is strictly decreasing in $M$ as long as $\Var[P]{X(r) \bigm| D} > 0$. 
\end{proof}

%% file: proofs_reproducibility_measure.tex
\subsubsection{Proofs and Extra Definitions of Section \ref{section.reproducibility_measure}} \label{section.proofs_reproducibility_measure}

I first define some objects used in the proofs. 

$$g_{\theta}(r) = \frac{1}{K} \sum_{\s \in r} v_{D}^{-1} \Psi_{\etahat_{\stilde}}(\theta), \quad G_{\etahat}(\theta) = \frac{1}{M} \sum_{r \in \cR} g_{\theta}(r), \quad G_{\etabar}(\theta) = \E[P]{g_{\theta}(r) | D},$$
and $\theta_{\etabar}$ uniquely solves $G_{\etabar}(\theta_{\etabar}) = 0$. 
Note that $\theta_{\etahat}$ uniquely solves $G_{\etahat}(\theta_{\etahat}) = 0$. 
$$\dot{G}_{\etahat} = \frac{1}{M} \sum_{r \in \cR} \frac{1}{K} \sum_{\s \in r} v_{D}^{-1} \Psidot_{\etahat_{\stilde}}(\theta_{\etahat_{\stilde}});$$
$$\dot{G}_{\etabar} = \E[P]{\frac{1}{K} \sum_{\s \in r} v_{D}^{-1} \Psidot_{\etahat_{\stilde}}(\theta_{\etabar}) \Bigm| D};$$
$$\hat{V}_{G} = \lrp{\frac{1}{M K} \sum_{r \in \cR_1} \sum_{\s \in r} \Psihat_{\s, \etahat_{\stilde}}(\thetahat_{\etahat_1})} \lrp{\frac{1}{M K} \sum_{r \in \cR_1} \sum_{\s \in r} \Psihat_{\s, \etahat_{\stilde}}(\thetahat_{\etahat_1})}^{T};$$
$$\hat{\cV}_{\etahat}(\theta) = \lrp{\frac{1}{M K} \sum_{r \in \cR} \sum_{\s \in r} \frac{1}{b} \sum_{i \in \s} \psi_{\theta,\etahat_{\stilde}}(W_i) \psi^{T}_{\theta,\etahat_{\stilde}}(W_i)};$$
$$\hat{\cV}_{\etastarP}(\theta) = \lrp{\frac{1}{n} \sum_{i = 1}^n \psi_{\theta,\etastarP}(W_i) \psi^{T}_{\theta,\etastarP}(W_i)};$$
$\cV_{\eta}(\theta) = P \hat{\cV}_{\eta}(\theta)$. 
Note that 
$$\sqrt{n} \lrp{\frac{1}{M K} \sum_{r \in \cR} \sum_{\s \in r} \frac{1}{b} \sum_{i \in \s} \psi_{\theta,\etastarP}(W_i) \psi^{T}_{\theta,\etastarP}(W_i)} = \sqrt{n} \hat{\cV}_{\theta, \etastarP} + o_P(1)$$
from \cref{eq.rep_Minf} (the equality holds without $o_P(1)$ if $K > 1$). 
$$\sigma_{\etahat}^2 = V_{M, K} \hdot(\thetahatetahat) \widehat{\dot{\Psi}}_{\etahat}^{-1} \cV_{\etahat}(\thetahat_{\etahat}) \lrp{\widehat{\dot{\Psi}}_{\etahat}^{-1}}^{T} \hdot(\thetahatetahat)^{T};$$
$$v_D^2 = \Var[P]{\sigma_{\etastarP}^{-1} \hdot(\theta_{\etabar}) \dot{G}^{-1}_{\etabar} \sqrt{M} G_{\etahat}(\theta_{\etabar}) \bigm| D};$$
$$\hat{v}_D^2 = \sigmahat_{\etahat_1}^{-2} \hdot(\thetahat_{\etahat_1}) \widehat{\dot{\Psi}}_{\etahat_1}^{-1} \hat{V}_{G} \lrp{\widehat{\dot{\Psi}}_{\etahat_1}^{-1}}^{T} \hdot(\thetahat_{\etahat_1})^{T};$$
$$\zeta_D^2 = \Var[P]{2^{-1} \sigma_{\etastarP}^{-3} (h(\theta_{\etastarP}) - \tau) V_{M, K} \hdot(\theta_\etastarP) \dot{\Psi}_{\etastarP}^{-1} \sqrt{M} \cV_{\etahat}(\thetahat_{\etahat}) \lrp{\dot{\Psi}_{\etastarP}^{-1}}^{T} \hdot(\theta_{\etastarP})^{T} \Bigm| D};$$
$$\hat{a} = (\hat{a}_{1} \; \cdots \; \hat{a}_{d}) = \hdot(\thetahat_{\etahat_1}) \widehat{\dot{\Psi}}_{\etahat_1}^{-1};$$
$\hat{v}_{(i,j)}$ are the entries of $\hat{\cV}_{\etahat_1}(\thetahat_{\etahat_1})$,
\begin{align*}
    \hat{c}_{(i,j),(i',j')} & = \frac{1}{M K} \sum_{r \in \cR} \sum_{\s \in r} \lrp{\frac{1}{b} \sum_{\ell \in s} \psi_{\theta,\etahat_{\stilde},i}(W_\ell) \psi_{\theta,\etahat_{\stilde},j}(W_\ell) - \hat{v}_{(i,j)}} \\
    & \qquad \times \lrp{\frac{1}{b} \sum_{\ell \in s} \psi_{\theta,\etahat_{\stilde},i'}(W_\ell) \psi_{\theta,\etahat_{\stilde},j'}(W_\ell) - \hat{v}_{(i',j')}};
\end{align*}
$$\hat{\zeta}_D^2 = 2^{-2} \sigmahat_{\etahat_1}^{-6} (h(\thetahat_{\etahat_1}) - \tau)^2 V_{M, K}^2 \sum_{i=1}^d \sum_{j=1}^d \sum_{i'=1}^d \sum_{j'=1}^d \hat{a}_{i} \hat{a}_{j} \hat{a}_{i'} \hat{a}_{j'} \hat{c}_{(i,j),(i',j')};$$
\begin{align*}
    \rho_{v, \zeta} & = \Cov[P] \Biggl[\sigma_{\etastarP}^{-1} \hdot(\theta_{\etabar}) \dot{G}^{-1}_{\etabar} \sqrt{M} G_{\etahat}(\theta_{\etabar}), \\
    & 2^{-1} \sigma_{\etastarP}^{-3} (h(\theta_{\etastarP}) - \tau) V_{M, K} \hdot(\theta_\etastarP) \dot{\Psi}_{\etastarP}^{-1} \sqrt{M} \cV_{\etahat}(\thetahat_{\etahat}) \lrp{\dot{\Psi}_{\etastarP}^{-1}}^{T} \hdot(\theta_{\etastarP})^{T} \Biggm| D \Biggr];
\end{align*}
$$\hat{d}_{i, (j, \ell)} = \frac{1}{M K} \sum_{r \in \cR} \sum_{\s \in r} \lrp{\Psihat_{\s,\etahat_{\stilde},i}(\thetahat_{\etahat_1}) - \Psihat_{\etahat_1,i}(\thetahat_{\etahat_1})} \lrp{\frac{1}{b} \sum_{i' \in s} \psi_{\theta,\etahat_{\stilde},j}(W_{i'}) \psi_{\theta,\etahat_{\stilde},\ell}(W_{i'}) - \hat{v}_{(j,\ell)}};$$
$$\hat{\rho}_{v, \zeta} = 2^{-1} \sigmahat_{\etahat_1}^{-4} (h(\thetahat_{\etahat_1}) - \tau) V_{M, K} \sum_{i=1}^d \sum_{j=1}^d \sum_{\ell=1}^d \hat{a}_{i} \hat{a}_{j} \hat{a}_{\ell} \hat{d}_{i, (j, \ell)};$$
\begin{equation} \label{eq.sigma2D}
    \sigma^2_D = 2 (v_D^2 + \zeta_D^2 + 2 \rho_{v, \zeta});
\end{equation}
$$\hat{\sigma}^2_D = 2 (\hat{v}_D^2 + \hat{\zeta}_D^2 + 2 \hat{\rho}_{v, \zeta}).$$

\begin{assumption} \label{as.reproducibility}
    The following conditions hold: 
    \begin{assumptionenum}
        \item \label{as.reproducibility_asymp_equic} For any $\delta_n \downarrow 0$ and $\varepsilon > 0$,
        $$P \lrp{\sup_{\norm{\theta - \theta'} < \delta_n} \norm{\sqrt{M} \lrp{G_{\etahat} - G_{\etabar}}(\theta) - \sqrt{M} \lrp{G_{\etahat} - G_{\etabar}}(\theta')} > \varepsilon \Biggm| D} \Pto 0$$
        uniformly in $P \in \cP$; 
        \item \label{as.reproducibility_donsker_sigma} 
        For $(i,j) \in \lrbk{d}^2$, \Cref{as.emp_proc} holds with $T = \Theta'$ and 
        $$\cF_\eta = \lrbc{\psi_{\theta, \eta, i} \psi_{\theta, \eta, j} : \theta \in \Theta'},$$
        where $\psi_{\theta, \eta, i}$ is the $i$-th coordinate of $\psi_{\theta, \eta}$;
        \item \label{as.reproducibility_psidot} There exists an estimator $\widehat{\dot{\Psi}}_{\etahat}$ such that 
        $$\norm{\widehat{\dot{\Psi}}_{\etahat} - \dot{\Psi}_{\etastarP}} \Pto 0$$
        uniformly in $P \in \cP$;
        \item \label{as.reproducibility_sigma2} For some $\underbar{v} > 0$, 
        $$\sigma^2_{\etastarP} = \hdot(\theta_{\etastarP}) V_{\etastarP} \hdot(\theta_{\etastarP})^{T} \ge \underbar{v};$$ 
        \item \label{as.reproducibility_rate} $M^{-1} n \sigma^2_D = O_P(1)$. 
        \item \label{as.reproducibility_correlation} Either 
        $$v_D^{-1} \zeta_D \Pto c_1 \neq 1$$
        or
        $$2 \frac{\rho_{v, \zeta}}{\zeta_D v_D} \Pto c_2 \neq -1.$$
    \end{assumptionenum}
\end{assumption}

\Cref{as.reproducibility_asymp_equic} is a Donsker condition on $\lrbc{v_{D}^{-1} \Psi_{\etahat_{\stilde}} : s \subseteq \bkn}$ conditional on the data. 
It is similar to \Cref{as.z_estimator_technical}, and can typically be verified using arguments similar to the ones used to verify \Cref{as.entropy}. 
It holds, for example, if $\Theta'$ is bounded and $\psi_{\theta,\eta}$ is Lipschitz in $\theta$ with a Lipschitz constant that does not depend on $\eta$ or $w$ (see, e.g., Example 19.7 in \citealp{van2000asymptotic}). 
\Cref{as.reproducibility_donsker_sigma} is a Donsker condition similar to \Cref{as.z_donsker}, but in terms of the product $\psi_{\theta, \eta, i} \psi_{\theta, \eta, j}$ instead of $\psi_{\theta, \eta, i}$. 
It is used to derive asymptotic normality of the standard errors $\sigmahat_{\etahat_{1}}, \sigmahat_{\etahat_{2}}$. 
If $\psi_{\theta, \eta, i}(w) \le \bar{C}$ for some $\bar{C} < \infty$, that is, if the functions $\psi_{\theta, \eta, i}$ are uniformly bounded, then \Cref{as.reproducibility_donsker_sigma} is implied by \Cref{as.z_donsker} (see, e.g., Example 2.10.10 in \citealp{van2023weak}). 
\Cref{as.reproducibility_psidot} assumes the existence of a consistent estimator of $\Psidot_{\etastarP}$. 
If $\psi_{\theta, \eta}(w)$ is differentiable in $\theta$, the plug-in estimator defined in \cref{eq.Psi_dot_hat} satisfies this assumption under a uniform integrability condition on this derivative. 
Otherwise, consistent estimators can typically be constructed on a case-by-case basis \citep{hansen2022econometrics}. 
\Cref{as.reproducibility_sigma2} requires the asymptotic variance of $h(\thetahatetahat)$ to be lower bounded. 
\Cref{as.reproducibility_rate} establishes the asymptotic regime. 
Finally, \Cref{as.reproducibility_correlation} restricts a corner case where the variance of the t-statistic $\sigmahat_{\etahat_{1}}^{-1} \sqrt{n} (h(\thetahat_{\etahat_{1}}) - \tau)$ is zero because of perfect negative correlation between $\sigmahat_{\etahat_{1}}^{-1}$ and $h(\thetahat_{\etahat_{1}})$. 
Note that the quantities $\rho_{v, \zeta},\zeta_D,v_D$ can all be consistently estimated with $\hat{\rho}_{v, \zeta},\hat{\zeta}_D,\hat{v}_D$ defined previously. 

Before proving \Cref{th.reproducibility_clt}, I establish some key intermediary results.

\begin{lemma} \label{lemma.sigma_Op}
    Let the conditions of \Cref{th.reproducibility_clt} hold. 
    Then, 
    $$\sigma_{D}^{-1} v_{D} = O_P(1), \quad \sigma_{D}^{-1} \zeta_{D} = O_P(1).$$
\end{lemma}

\begin{proof}
    I show $\sigma_{D}^{-1} v_{D} = O_P(1)$ and the second result follows analogously. 

    $$\sigma^2_D = 2 (v_D^2 + \zeta_D^2 + 2 \rho_{v, \zeta}),$$

    $$\sigma^{-2}_D v_D^2 = 2^{-1} (1 + v_D^{-2} \zeta_D^2 + 2 v_D^{-2} \rho_{v, \zeta})^{-1},$$
    \begin{align*}
        v_D^{-2} \zeta_D^2 + 2 v_D^{-2} \rho_{v, \zeta} & = v_D^{-2} \zeta_D^2 + 2 v_D^{-1} \zeta_D \frac{\rho_{v, \zeta}}{\zeta_D v_D} \\
        & = v_D^{-1} \zeta_D \lrp{v_D^{-1} \zeta_D + 2 \frac{\rho_{v, \zeta}}{\zeta_D v_D}} \\
        & = \begin{cases}
            O_P(1) \text{, if } v_D^{-1} \zeta_D = O_P(1), \\ 
            o_P(1) \text{, if } v_D \zeta_D^{-1} = o_P(1),
        \end{cases}
    \end{align*}
    since $\lrm{\frac{\rho_{v, \zeta}}{\zeta_D v_D}} \le 1$. 
    Note that 
    $$v_D^{-1} \zeta_D \lrp{v_D^{-1} \zeta_D + 2 \frac{\rho_{v, \zeta}}{\zeta_D v_D}} \Pto -1 \iff v_D^{-1} \zeta_D \Pto 1 \land 2 \frac{\rho_{v, \zeta}}{\zeta_D v_D} \Pto -1,$$
    which is ruled out by \Cref{as.reproducibility_correlation}. 
\end{proof}

\begin{theorem} \label{th.reproducibility_clt_h_theta}
    Let \Cref{as.reproducibility} hold. 
    Then, for any $\varepsilon > 0$, 
    $$P \lrp{\norm{v_{D}^{-1} \sqrt{M} \lrp{\theta_{\etahat} - \theta_{\etabar}} - \lrp{- v_{D}^{-1} \dot{G}^{-1}_{\etabar} \sqrt{M} G_{\etahat}(\theta_{\etabar})}} > \varepsilon \Bigm| D} \Pto 0$$
    uniformly in $P \in \cP$, and hence 
    $$\sup_{P \in \cP} P \lrp{\norm{v_{D}^{-1} \sqrt{M} \lrp{\theta_{\etahat} - \theta_{\etabar}} - \lrp{- v_{D}^{-1} \dot{G}^{-1}_{\etabar} \sqrt{M} G_{\etahat}(\theta_{\etabar})}} > \varepsilon } \to 0.$$
    Moreover, 
    $$v_{D}^{-1} \dot{G}^{-1}_{\etabar} \sqrt{M} G_{\etahat}(\theta_{\etabar}) = O_P(1).$$
\end{theorem}

\begin{proof}[\hypertarget{proof.th.reproducibility_clt_h_theta}{Proof of \Cref{th.reproducibility_clt_h_theta}}] 

    For a random variable $X_M$ and a deterministic (conditional on $D$) sequence $a_M(D)$, I use $X_M = o_{P|D}(a_M(D))$ to denote 
    $$P \lrp{\norm{\frac{X_M}{a_M(D)}} > \varepsilon \Bigm| D} \Pto 0$$
    uniformly in $P \in \cP$ for any $\varepsilon > 0$, 
    and analogously define $O_{P|D}(a_M(D))$ similar to the $O_P$ notation.

    By differentiability of $G$, 
    \begin{align} 
        & v_{D}^{-1} \sqrt{M} \lrp{G_{\etabar}(\theta_{\etahat}) - G_{\etabar}(\theta_{\etabar})} \nonumber \\
        & = v_{D}^{-1} \sqrt{M} \dot{G}_{\etabar} \lrp{\theta_{\etahat} - \theta_{\etabar}} + o_{P|D} \lrp{\norm{v_{D}^{-1} \dot{G}^{-1}_{\etabar} \sqrt{M} \lrp{\theta_{\etahat} - \theta_{\etabar}}}}. \label{eq.proof.th.reproducibility_clt_h_theta.differentiability}
    \end{align}

    Further, 
    \begin{align}
        \sqrt{M} \lrp{G_{\etabar}(\theta_{\etahat}) - G_{\etabar}(\theta_{\etabar})}
        & = - \sqrt{M} \lrp{G_{\etahat}(\theta_{\etahat}) - G_{\etabar}(\theta_{\etahat})} \label{eq.proof.th.reproducibility_clt_h_theta.asympequic_1} \\
        & = - \sqrt{M} \lrp{G_{\etahat}(\theta_{\etabar}) - G_{\etabar}(\theta_{\etabar})} + o_{P|D}(1) \label{eq.proof.th.reproducibility_clt_h_theta.asympequic_2} \\
        & = - \sqrt{M} G_{\etahat}(\theta_{\etabar}) + o_{P|D}(1) \label{eq.proof.th.reproducibility_clt_h_theta.asympequic_3} \\
        & = O_{P|D}(1). \label{eq.proof.th.reproducibility_clt_h_theta.asympequic_4}
    \end{align}
    \cref{eq.proof.th.reproducibility_clt_h_theta.asympequic_1} uses the definitions $G_{\etabar}(\theta_{\etabar}) = G_{\etahat}(\theta_{\etahat}) = 0$, and \cref{eq.proof.th.reproducibility_clt_h_theta.asympequic_2} uses \Cref{as.reproducibility_asymp_equic}. 
    \cref{eq.proof.th.reproducibility_clt_h_theta.asympequic_4} follows from the Lindeberg CLT. 

    Combining \cref{eq.proof.th.reproducibility_clt_h_theta.differentiability,eq.proof.th.reproducibility_clt_h_theta.asympequic_3} gives 
    \begin{align*}
        & v_{D}^{-1} \sqrt{M} \lrp{\theta_{\etahat} - \theta_{\etabar}} \\
        & = - v_{D}^{-1} \dot{G}^{-1}_{\etabar} \sqrt{M} G_{\etahat}(\theta_{\etabar}) + o_{P|D} \lrp{\norm{v_{D}^{-1} \dot{G}^{-1}_{\etabar} \sqrt{M} \lrp{\theta_{\etahat} - \theta_{\etabar}}}} + o_{P|D} \lrp{\norm{v_{D}^{-1} \dot{G}^{-1}_{\etabar}}} \\
        & = - v_{D}^{-1} \dot{G}^{-1}_{\etabar} \sqrt{M} G_{\etahat}(\theta_{\etabar}) + o_{P|D} \lrp{1}, 
    \end{align*}
    since 
    $$v_{D}^{-1} \dot{G}^{-1}_{\etabar} = \E[P]{\frac{1}{K} \sum_{\s \in r} \Psidot_{\etahat_{\stilde}}(\theta_{\etabar}) \Bigm| D}^{-1} = \Psidot_{\etastarP}^{-1}(\theta_{\etastarP}) + o_P(1) = O_P(1)$$
    by \Cref{as.z_jacobian}, and an argument similar to \cref{eq.proof.th.clt_z.little_op}, exploring \cref{eq.proof.th.reproducibility_clt_h_theta.asympequic_4}, gives
    $$\sqrt{M} \norm{\theta_{\etahat} - \theta_{\etabar}} = O_{P|D}(1).$$ 
    The second result follows since, for any events $A$ and $B$,
    $$P(A | B) = o_P(1) \implies \sup_{P \in \cP} P(A) = \sup_{P \in \cP} \E[P]{P(A | B)} \to 0.$$

    Finally, 
    $$\norm{v_{D}^{-1} \dot{G}^{-1}_{\etabar} \sqrt{M} G_{\etahat}(\theta_{\etabar})} \le \norm{v_{D}^{-1} \dot{G}^{-1}_{\etabar}} \norm{\sqrt{M} G_{\etahat}(\theta_{\etabar})} = O_P(1) O_{P|D}(1).$$
\end{proof}

\begin{proof}[\hypertarget{proof.th.reproducibility_clt}{Proof of \Cref{th.reproducibility_clt}}] 

    The proof is divided into three main steps. 
    First, I show that 
    \begin{equation} \label{eq.proof.th.reproducibility.decomp_thetahat}
        v_{D}^{-1} \sqrt{M} (h(\thetahat_{\etahat_{1}}) - h(\thetahat_{\etahat_{2}})) = - \hdot(\theta_{\etabar}) v_{D}^{-1} \dot{G}^{-1}_{\etabar} \sqrt{M} \lrp{G_{\etahat_{1}}(\theta_{\etabar}) - G_{\etahat_{2}}(\theta_{\etabar})} + o_P \lrp{1}.
    \end{equation}
    Second, I show that 
    \begin{align} 
        & \zeta_{D}^{-1} \sqrt{M} \lrp{\sigmahat_{\etahat_{1}} - \sigmahat_{\etahat_{2}}} \label{eq.proof.th.reproducibility.decomp_sigmahat} \\
        & = (2 \sigma_{\etastarP})^{-1} V_{M, K} \hdot(\theta_\etastarP) \dot{\Psi}_{\etastarP}^{-1} \zeta_{D}^{-1} \sqrt{M} \lrp{\cV_{\etahat_1}(\thetahat_{\etahat_1}) - \cV_{\etahat_2}(\thetahat_{\etahat_2})} \lrp{\dot{\Psi}_{\etastarP}^{-1}}^{T} \hdot(\theta_{\etastarP})^{T} + o_P(1). \nonumber
    \end{align}
    Finally, I combine the previous steps to reach the result. 

    \paragraph{Step one.} 
    \begin{align*}
        & v_{D}^{-1} \sqrt{M} (h(\thetahat_{\etahat_{1}}) - h(\thetahat_{\etahat_{2}})) \\
        & = v_{D}^{-1} \sqrt{M} (h(\thetahat_{\etahat_{1}}) - h(\theta_{\etahat_{1}})) - v_{D}^{-1} \sqrt{M} (h(\thetahat_{\etahat_{2}}) - h(\theta_{\etahat_{2}})) + v_{D}^{-1} \sqrt{M} (h(\theta_{\etahat_{1}}) - h(\theta_{\etahat_{2}})) \\
        & = v_{D}^{-1} \sqrt{M} (h(\theta_{\etahat_{1}}) - h(\theta_{\etahat_{2}})) + o_P(1),
    \end{align*}
    since 
    \begin{align*}
        \sqrt{n} \lrp{h(\thetahat_{\etahat_{1}}) - h(\theta_{\etahat_{1}})} & = \hdot(\theta_{\etahat_{1}}) \sqrt{n} (\thetahat_{\etahat_{1}} - \theta_{\etahat_{1}}) + o_P\lrp{\sqrt{n} \norm{\thetahat_{\etahat_{1}} - \theta_{\etahat_{1}}}} \\
        & = \hdot(\theta_{\etahat_{1}}) \Psidot_{\etastarP}^{-1} \sqrt{n} \Psihat_{\etahat_1}(\theta_{\etahat_1}) + o_P(1) \\
        & = \hdot(\theta_{\etastarP}) \Psidot_{\etastarP}^{-1} \sqrt{n} \lrp{\Psihat_{\etastarP}(\theta_{\etastarP}) - \Psi_{\etastarP}(\theta_{\etastarP})} + o_P(1),
    \end{align*}
    where $o_P\lrp{\sqrt{n} \norm{\thetahat_{\etahat_{1}} - \theta_{\etahat_{1}}}} = o_P(1)$ by \Cref{th.clt_z}, the second equality holds from \Cref{th.clt_z}, and the last equality from \Cref{th.clt_general}, using the fact that $\norm{\theta_{\etahat_1} - \theta_{\etastarP}} = o_P(1)$. 
    Note that $v_{D}^{-1} \sqrt{M}/\sqrt{n} = O_P(1)$ from \Cref{lemma.sigma_Op}. 

    By differentiability of $h$, 
    $$v_{D}^{-1} \sqrt{M} \lrp{h(\theta_{\etahat_{1}}) - h(\theta_{\etabar})} = v_{D}^{-1} \sqrt{M} \hdot(\theta_{\etabar}) (\theta_{\etahat_{1}} - \theta_{\etabar}) + o_P \lrp{1},$$
    since $v_{D}^{-1} \sqrt{M} \norm{\theta_{\etahat_{1}} - \theta_{\etabar}} = O_P(1)$ from \Cref{th.reproducibility_clt_h_theta}. 
    This implies 
    $$v_{D}^{-1} \sqrt{M} \lrp{h(\theta_{\etahat_{1}}) - h(\theta_{\etahat_{2}})} = \hdot(\theta_{\etabar}) v_{D}^{-1} \sqrt{M} (\theta_{\etahat_{1}} - \theta_{\etahat_{2}}) + o_P \lrp{1}.$$
    \Cref{th.reproducibility_clt_h_theta} gives 
    \begin{align*}
        v_{D}^{-1} \sqrt{M} \lrp{\theta_{\etahat_1} - \theta_{\etahat_2}} & = v_{D}^{-1} \sqrt{M} \lrp{\theta_{\etahat_1} - \theta_{\etabar}} - v_{D}^{-1} \sqrt{M} \lrp{\theta_{\etahat_2} - \theta_{\etabar}} \\ 
        & = - \dot{G}^{-1}_{\etabar} v_{D}^{-1} \sqrt{M} \lrp{G_{\etahat_{1}}(\theta_{\etabar}) - G_{\etahat_{2}}(\theta_{\etabar})} + o_P(1).
    \end{align*}
    \cref{eq.proof.th.reproducibility.decomp_thetahat} follows from combining the two previous displays. 

    \paragraph{Step two.} 
    \begin{align*}
        \zeta_{D}^{-1} \sqrt{M} (\sigmahat_{\etahat_1}^2 - \sigmahat_{\etahat_2}^2) & = \zeta_{D}^{-1} \sqrt{M} (\sigmahat_{\etahat_1}^2 - \sigma_{\etahat_1}^2) - \zeta_{D}^{-1} \sqrt{M} (\sigmahat_{\etahat_2}^2 - \sigma_{\etahat_2}^2) + \zeta_{D}^{-1} \sqrt{M} (\sigma_{\etahat_1}^2 - \sigma_{\etahat_2}^2) \\
        & = \zeta_{D}^{-1} \sqrt{M} (\sigma_{\etahat_1}^2 - \sigma_{\etahat_2}^2) + o_P(1),
    \end{align*}
    since 
    \begin{align*}
        & \zeta_{D}^{-1} \sqrt{M} (\sigmahat_{\etahat_1}^2 - \sigma_{\etahat_1}^2) - \zeta_{D}^{-1} \sqrt{M} (\sigmahat_{\etahat_2}^2 - \sigma_{\etahat_2}^2) \\
        & = \lrp{\frac{\sqrt{M}}{\sqrt{n}} \zeta_{D}^{-1}} \lrp{\sqrt{n} (\sigmahat_{\etahat_1}^2 - \sigma_{\etahat_1}^2) - \sqrt{n} (\sigmahat_{\etahat_2}^2 - \sigma_{\etahat_2}^2)} \\
        & = O_P(1) \lrp{\sqrt{n} (\sigmahat_{\etahat_1}^2 - \sigma_{\etahat_1}^2) - \sqrt{n} (\sigmahat_{\etahat_2}^2 - \sigma_{\etahat_2}^2)},
    \end{align*}
    and 
    \begin{align*}
        & \sqrt{n} (\sigmahat_{\etahat_1}^2 - \sigma_{\etahat_1}^2) - \sqrt{n} (\sigmahat_{\etahat_2}^2 - \sigma_{\etahat_2}^2) \\
        & = V_{M, K} \hdot(\thetahat_{\etahat_1}) \widehat{\dot{\Psi}}_{\etahat_1}^{-1} \sqrt{n} \lrp{\hat{\cV}_{\etahat_1}(\thetahat_{\etahat_1}) - \cV_{\etahat_1}(\thetahat_{\etahat_1})} \lrp{\widehat{\dot{\Psi}}_{\etahat_1}^{-1}}^{T} \hdot(\thetahat_{\etahat_1})^{T} \\
        & \quad - V_{M, K} \hdot(\thetahat_{\etahat_2}) \widehat{\dot{\Psi}}_{\etahat_2}^{-1} \sqrt{n} \lrp{\hat{\cV}_{\etahat_2}(\thetahat_{\etahat_2}) - \cV_{\etahat_2}(\thetahat_{\etahat_2})} \lrp{\widehat{\dot{\Psi}}_{\etahat_2}^{-1}}^{T} \hdot(\thetahat_{\etahat_2})^{T} \\
        & = V_{M, K} \hdot(\thetahat_{\etahat_1}) \widehat{\dot{\Psi}}_{\etahat_1}^{-1} \sqrt{n} \lrp{\hat{\cV}_{\etastarP}(\theta_{\etabar}) - \cV_{\etastarP}(\theta_{\etabar})} \lrp{\widehat{\dot{\Psi}}_{\etahat_1}^{-1}}^{T} \hdot(\thetahat_{\etahat_1})^{T} \\
        & \quad - V_{M, K} \hdot(\thetahat_{\etahat_2}) \widehat{\dot{\Psi}}_{\etahat_2}^{-1} \sqrt{n} \lrp{\hat{\cV}_{\etastarP}(\theta_{\etabar}) - \cV_{\etastarP}(\theta_{\etabar})} \lrp{\widehat{\dot{\Psi}}_{\etahat_2}^{-1}}^{T} \hdot(\thetahat_{\etahat_2})^{T} + o_P(1) \\
        & = o_P(1),
    \end{align*}
    where the second equality follows from \Cref{as.reproducibility_donsker_sigma} and \Cref{th.clt_general}, and the last equality uses $\sqrt{n} \lrp{\hat{\cV}_{\etastarP}(\theta_{\etabar}) - \cV_{\etastarP}(\theta_{\etabar})} = O_P(1)$. 
        
    Finally, 
    $$\zeta_{D}^{-1} \sqrt{M} (\sigma_{\etahat_1} - \sigma_{\etahat_2}) = \frac{\zeta_{D}^{-1} \sqrt{M} (\sigma_{\etahat_1}^2 - \sigma_{\etahat_2}^2)}{\sigma_{\etahat_1} + \sigma_{\etahat_2}} = (2 \sigma_{\etastarP})^{-1} \zeta_{D}^{-1} \sqrt{M} (\sigma_{\etahat_1}^2 - \sigma_{\etahat_2}^2) + o_P(1),$$
    and 
    \begin{align*}
        & \zeta_{D}^{-1} \sqrt{M} (\sigma_{\etahat_1}^2 - \sigma_{\etahat_2}^2) \\
        & = V_{M, K} \hdot(\thetahat_{\etahat_1}) \widehat{\dot{\Psi}}_{\etahat_1}^{-1} \zeta_{D}^{-1} \sqrt{M} \lrp{\cV_{\etahat_1}(\thetahat_{\etahat_1}) - \cV_{\etahat_2}(\thetahat_{\etahat_2})} \lrp{\widehat{\dot{\Psi}}_{\etahat_1}^{-1}}^{T} \hdot(\thetahat_{\etahat_1})^{T} + o_P(1) \\
        & = V_{M, K} \hdot(\theta_\etastarP) \dot{\Psi}_{\etastarP}^{-1} \zeta_{D}^{-1} \sqrt{M} \lrp{\cV_{\etahat_1}(\thetahat_{\etahat_1}) - \cV_{\etahat_2}(\thetahat_{\etahat_2})} \lrp{\dot{\Psi}_{\etastarP}^{-1}}^{T} \hdot(\theta_{\etastarP})^{T} + o_P(1),
    \end{align*}
    using the fact that $\zeta_{D}^{-1} \sqrt{M} \lrp{\cV_{\etahat_1}(\thetahat_{\etahat_1}) - \cV_{\etahat_2}(\thetahat_{\etahat_2})} = O_P(1)$. 

    \paragraph{Step three.} 
    \begin{align*}
        & \lrp{\frac{\sqrt{n} \hat{\sigma}_{D}}{\sqrt{M}}}^{-1} \lrp{\frac{\sqrt{n} (h(\thetahat_{\etahat_{1}}) - \tau)}{\sigmahat_{\etahat_{1}}} - \frac{\sqrt{n} (h(\thetahat_{\etahat_{2}}) - \tau)}{\sigmahat_{\etahat_{2}}}} \\
        & = \hat{\sigma}_{D}^{-1} \sqrt{M} \lrp{\frac{h(\thetahat_{\etahat_{1}}) - \tau}{\sigmahat_{\etahat_{1}}} - \frac{ h(\thetahat_{\etahat_{2}}) - \tau}{\sigmahat_{\etahat_{2}}}} \\
        & = \hat{\sigma}_{D}^{-1} \sqrt{M} \frac{h(\thetahat_{\etahat_{1}}) - h(\thetahat_{\etahat_{2}})}{\sigmahat_{\etahat_{1}}} - \hat{\sigma}_{D}^{-1} \sqrt{M} (\sigmahat_{\etahat_{1}} - \sigmahat_{\etahat_{2}}) \frac{h(\thetahat_{\etahat_{2}}) - \tau}{\sigmahat_{\etahat_{1}} \sigmahat_{\etahat_{2}}} \\
        & = \sigma_{D}^{-1} \sqrt{M} \frac{h(\thetahat_{\etahat_{1}}) - h(\thetahat_{\etahat_{2}})}{\sigma_{\etastarP}} - \sigma_{D}^{-1} \sqrt{M} (\sigmahat_{\etahat_{1}} - \sigmahat_{\etahat_{2}}) \frac{h(\theta_{\etastarP}) - \tau}{\sigma_{\etastarP}^2} + o_P(1) \\
        & = - \sigma_{\etastarP}^{-1} \hdot(\theta_{\etabar}) \dot{G}^{-1}_{\etabar} \sigma_{D}^{-1} \sqrt{M} \lrp{G_{\etahat_{1}}(\theta_{\etabar}) - G_{\etahat_{2}}(\theta_{\etabar})} + o_P \lrp{\sigma_{D}^{-1} v_{D}} \\
        & \quad - 2^{-1} \sigma_{\etastarP}^{-3} (h(\theta_{\etastarP}) - \tau) V_{M, K} \hdot(\theta_\etastarP) \dot{\Psi}_{\etastarP}^{-1} \frac{\sqrt{M}}{\sigma_{D}} \lrp{\cV_{\etahat_1}(\thetahat_{\etahat_1}) - \cV_{\etahat_2}(\thetahat_{\etahat_2})} \lrp{\dot{\Psi}_{\etastarP}^{-1}}^{T} \hdot(\theta_{\etastarP})^{T} \\
        & \quad + o_P \lrp{\sigma_{D}^{-1} \zeta_D} + o_P(1) \\
        & \leadsto \cN(0, 1),
    \end{align*}
    conditional on $D$ with probability approaching one, by Lindeberg's CLT, by definition of $\sigma_{D}$, and since $G_{\etahat_{1}}(\theta_{\etabar}) \perp G_{\etahat_{2}}(\theta_{\etabar}), \cV_{\etahat_2}(\thetahat_{\etahat_2})$ and $\cV_{\etahat_1}(\thetahat_{\etahat_1}) \perp G_{\etahat_{2}}(\theta_{\etabar}), \cV_{\etahat_2}(\thetahat_{\etahat_2})$ conditional on $D$. 
    Note that $o_P \lrp{\sigma_{D}^{-1} v_{D}}, o_P \lrp{\sigma_{D}^{-1} \zeta_D} = o_P(1)$ by \Cref{lemma.sigma_Op}. 
\end{proof}

\begin{proof}[\hypertarget{proof.th.reproducibility_pvalue}{Proof of \Cref{th.reproducibility_pvalue}}] 

For $(p_j, \hat{\delta}(\beta)) = (p_j^+, \hat{\delta}^+(\beta))$,
\begin{align*}
    & P \lrp{p_2^+ > p_1^+ + \hat{\delta}^+(\beta) \Biggm| D} \\
    & = P \lrp{\Phi \lrp{\frac{\sqrt{n} (h(\thetahat_{\etahat_2}) - \tau)}{\sigmahat_{\etahat_2}}} > \Phi \lrp{\frac{\sqrt{n} (h(\thetahat_{\etahat_1}) - \tau)}{\sigmahat_{\etahat_1}}} + \hat{\delta}^+(\beta) \Biggm| D} \\
    & = P \lrp{\Phi \lrp{\frac{\sqrt{n} (h(\thetahat_{\etahat_2}) - \tau)}{\sigmahat_{\etahat_2}}} > \Phi \lrp{\frac{\sqrt{n} (h(\thetahat_{\etahat_1}) - \tau)}{\sigmahat_{\etahat_1}} - \frac{\sqrt{n} \hat{\sigma}_{D}}{\sqrt{M}} \Phi^{-1}(\beta)} \Biggm| D} \\
    & = P \lrp{\frac{\sqrt{n} (h(\thetahat_{\etahat_2}) - \tau)}{\sigmahat_{\etahat_2}} - \frac{\sqrt{n} (h(\thetahat_{\etahat_1}) - \tau)}{\sigmahat_{\etahat_1}} > - \frac{\sqrt{n} \hat{\sigma}_{D}}{\sqrt{M}} \Phi^{-1}(\beta) \Biggm| D} \\
    & = P \lrp{\lrp{\frac{\sqrt{n} \hat{\sigma}_{D}}{\sqrt{M}}}^{-1} \lrp{\frac{\sqrt{n} (h(\thetahat_{\etahat_1}) - \tau)}{\sigmahat_{\etahat_1}} - \frac{\sqrt{n} (h(\thetahat_{\etahat_2}) - \tau)}{\sigmahat_{\etahat_2}}} < \Phi^{-1}(\beta) \Biggm| D} \\
    & = \beta + o_P(1),
\end{align*}
where the last equality follows from \Cref{th.reproducibility_clt}. 

For $(p_j, \hat{\delta}(\beta)) = (p_j^-, \hat{\delta}^-(\beta))$,
\begin{align*}
    & P \lrp{p_2^- > p_1^- + \hat{\delta}^-(\beta) \Biggm| D} \\
    & = P \lrp{\Phi \lrp{-\frac{\sqrt{n} (h(\thetahat_{\etahat_2}) - \tau)}{\sigmahat_{\etahat_2}}} > \Phi \lrp{-\frac{\sqrt{n} (h(\thetahat_{\etahat_1}) - \tau)}{\sigmahat_{\etahat_1}}} + \hat{\delta}^-(\beta) \Biggm| D} \\
    & = P \lrp{\Phi \lrp{-\frac{\sqrt{n} (h(\thetahat_{\etahat_2}) - \tau)}{\sigmahat_{\etahat_2}}} > \Phi \lrp{-\frac{\sqrt{n} (h(\thetahat_{\etahat_1}) - \tau)}{\sigmahat_{\etahat_1}} - \frac{\sqrt{n} \hat{\sigma}_{D}}{\sqrt{M}} \Phi^{-1}(\beta)} \Biggm| D} \\
    & = P \lrp{\frac{\sqrt{n} (h(\thetahat_{\etahat_1}) - \tau)}{\sigmahat_{\etahat_1}} - \frac{\sqrt{n} (h(\thetahat_{\etahat_2}) - \tau)}{\sigmahat_{\etahat_2}} > - \frac{\sqrt{n} \hat{\sigma}_{D}}{\sqrt{M}} \Phi^{-1}(\beta) \Biggm| D} \\
    & = P \lrp{\lrp{\frac{\sqrt{n} \hat{\sigma}_{D}}{\sqrt{M}}}^{-1} \lrp{\frac{\sqrt{n} (h(\thetahat_{\etahat_1}) - \tau)}{\sigmahat_{\etahat_1}} - \frac{\sqrt{n} (h(\thetahat_{\etahat_2}) - \tau)}{\sigmahat_{\etahat_2}}} > -\Phi^{-1}(\beta) \Biggm| D} \\
    & = 1 - \Phi \lrp{-\Phi^{-1}(\beta)} + o_P(1) \\
    & = \beta + o_P(1).
\end{align*}

For $(p_j, \hat{\delta}(\beta)) = (p_j^\pm, \hat{\delta}^\pm(\beta))$, 
\begin{align*}
    & P \lrp{p_2^\pm > p_1^\pm + \hat{\delta}^\pm(\beta) \Biggm| D} \\
    & = P \lrp{2 \Phi \lrp{-\lrm{\frac{\sqrt{n} (h(\thetahat_{\etahat_2}) - \tau)}{\sigmahat_{\etahat_2}}}} > 2 \Phi \lrp{-\lrm{\frac{\sqrt{n} (h(\thetahat_{\etahat_1}) - \tau)}{\sigmahat_{\etahat_1}}}} + \hat{\delta}^\pm(\beta) \Biggm| D} \\
    & = P \lrp{2 \Phi \lrp{-\lrm{\frac{\sqrt{n} (h(\thetahat_{\etahat_2}) - \tau)}{\sigmahat_{\etahat_2}}}} > 2 \Phi \lrp{-\lrm{\frac{\sqrt{n} (h(\thetahat_{\etahat_1}) - \tau)}{\sigmahat_{\etahat_1}}} - \frac{\sqrt{n} \hat{\sigma}_{D}}{\sqrt{M}} \Phi^{-1}(\beta/2)} \Biggm| D} \\
    & = P \lrp{\lrm{\frac{\sqrt{n} (h(\thetahat_{\etahat_1}) - \tau)}{\sigmahat_{\etahat_1}}} - \lrm{\frac{\sqrt{n} (h(\thetahat_{\etahat_2}) - \tau)}{\sigmahat_{\etahat_2}}} > - \frac{\sqrt{n} \hat{\sigma}_{D}}{\sqrt{M}} \Phi^{-1}(\beta/2) \Biggm| D} \\
    & \le P \lrp{\lrm{\frac{\sqrt{n} (h(\thetahat_{\etahat_2}) - \tau)}{\sigmahat_{\etahat_2}} - \frac{\sqrt{n} (h(\thetahat_{\etahat_1}) - \tau)}{\sigmahat_{\etahat_1}}} > - \frac{\sqrt{n} \hat{\sigma}_{D}}{\sqrt{M}} \Phi^{-1}(\beta/2) \Biggm| D} \\
    & = 2 P \lrp{\frac{\sqrt{n} (h(\thetahat_{\etahat_2}) - \tau)}{\sigmahat_{\etahat_2}} - \frac{\sqrt{n} (h(\thetahat_{\etahat_1}) - \tau)}{\sigmahat_{\etahat_1}} < \frac{\sqrt{n} \hat{\sigma}_{D}}{\sqrt{M}} \Phi^{-1}(\beta/2) \Biggm| D} \\
    & = 2 P \lrp{\lrp{\frac{\sqrt{n} \hat{\sigma}_{D}}{\sqrt{M}}}^{-1} \lrp{\frac{\sqrt{n} (h(\thetahat_{\etahat_2}) - \tau)}{\sigmahat_{\etahat_2}} - \frac{\sqrt{n} (h(\thetahat_{\etahat_1}) - \tau)}{\sigmahat_{\etahat_1}}} < \Phi^{-1}(\beta/2) \Biggm| D} \\
    & = 2 \Phi \lrp{\Phi^{-1}(\beta/2)} + o_P(1) \\
    & = \beta + o_P(1).
\end{align*}

\end{proof}

\begin{proof}[\hypertarget{proof.th.reproducibility_m_fast}{Proof of \Cref{th.reproducibility_m_fast}}] 

    The first result follows since, from the proof of \Cref{th.reproducibility_clt}, 
    $$\lrp{\frac{\sqrt{n} (h(\thetahat_{\etahat_{1}}) - \tau)}{\sigmahat_{\etahat_{1}}} - \frac{\sqrt{n} (h(\thetahat_{\etahat_{2}}) - \tau)}{\sigmahat_{\etahat_{2}}}} = O_P(1).$$

    For $(p_j, \hat{\delta}(\beta)) = (p_j^+, \hat{\delta}^+(\beta))$, from the proof of \Cref{th.reproducibility_pvalue}, 
    \begin{align*}
        & P \lrp{p_2^+ > p_1^+ + \hat{\delta}^+(\beta) \Biggm| D} \\
        & = P \lrp{\frac{\sqrt{n} (h(\thetahat_{\etahat_1}) - \tau)}{\sigmahat_{\etahat_1}} - \frac{\sqrt{n} (h(\thetahat_{\etahat_2}) - \tau)}{\sigmahat_{\etahat_2}} < \lrp{\frac{\sqrt{n} \hat{\sigma}_{D}}{\sqrt{M}}} \Phi^{-1}(\beta) \Biggm| D},
    \end{align*}
    which converges to zero since 
    $$\frac{\sqrt{n} (h(\thetahat_{\etahat_1}) - \tau)}{\sigmahat_{\etahat_1}} - \frac{\sqrt{n} (h(\thetahat_{\etahat_2}) - \tau)}{\sigmahat_{\etahat_2}} = O_P(1)$$
    from \Cref{th.reproducibility_clt}, and 
    $$\lrp{\frac{\sqrt{n} \hat{\sigma}_{D}}{\sqrt{M}}} \Phi^{-1}(\beta) \Pto - \infty$$
    since $\Phi^{-1}(\beta) < 0$. 
    Analogous results follow for $(p_j^-, \hat{\delta}^-(\beta))$ and $(p_j^\pm, \hat{\delta}^\pm(\beta))$. 
\end{proof}

%% file: appendix_ghana.tex
\subsection{Details of Section \ref{section.application_ghana}} \label{appendix.ghana}

\subsubsection{Covariates Description}

The following variables from the Ghana Socioeconomic Panel Survey are used as predictive covariates for poverty prediction in \Cref{section.application_ghana}:

\paragraph{Household Demographics}
\begin{itemize}
    \item \texttt{children}: Number of children in household
    \item \texttt{adults}: Number of adults in household
    \item \texttt{female\_head}: Indicator for female household head
    \item \texttt{married\_head}: Indicator for married household head
    \item \texttt{spouse\_in}: Indicator for spouse living in the household
\end{itemize}

\paragraph{Religion}
\begin{itemize}
    \item \texttt{christian}: Proportion Christian
    \item \texttt{muslim}: Proportion Muslim
    \item \texttt{traditional}: Proportion traditional religion
\end{itemize}

\paragraph{Political and Traditional Leadership}
\begin{itemize}
    \item \texttt{ever\_political\_office}: Indicator for ever holding political office
    \item \texttt{today\_political\_office}: Indicator for currently holding political office
    \item \texttt{ever\_traditional\_office}: Indicator for ever holding traditional office
    \item \texttt{today\_traditional\_office}: Indicator for currently holding traditional office
\end{itemize}

\paragraph{Parental Education}
\begin{itemize}
    \item \texttt{father\_primary}: Indicator for father completed primary education
    \item \texttt{father\_middle}: Indicator for father completed middle school
    \item \texttt{father\_secondary}: Indicator for father completed secondary education
    \item \texttt{father\_tertiary}: Indicator for father completed tertiary education
    \item \texttt{mother\_primary}: Indicator for mother completed primary education
    \item \texttt{mother\_middle}: Indicator for mother completed middle school
    \item \texttt{mother\_secondary}: Indicator for mother completed secondary education
    \item \texttt{mother\_tertiary}: Indicator for mother completed tertiary education
\end{itemize}

\paragraph{Asset Holdings}
\begin{itemize}
    \item \texttt{plot\_acreage}: Total land holdings in acres
    \item \texttt{livestock\_value}: Total value of livestock
    \item \texttt{livestock\_expenses}: Annual livestock maintenance expenses
\end{itemize}

\paragraph{Financial Resources}
\begin{itemize}
    \item \texttt{health\_insurance}: Proportion of household members covered by health insurance
    \item \texttt{savings\_home}: Amount of savings kept at home
    \item \texttt{d\_saving\_bank}: Distance to nearest bank (in km)
    \item \texttt{savings\_bank}: Amount of savings in bank account
\end{itemize}

\subsubsection{Fraction Per Tercile as a Z-Estimator}

For a given split $\s$, the vector 
$$\lrp{\lrp{\frac{\sum_{i \in \s} Y_i \I{\hat{t}_{j - 1,\stilde} < \etahat_{\stilde}(X_i) \le \hat{t}_{j,\stilde}}}{\sum_{i \in \s} \I{\hat{t}_{j - 1,\stilde} < \etahat_{\stilde}(X_i) \le \hat{t}_{j,\stilde}}}}_{j=1}^3, (\hat{t}_{j,\stilde})_{j=1}^2}^T$$
is a Z-estimator with the moment functions 
$$\psi_{(\theta, t), \eta}(y, x) = \begin{pmatrix}
y \I{t_0 < \eta(x) \le t_1} - \theta_1 \I{t_0 < \eta(x) \le t_1} \\
y \I{t_1 < \eta(x) \le t_2} - \theta_2 \I{t_1 < \eta(x) \le t_2} \\
y \I{t_2 < \eta(x) \le t_3} - \theta_3 \I{t_2 < \eta(x) \le t_3} \\
\I{\eta(x) \le t_1} - \frac{1}{3} \\
\I{\eta(x) \le t_2} - \frac{2}{3}
\end{pmatrix}.$$
Hence, the final estimators $\thetahat_{\etahat, {\rm Frac} j}$ are averages over split-specific estimators as in \cref{eq.thetahat_z_1}. 

Note that the conditions in \Cref{th.clt_z} are met whenever $\etastarP(x)$ is not flat in $x$.
This condition is testable, for example using the one-sided test for the accuracy in \Cref{fig.mse_combined}.

\subsubsection{Monte Carlo Designs}

I simulate outcome and covariates by (i) converting each observed column to rank-based uniforms $U=\text{rank}(X)/(n+1)$, (ii) Gaussianizing to $Z=\Phi^{-1}(U)$ and estimating the latent normal correlation $\Sigma^*$, (iii) drawing $Z^\ast\sim\mathcal N(0,\Sigma^*)$ and mapping back to uniforms $U^\ast=\Phi(Z^\ast)$, and (iv) inverting each margin with the empirical CDF of the corresponding variable.
For the correlated design, I modify $\Sigma^*$ by multiplying by 3 the first row/column, the one corresponding to the correlation between outcome and covariates, and use as correlation matrix its nearest positive definite matrix in case the modified $\Sigma^*$ is no longer positive definite. 
For the uncorrelated design, I sample covariates the same way, and the outcome is sampled independently from a binomial distribution with probability $0.07$. 

\subsubsection{Comparison of Top-Bottom Estimates}

\Cref{fig.tmb_combined} compares the top minus bottom estimates across datasets and methods, similar to \Cref{fig.mse_combined}. 

\begin{figure}[!ht]
\includegraphics[width=\textwidth]{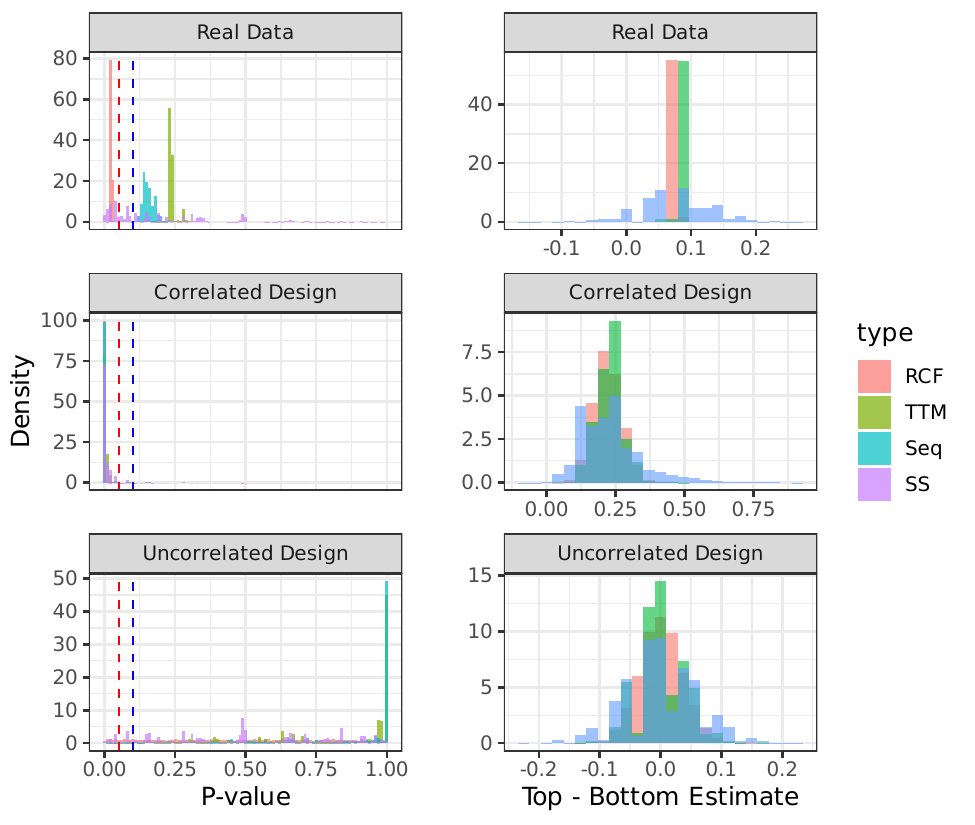}
\caption{Comparison of Top-Bottom Estimates Across Methods and Datasets}
\label{fig.tmb_combined}
\vspace{0.3cm}
\small
\textbf{Notes:} Left panels show distribution across Monte Carlo iterations of p-values for testing whether the top tercile has a higher fraction below the poverty line than the bottom tercile. 
Vertical red and blue lines are respectively 0.05 and 0.10. 
Right panels show distribution of point estimates for the difference between top and bottom terciles. 
Rows show results for real data (top), simulations from correlated design (middle), and simulations from uncorrelated design (bottom). 
Methods: RCF (repeated cross-fitting), TTM (twice-the-median), Seq (sequential aggregation), SS (sample-splitting).
\end{figure}

%% file: appendix_hte.tex
\subsection{\texorpdfstring{Details of \Cref{section.application_hte}}{Details of Section~\ref{section.application_hte}}} \label{appendix.hte}

\subsubsection{Covariates Description}

\textbf{Donation History Variables:}
\begin{itemize}
    \item \texttt{hpa}: Highest previous contribution
    \item \texttt{freq}: Number of prior donations
    \item \texttt{years}: Number of years since initial donation
    \item \texttt{mrm2}: Number of months since last donation
\end{itemize}

\textbf{Individual Demographics:}
\begin{itemize}
    \item \texttt{female}: Female indicator
\end{itemize}

\textbf{State-Level Political Variables:}
\begin{itemize}
    \item \texttt{cases}: Count of court cases between 2002 and 2005 in which the organization was either a party to or filed a brief
    \item \texttt{perbush}: State vote share for Bush
    \item \texttt{nonlit}: Count of incidences relevant to this organization from each state reported in 2004-5 (values range from zero to six) in the organization's monthly newsletter to donors
\end{itemize}

\textbf{Zip Code Demographics and Economics:}
\begin{itemize}
    \item \texttt{pwhite}: Proportion white within zip code
    \item \texttt{pblack}: Proportion black within zip code
    \item \texttt{page18\_39}: Proportion age 18-39 within zip code
    \item \texttt{ave\_hh\_sz}: Average household size within zip code
    \item \texttt{median\_hhincome}: Median household income within zip code
    \item \texttt{powner}: Proportion house owner within zip code
    \item \texttt{psch\_atlstba}: Proportion who finished college within zip code
    \item \texttt{pop\_propurban}: Proportion of population urban within zip code
\end{itemize}

\subsubsection{Monte Carlo Designs}

The designs are explicitly calibrated to the observed data so that simulated covariates and outcomes are distributionally aligned with the original sample. 

\emph{Treatment assignment.} 
I draw the treatment assignment indicator from a Bernoulli distribution with mean $0.5$. 

\emph{Covariates and potential outcome under control.} 
Starting from the observed outcome and covariate matrix for the control sample, I form pseudo-uniforms for each column by ranking within sample and scaling, \(U=\mathrm{rank}(X)/(n{+}1)\). 
I then Gaussianize to \(Z=\Phi^{-1}(U)\) and estimate the latent normal correlation \(\Sigma^*\) on these \(Z\) (taking the nearest positive definite matrix if needed). 
To generate synthetic $Y(0)$ and covariates, I draw \(Z^\ast\sim\mathcal N(0,\Sigma^*)\), map to uniforms \(U^\ast=\Phi(Z^\ast)\), and invert each margin via the empirical CDF of the corresponding original variable. 

\emph{Treatment effect.} 
From the original data, I estimate two arm-specific components as functions of treatment and covariates. 
The first is a logistic regression for whether $Y=0$ (no donation), using treatment, covariates and their interactions. 
The second is a Poisson regression, with amount of donation as outcome and same variables in the model. 
For generating simulated observations, the treatment effect is zero with probability $q_0(x,y_0) - q_1(x,y_0)$ (rounded to zero or one if necessary), where 
$$q_d(x,y_0)=(1-\pi_d(x))\;\hat{P}(Y \ge y_0{+}1\mid X{=}x,D{=}d),$$
with $x$ being the covariates, $y_0$ the value of potential outcome under control, $\pi_d(x)$ the probability that $Y = 0$ coming from the logit model with coefficients associated with treatment $ = 1$ being multiplied by 4, and $\hat{P}(Y \ge y_0{+}1\mid X{=}x,D{=}d)$ coming from the Poisson model with mean multiplied by $0.05$. 
Conditional on the treatment effect being different from zero, I draw $Y(1)$ from a truncated Poisson distribution starting at $Y(0)$ with the same mean coming from the Poisson regression. 

\emph{Final outcome.} 
For the design where treatment effect heterogeneity is predictable, I generate the observed outcome as $Y(1)$ if treatment is $1$, and $Y(0)$ otherwise. 
For the design where treatment effect heterogeneity is not predictable, I generate the entire dataset exactly the same way, but shuffle the treatment assignment indicator at random as the last step. 

\subsubsection{Additional Figures and Table}

\Cref{fig.hte_real_shuffle} displays results with the real dataset with shuffled treatment indicator (at random, so treatment effect is constant and equal to zero), and \Cref{fig.hte_fake_hte} displays results for the synthetic DGP where there is explainable treatment effect heterogeneity. 
\Cref{tab.iterations} gives the number of Monte Carlo iterations used for each specification. 

\begin{figure}[!ht]
\includegraphics[width=\textwidth]{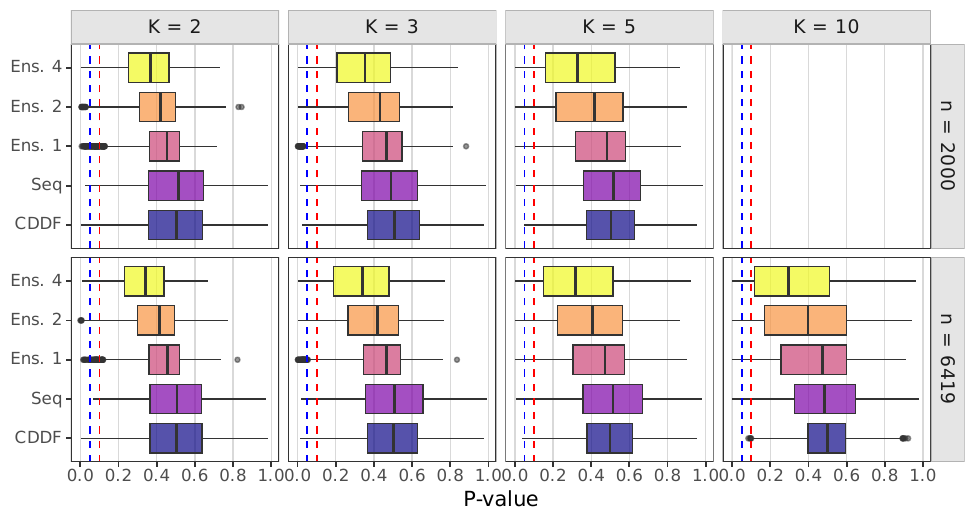}
\caption{Distribution of p-values for Top - Bottom GATES Groups -- Real Data with Shuffled Treatment Assignment}
\label{fig.hte_real_shuffle}
\vspace{0.3cm}
\footnotesize
Notes: Distribution of one-sided p-values for testing whether the top tercile has a larger ATE than the bottom tercile across Monte Carlo iterations using the real dataset. 
Rows show different sample sizes ($n = 2000, 6419$), columns show different numbers of folds ($K = 2, 3, 5, 10$). 
Ens. 1, Ens. 2, and Ens. 4 represent the Ensemble method using respectively 1, 2, and 4 algorithms. 
Each box represents the distribution across Monte Carlo iterations with 100 repetitions of sample-splitting per iteration. 
Boxplots show the median (center line), interquartile range (box), and whiskers extending to 1.5 times the IQR, with points beyond shown as outliers. 
Sources of randomness are the subsample when $n=2000$, which ML algorithms are used, how the data are split, and how the treatment assignment indicator is shuffled. 
Red dashed line at 0.1, blue dashed line at 0.05. 
Specifications with $K = 10, n = 2000$ are excluded.
\end{figure}

\begin{figure}[!ht]
\includegraphics[width=\textwidth]{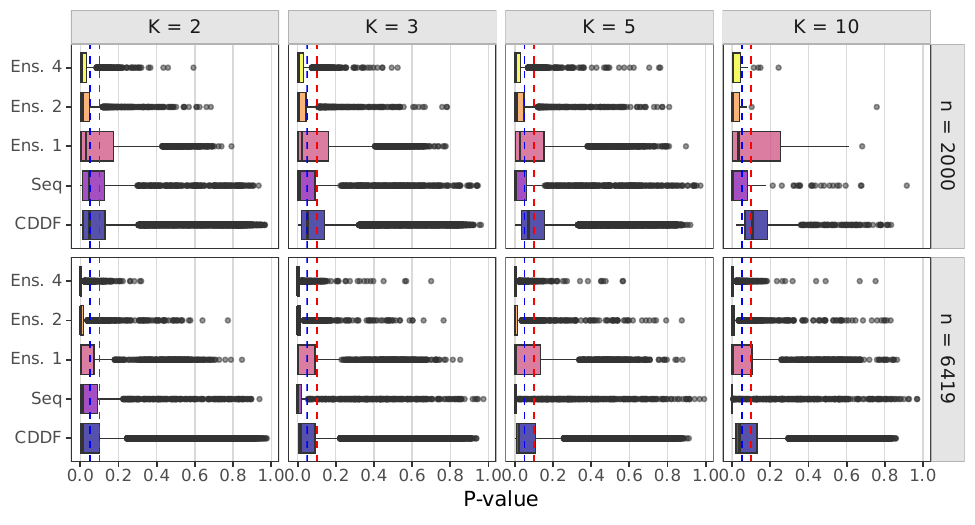}
\caption{Distribution of p-values for Top - Bottom GATES Groups -- Synthetic DGP with Heterogeneity}
\label{fig.hte_fake_hte}
\vspace{0.3cm}
\footnotesize
Notes: Distribution of one-sided p-values for testing whether the top tercile has a larger ATE than the bottom tercile across Monte Carlo iterations using the real dataset. 
Rows show different sample sizes ($n = 2000, 6419$), columns show different numbers of folds ($K = 2, 3, 5, 10$). 
Ens. 1, Ens. 2, and Ens. 4 represent the Ensemble method using respectively 1, 2, and 4 algorithms. 
Each box represents the distribution across Monte Carlo iterations with 100 repetitions of sample-splitting per iteration. 
Boxplots show the median (center line), interquartile range (box), and whiskers extending to 1.5 times the IQR, with points beyond shown as outliers. 
Data is generated from a synthetic DGP where there is explainable treatment effect heterogeneity (\Cref{appendix.hte}). 
Red dashed line at 0.1, blue dashed line at 0.05. 
Specifications with $K = 10, n = 2000$ are excluded.
\end{figure}

\begin{table}[ht]
\centering
\caption{Number of Monte Carlo Iterations by Specification}
\label{tab.iterations}
\scriptsize
\input{iterations_table.tex}
\end{table}

\subsubsection{Theoretical Properties of Ensemble Approach}

I establish the theoretical properties of the ensemble estimator using the CLTs proven in this paper. 
I show that when there is detectable heterogeneity, i.e., when the ensemble weights $(\betahat_a)_{a=1}^A$ do not converge to zero, the confidence interval based on the normal approximation is asymptotically exact. 
If there is no detectable heterogeneity, however, my theoretical result gives no coverage guarantee to the normal approximation CI. 
Extensive simulation exercises, including but not limited to those of \Cref{section.application_hte}, suggest that the normal approximation CI is actually conservative under the null hypothesis of no heterogeneity for small values of $A$ and $K$ such as $A=4$ and $K=3$. 
Hence, my recommendation for empirical practice is to use the normal approximation CI with no more than 4 algorithms and 5 folds. 
I also propose an adaptive approach using ideas developed in \Cref{section.diff_performance} that is valid even when there is no detectable heterogeneity, at the cost of having smaller power.

First, I introduce additional notation. 
Denote the set of splits 
$$\cS = \lrp{\s_{m,k}}_{m \in \lrbk{M}, k \in \lrbk{K}},$$
and the set of model $\etahat = (\etahat_{\stilde})_{\s \in \cS}$. 
I use $F_{P}(x)$ to denote the cdf of the random variable $\sum_{a=1}^{A} {\betastarP}_a {\etastarP}_a(X)$ and 
$$F_{P}^{-1}(p) = \inf \lrbc{x \in \cX : p \le F_{P}(x)}.$$
For some results, I focus on a set $\cP_{hte} \subseteq \cP$ such that $\lrp{F_{P}^{-1}(t)}_{P \in \cP_{hte}}$ is equicontinuous at points $t=j/J$ for $j =1,\dots,J$. 
This is a collection of DGPs where the $J$ quantiles of the limit predicted ITE $\sum_{a=1}^{A} {\betastarP}_a {\etastarP}_a(X)$ are well-defined. 
This is required so that the groups defined in \cref{eq.def_gates_groups_ensemble} are well-defined in the limit. 
Note that $F_{P}^{-1}(j/J)$ being continuous implies that the limit predictor $\sum_{a=1}^{A} {\betastarP}_a {\etastarP}_a(X)$ is not flat in $X$, so this class essentially excludes DGPs where there is no detectable heterogeneity, that is, where the true CATE $\eta_{P}(x)$ is flat in $x$. 

My first result is that the normal approximation CI is asymptotically exact when there is detectable heterogeneity. 
It relies on \Cref{as.ensemble}, defined in \Cref{section.proofs_section.ensemble}. 
It is a mild but technical assumption that requires: (i) the weights $\betahat_{\ell,a}$ have finite limits, (ii) a standard moments condition, (iii) propensity scores are bounded away from $0$ and $1$, (iv) the variance-covariance matrix of the regressors $Z$ is positive definite, and (v) the models estimated with ML converge to any limit at any rate. 

\begin{theorem} \label{th.ensemble_hte}
    Let \Cref{as.ensemble} hold, and let $\cP_{hte} \subseteq \cP$ be such that $\lrp{F_{P}^{-1}(t)}_{P \in \cP_{hte}}$ is equicontinuous at points $t=j/J$ for $j =1,\dots,J$. 
    Then, for any sequence $(P_n)_{n \ge 1} \subseteq \cP_{hte}$, 
    $$P_n \lrp{\delta_\etahat \in \lrbk{\hat{\delta}_\etahat - z_{1-\alpha/2} \sigmahat_{\etahat}, \hat{\delta}_\etahat + z_{1-\alpha/2} \sigmahat_{\etahat}}} \to 1 - \alpha.$$
\end{theorem}

Although \Cref{th.ensemble_hte} does not cover cases when there is no detectable heterogeneity, extensive simulation exercises, including but not limited to the ones of \Cref{section.application_hte}, suggest that the coverage probability is larger than $1 - \alpha$ in those cases at least when $A \le 4$, $K \le 5$, that is, the CI \Cref{th.ensemble_hte} is conservative. 
Next, I consider a test for detectable heterogeneity that can be used, for example, when $A > 4$ and/or $K > 5$. 
If the test rejects no detectable heterogeneity, the normal approximation CI may be used. 

\subsubsection{A Test for Detectable Heterogeneity} \label{appendix.test_hte}

I propose using a version of the test proposed in \Cref{section.diff_onesided_test} for testing whether the models $\etahat = (\etahat_{\stilde_{m,k}})$ have explanatory power for heterogeneous treatment effects. 
Specifically, I first calculate the mean squared of residuals from the BLP regression 
\begin{equation} \label{eq.ensemble_msr}
    Y_i = \alpha_1 + \sum_{a = 1}^A \beta_a (\etahat_{\stilde, a}(X_i) - \taubar_{\s,a}) \lrbk{T_i - p(X_i)} + \alpha_2 Z_i + \varepsilon_i, \qquad i \in \s
\end{equation}
with weights $\omega_i = \lrbc{p(X_i) \lrbk{1 - p(X_i)}}^{-1}$, $\taubar_{\s,a} = |s|^{-1} \sum_{i \in \s} \etahat_{\stilde, a}(X_i)$, for $\s \in \cS$, as in \cref{eq.ensemble_weights} but at the fold level. 
Denote it by  
$$MSR_\s = \frac{1}{|s|} \sum_{i \in \s} \lrp{Y_i - \hat{\alpha}_{1,\s} + \sum_{a = 1}^A \hat{\beta}_{a,\s} (\etahat_{\stilde, a}(X_i) - \taubar_{\s,a}) \lrbk{T_i - p(X_i)} + \hat{\alpha}_{2,\s} Z_i}^2$$
I compare $(MSR_\s)_{\s \in \cS}$ with 
$$MSR_{b} = \frac{1}{n} \sum_{i=1}^{n} \lrp{Y_i - \hat{\alpha}_{1,b} + \hat{\alpha}_{2,b} Z_i}^2,$$
where $\hat{\alpha}_{1,b}$ and $\hat{\alpha}_{2,b}$ are the estimates from the weighted least squares regression 
$$Y_i = \alpha_{1} + \alpha_{2} Z_i + \varepsilon_i, \qquad i \in {1,\dots,n}.$$
Let $\Sigmahat$ be an estimate of the asymptotic variance of $\sqrt{n} \lrp{MSR_\s - MSR_b}_{\s \in \cS}$, and $\sigmahat^2_\s$ are the entries of the main diagonal.  
I propose calculating the test-statistic 
$$\hat{T} = \sum_{\s \in \cS} \lrp{\min \lrbc{\sqrt{n} \frac{MSR_\s - MSR_b}{\sigmahat_{\s}}, 0}}^2.$$
I establish the validity of this test in \Cref{th.ensemble_onesided_test}, where $\hat{c}_{1-\alpha}$ is calculated as in \Cref{section.diff_onesided_test}. 
The result follows from \Cref{th.diff_onesided_test}. 

\begin{theorem} \label{th.ensemble_onesided_test}
    Let \Cref{as.ensemble} hold, and let 
    $$\cP_0 = \lrbc{P \in \cP : \exists c \in \R, \eta_{P}(x) = c}.$$
    Then, for any sequence $(P_n)_{n \ge 1} \subseteq \cP_{0}$, 
    $$P_n \lrp{\hat{T} > \hat{c}_{1-\alpha}} \to 1 - \alpha.$$
\end{theorem}

Denote the normal approximation CI 
$$\widehat{\rm CI}_{\alpha,\cN} = \lrbk{\hat{\delta}_\etahat - z_{1-\alpha/2} \sigmahat_{\etahat}, \hat{\delta}_\etahat + z_{1-\alpha/2} \sigmahat_{\etahat}},$$
and the extended CI 
$$\widehat{\rm CI}_{\alpha,{\rm ext}} = \operatorname{Conv}\lrp{\widehat{\rm CI}_{\alpha,\cN} \cup \{0\}},$$
where $\operatorname{Conv}$ denotes the convex hull, that is, $\widehat{\rm CI}_{\alpha,{\rm ext}}$ has all the elements in $\widehat{\rm CI}_{\alpha,\cN}$, $0$, and all elements in between. 
For a given fixed $\bar{c}_5 \ge 0$, denote the final CI 
$$\CIalpha = \begin{cases}
    \widehat{\rm CI}_{\alpha,\cN}, \text{ if } \hat{T}' > \hat{c}_{1 - \alpha} \\
    \widehat{\rm CI}_{\alpha,{\rm ext}}, \text{ otherwise}, 
\end{cases}$$
\Cref{th.ensemble_onesided_test} implies that this CI is asymptotically valid pointwise in $P \in \cP$ for $\bar{c}_5 = 0$, and uniformly in $P \in \cP$ for any $\bar{c}_5 > 0$.

\input{proofs_ensemble}

%% file: iterations_table.tex
\begin{tabular*}{\linewidth}{@{\extracolsep{\fill}}llccccccccccc}
\toprule
 & data & \multicolumn{2}{c}{n500} & \multicolumn{2}{c}{n1000} & \multicolumn{3}{c}{n2000} & \multicolumn{4}{c}{n6419} \\ 
\cmidrule(lr){2-2} \cmidrule(lr){3-4} \cmidrule(lr){5-6} \cmidrule(lr){7-9} \cmidrule(lr){10-13}
Method & Data Type & K2 & K3 & K2 & K3 & K2 & K3 & K5 & K2 & K3 & K5 & K10 \\ 
\midrule\addlinespace[2.5pt]
CDDF & Real (Shuffled) & 26,175 & 6,990 & 27,585 & 20,976 & 16,850 & 17,050 & 17,125 & 13,784 & 14,029 & 13,393 & 11,632 \\ 
CDDF & MC: No HTE & 26,109 & 26,320 & 26,637 & 26,211 & 16,433 & 15,805 & 18,756 & 12,384 & 12,142 & 12,376 & 11,676 \\ 
CDDF & Real Data & 23,324 & 4,327 & 27,026 & 12,845 & 17,204 & 17,131 & 8,579 & 13,765 & 14,005 & 13,420 & 9,417 \\ 
CDDF & MC: With HTE & 21,283 & 18,018 & 27,783 & 25,352 & 17,303 & 17,120 & 18,073 & 13,756 & 13,545 & 13,900 & 11,794 \\ 
Seq & Real (Shuffled) & 644 & 192 & 670 & 502 & 1,316 & 1,164 & 1,040 & 1,054 & 1,092 & 1,338 & 446 \\ 
Seq & MC: No HTE & 668 & 700 & 760 & 570 & 1,182 & 1,246 & 1,454 & 1,230 & 1,150 & 1,296 & 1,214 \\ 
Seq & Real Data & 1,154 & 172 & 1,342 & 650 & 1,134 & 1,380 & 568 & 1,018 & 1,004 & 1,430 & 656 \\ 
Seq & MC: With HTE & 508 & 464 & 696 & 622 & 1,366 & 1,368 & 1,248 & 1,302 & 1,262 & 1,320 & 1,226 \\ 
Ens. 1 & Real (Shuffled) & 3,177 & 901 & 3,390 & 2,638 & 1,969 & 2,084 & 1,990 & 1,493 & 1,543 & 1,580 & 1,505 \\ 
Ens. 1 & MC: No HTE & 3,399 & 3,368 & 3,327 & 3,335 & 2,105 & 2,057 & 2,346 & 1,492 & 1,537 & 1,560 & 1,522 \\ 
Ens. 1 & Real Data & 2,871 & 491 & 3,372 & 1,572 & 2,038 & 2,014 & 1,032 & 1,493 & 1,614 & 1,530 & 1,268 \\ 
Ens. 1 & MC: With HTE & 2,744 & 2,229 & 3,206 & 3,132 & 2,004 & 2,040 & 2,091 & 1,549 & 1,569 & 1,563 & 1,549 \\ 
Ens. 2 & Real (Shuffled) & 3,183 & 841 & 3,433 & 2,664 & 2,096 & 2,078 & 1,974 & 1,543 & 1,552 & 1,571 & 1,485 \\ 
Ens. 2 & MC: No HTE & 3,370 & 3,409 & 3,417 & 3,340 & 2,124 & 2,003 & 2,374 & 1,582 & 1,538 & 1,544 & 1,455 \\ 
Ens. 2 & Real Data & 2,865 & 499 & 3,269 & 1,584 & 2,160 & 1,999 & 987 & 1,574 & 1,556 & 1,561 & 1,265 \\ 
Ens. 2 & MC: With HTE & 2,625 & 2,226 & 3,429 & 3,170 & 2,120 & 2,052 & 2,075 & 1,589 & 1,585 & 1,588 & 1,443 \\ 
Ens. 4 & Real (Shuffled) & 3,261 & 868 & 3,476 & 2,512 & 2,048 & 2,072 & 2,035 & 1,581 & 1,643 & 1,511 & 1,389 \\ 
Ens. 4 & MC: No HTE & 3,367 & 3,405 & 3,421 & 3,410 & 2,069 & 1,996 & 2,319 & 1,524 & 1,438 & 1,525 & 1,420 \\ 
Ens. 4 & Real Data & 2,876 & 546 & 3,375 & 1,614 & 2,052 & 2,089 & 991 & 1,567 & 1,575 & 1,506 & 1,114 \\ 
Ens. 4 & MC: With HTE & 2,569 & 2,196 & 3,451 & 3,081 & 2,050 & 2,073 & 2,116 & 1,591 & 1,547 & 1,614 & 1,484 \\ 
\bottomrule
\end{tabular*}

%% file: proofs_ensemble.tex
\subsubsection{Proofs and Extra Definitions} \label{section.proofs_section.ensemble}

Define $\cX \subseteq \R^{d_x}$ as the space that contains the covariates $X \in \cX$ for some integer $d_x > 0$. 
Let $Y^T = (Y_i)_{i=1}^n$. 
For any
$$d = (d_{j,k})_{j \in [J], k \in [K]},$$
$$\beta = (\beta_{\ell,a})_{\ell \in [L], a \in [A]},$$
and $\eta \in H$, 
let 
$$H_{\eta,\beta,d}^T = \lrp{Z_i, \lrbk{\lrbc{T_i - p(X_i)} \I{d_{j-1, k(i)} \le \sum_{a=1}^{A} \beta_{\ell(i),a} \eta_a(X_i) < d_{j, k(i)}}}_{j=1}^J}_{i=1}^n.$$
$$H_{\etahat,\beta,d}^T = \lrp{Z_i, \lrbk{\lrbc{T_i - p(X_i)} \I{d_{j-1, k(i)} \le \sum_{a=1}^{A} \beta_{\ell(i),a} \etahat_{\stilde_{k(i)},a}(X_i) < d_{j, k(i)}}}_{j=1}^J}_{i=1}^n.$$

$\Omega$ is the n-by-n diagonal matrix of weights:
$$\Omega = \operatorname{diag}(\omega_1,\dots,\omega_n).$$
${\betastarP}_{\ell,a}$ is the coefficient of the linear projection with weights $\omega$ of $Y$ on 
$$\lrbk{\lrbc{T_i - p(X_i)} (\etastarP_a(X) - \E[P]{\etastarP_a(X)})}_{j=1}^J$$
when that is well-defined, and zero otherwise. 
Note ${\betastarP}_{\ell,a}$ is the same for all $\ell$ since the limit $\etastarP_a$ does not depend on the data. 
Let $F_{P}(x)$ be the cdf of the random variable $\sum_{a=1}^{A} {\betastarP}_a {\etastarP}_a(X)$ and 
$$F_{P}^{-1}(p) = \inf \lrbc{x \in \cX : p \le F_{P}(x)}.$$
Define 
$${\dstarP}_{j,k} = F_{P}^{-1}(j/J).$$
Similarly, ${\dstarP}_{j,k}$ is the same for all $k$. 
$$\dhat = (\dhat_{j,k})_{j \in [J], k \in [K]},$$
$$\betahat = (\betahat_{\ell,a})_{\ell \in [L], a \in [A]}.$$
Define $\theta_{\eta,\beta,d}$ and column vector $\varepsilon_{\eta,\beta,d}$ such that 
\begin{equation} \label{eq.ensemble_epsilon}
    Y = H_{\eta,\beta,d} \theta_{\eta,\beta,d} + \varepsilon_{\eta,\beta,d}.
\end{equation}
$\thetahat_{\etahat,\betahat,\dhat}^{(m)} = (\hat{\alpha}, \lrp{\hat{\gamma}_{j=1}^J})^T$ are the estimates from \cref{eq.gates_reg_ensemble}, and  
$\theta_{\etahat,\betahat,\dhat}^{(m)}$ denotes $\etahat,\betahat,\dhat$ from the $m$-th repetition. 

\begin{assumption} \label{as.ensemble} 
    The following conditions hold: 
    \begin{assumptionenum}
        \item For some $B = (B_\beta \times B_d) \subset \R^{L A} \times \R^{J K}$ with compact $B_\beta$, 
        $$\bigcup_{P \in \cP} (\betastarP,\dstarP) \subseteq B;$$  
        \item For some $\bar{c}_6 > 0$, 
        $$\sup_{P \in \cP} \sup_{\eta \in H, (\beta, d) \in B} \E[P]{\lrm{H_{\eta,\beta,d,i}^T \varepsilon_{\eta,\beta,d,i}}^{2+\bar{c}_6}} < \infty;$$
        \item For some $\bar{c}_7 > 0$, and all $x \in \cX$, 
        $$\bar{c}_7 < p(x) < 1 - \bar{c}_7;$$
        \item $\inf_{P \in \cP} \det \lrp{\Var[P]{Z}} > 0$. 
        \item There exists $(\etastarP_a)_{a=1}^A$ such that 
        $$\E[P]{\lrm{\etatilde_a(X) - \etastarP_a(X)} | D} \Pto 0$$
        uniformly in $P \in \cP$, where $\etatilde_a = \cA_a(D)$ and $X \perp D$.  
    \end{assumptionenum}
\end{assumption}

\begin{theorem} \label{th.ensemble_decomp}
    Let \Cref{as.ensemble} hold, and $\cP_{hte} \subseteq \cP$ be such that $\lrp{F_{P}^{-1}(t)}_{P \in \cP_{hte}}$ is equicontinuous at points $t=j/J$ for $j =1,\dots,J$.
    Then, 
    $$\sqrt{n} \lrp{\thetahat_{\etahat,\betahat,\dhat} - \theta_{\etahat,\betahat,\dhat}} - \sqrt{n} \E[P]{H_{\etastarP,\betastarP,\dstarP}^T \Omega H_{\etastarP,\betastarP,\dstarP}}^{-1} H_{\etastarP,\betastarP,\dstarP}^T \Omega \varepsilon_{\etastarP,\betastarP,\dstarP} \Pto 0$$
    uniformly in $P \in \cP_{hte}$. 
\end{theorem}

\begin{proof}[\hypertarget{proof.th.ensemble_decomp}{Proof of \Cref{th.ensemble_decomp}}] 
    First, note that equicontinuity of $\lrp{F_{P}^{-1}(t)}_{P \in \cP_{hte}}$ implies the $J$ quantiles groups to be well-defined, which together with \Cref{as.ensemble} implies 
    $$\inf_{P \in \cP_{hte}} \det \lrp{H_{\etastarP,\betastarP,\dstarP}^T \Omega H_{\etastarP,\betastarP,\dstarP}} > 0.$$

    For each $m=1,\dots,M$, using \cref{eq.ensemble_epsilon} leads to the decomposition 
    $$\thetahat_{\etahat,\betahat,\dhat}^{(m)} - \theta_{\etahat,\betahat,\dhat}^{(m)} = \lrp{H_{\etahat,\betahat,\dhat}^T \Omega H_{\etahat,\betahat,\dhat}}^{-1} H_{\etahat,\betahat,\dhat}^T \Omega \varepsilon_{\etahat,\betahat,\dhat}.$$
    $$\lrp{n^{-1} H_{\etahat,\betahat,\dhat}^T \Omega H_{\etahat,\betahat,\dhat}}^{-1} \Pto \E[P]{H_{\etastarP,\betastarP,\dstarP}^T \Omega H_{\etastarP,\betastarP,\dstarP}}^{-1}$$
    by a uniform law of large numbers. 
    The terms in $n^{-1/2} H_{\etahat,\betahat,\dhat}^T \Omega \varepsilon_{\etahat,\betahat,\dhat}$ are given by 
    $$\lrp{\frac{1}{\sqrt{n}} \sum_{i=1}^n \omega_i \varepsilon_{\etahat,\betahat,\dhat,i} Z_i, \lrbc{\frac{1}{\sqrt{n}} \sum_{i=1}^n \omega_i \varepsilon_{\etahat,\betahat,\dhat,i} \lrbk{T_i - p(X_i)} \I{\dhat_{j-1, k(i)} \le \sum_{a=1}^{A} \beta_{\ell(i),a} \etahat_{\stilde_{k(i)},a}(X_i) < \dhat_{j, k(i)}}}_{j=1}^J }.$$
    These are split-sample empirical processes as in \Cref{th.clt_general}, with functions 
    $$f_{(\beta, d),\eta,1}(y, h,\omega) = \omega (y - h^T \theta_{\eta,\beta,d}) z $$
    and 
    $$f_{(\beta, d),\eta,1+j}(y, h, k, \ell,\omega) = \omega (y - h^T \theta_{\eta,\beta,d}) (t - p(x)) \I{d_{j-1, k} \le \sum_{a=1}^{A} \beta_{\ell,a} \etahat_{\stilde_{k},a}(X_i) < \dhat_{j, k}}.$$
    Step one of the proof of \Cref{th.clt_general} gives 
    $$\sup_{(\beta,d) \in B} \norm{n^{-1/2} H_{\etahat,\beta,d}^T \Omega \varepsilon_{\etahat,\beta,d} - n^{-1/2} H_{\etastarP,\beta,d}^T \Omega \varepsilon_{\etastarP,\beta,d}} \Pto 0.$$
    Together with consistency of $(\betahat, \dhat)$ to $(\betastarP,\dstarP)$, which follows from a uniform law of large numbers, this gives 
    $$n^{-1/2} H_{\etahat,\betahat,\dhat}^T \Omega \varepsilon_{\etahat,\betahat,\dhat} = n^{-1/2} H_{\etastarP,\betahat,\dhat}^T \Omega \varepsilon_{\etastarP,\betahat,\dhat} + o_P(1).$$
    Finally, asymptotic equicontinuity in $(\beta,d)$ gives 
    $$n^{-1/2} H_{\etastarP,\betahat,\dhat}^T \Omega \varepsilon_{\etastarP,\betahat,\dhat} = n^{-1/2} H_{\etastarP,\betastarP,\dstarP}^T \Omega \varepsilon_{\etastarP,\betastarP,\dstarP} + o_P(1).$$
    Summing over $m \in M$ concludes the proof. 
\end{proof}

\begin{proof}[\hypertarget{proof.th.ensemble_hte}{Proof of \Cref{th.ensemble_hte}}] 
    Follows from \Cref{th.ensemble_decomp}, Lyapunov's CLT and consistency of $\sigmahat_{\etahat}$, which follows by a law of large numbers. 
\end{proof}

\begin{proof}[\hypertarget{proof.th.ensemble_onesided_test}{Proof of \Cref{th.ensemble_onesided_test}}] 
    Follows directly from \Cref{th.diff_onesided_test}, noting that 
    $$\E[P]{MSR_\s | \stilde} \ge \E[P]{MSR_b}$$
    always holds when $\eta_{0,P}(x)$ is flat, since in that case the true coefficients $(\beta_a)$ in regression \cref{eq.ensemble_msr} are all zero. 
\end{proof}

%% file: modeling_power.tex
\section{Modeling Power} \label{appendix.modeling_power}

I formalize the notion that using a larger sample for training is desirable by the analyst by introducing the concept of modeling power. 
This appendix uses notation introduced in \Cref{section.setup}. 
I say that an estimator has better modeling power than another if its collection of splits has a smaller expected loss. 
Although my results rely on no assumptions on the training algorithm $\cA$ other than a mild stability condition on $\cA(D)$, in practice, $\cA$ typically minimizes some loss function. 
For example, in \Cref{example.classif_binary}, logistic regression minimizes log-likelihood, and neural networks minimize classification error over a class of network architectures. 
Let $\etahat_{\cR} = \lrp{\etahat_{\stilde_{m,k}}}_{m \in \lrbk{M}, k \in \lrbk{K}}$, $\ell_\eta(W)$ be a loss function, 
$$\phi(\eta) = \int \ell_\eta(w) dP(w)$$ 
be the loss value of function $\eta$, and 
$$\phi(\etahat_{\cR}) = (MK)^{-1} \sum_{r \in \cR} \sum_{\s \in r} \phi(\etahat_{\stilde}).$$ 
Note that $\phi(\etahat_{\cR})$ is equal to the expected value of $\phi(\etahat_{\s})$ over $\s \in \lrp{\stilde_{m,k}}_{m \in \lrbk{M}, k \in \lrbk{K}}$ uniformly at random, which is equivalent to the loss value of using a function $\etahat$ that takes value in $\etahat_{\cR}$ uniformly at random. 
The expected loss is defined as $\E[P]{\phi(\etahat_{\cR})}$. 

The expected loss, and thus the modeling power of an estimator depends only on the sample size used to estimate the functions in $\etahat_{\cR}$. 
That is because 
$$\E[P]{\phi(\etahat_{\cR})} = (MK)^{-1} \sum_{r \in \cR} \sum_{s \in r} \E[P]{\phi(\etahat_{\stilde})} = \E[P]{\phi(\etahat_{\xi})},$$
where $\xi$ is a random subset of $\bkn$ of size $n-b$, with $b = n/K$ if $K>1$, and assuming that $n$ is a multiple of $K$ for simplicity. 
If $\etahat_{\xi}$ is calculated with the goal of minimizing the loss $\phi(\eta)$ with respect to $\eta$, it is reasonable to assume that the expected loss $\E[P]{\phi(\etahat_{\xi})}$ decreases with the sample size used to calculate $\etahat_{\xi}$.  
If that is the case, the expected loss increases with $b$, since fewer data are used to estimate each $\etahat_{\stilde_{m,k}}$. 
Hence, to increase modeling power when $K=1$, one can pick a smaller $b$ (and $\pi$). 
However, if $\Mbar < \infty$, a smaller $b$ leads to smaller statistical power, since fewer data are used as evaluation sample at each split. 
When using cross-fitting, modeling power increases with $K$, since $b=n/K$. 
In this case, the returns to increasing $K$ are diminishing. 
For example, if $K=2$, $\etahat_{\stilde_{m,k}}$ is calculated with $50$\% of the sample, and this fraction raises to $90$\% with $K=10$. 
If $K=20$, however, the fraction only raises by another $5$\%. 
Although a large value of $K$ or small value of $\pi$ (when $K=1$) lead to better modelling power, my asymptotic framework takes these quantities as fixed. 
This means that the quality of the asymptotic approximation may be poor if $K$ is large (or $\pi$ small) relative to the sample size. 
For example, my asymptotic framework does not accommodate for leave-one-out cross-fitting, that is, $K=n$. 

%% file: averages.tex
\section{CLT for Split-Sample Averages} \label{section.avg}

I derive a CLT for split-sample estimators based on sample averages. 
The objective is to expose my main result in an accessible setting, and discuss the main insights of the proof. 
The result is generalized in \Cref{section.general}, where I derive a functional CLT uniformly over a large set of data generating processes, and in \Cref{section.z_estimators} where I prove a CLT for Z-estimators. 

The notation follows \Cref{section.setup}. 
Additionally, let $f_{\eta}: \cW \to \R$ be measurable functions for $\eta \in H$, and define 
\begin{equation} \label{def.marginal_exp}
    P f_\eta = \int_{w} f_\eta(w) dP(w),
\end{equation}
that is, $P f_\eta$ is a marginal expectation that takes $\eta$ as fixed. 
This is typical notation in the empirical process literature. 

\begin{RevExample}{example.classif_prob}
    In the probabilistic classifiers example, $W=(Y,X)$, $\eta$ is a function that predicts the probability of $Y=1$ given $X$, and 
    $$f_\eta(w) = \eta(x) \I{y = 1} + (1 - \eta(x)) \I{y = 0}.$$
    $P f_\eta$ is the correct classification rate of predictor $\eta$. 
\end{RevExample}

In this section, I consider estimators of the form 
\begin{equation} \label{def.thetahat_avg}
    \thetahat_{\etahat} = \frac{1}{M} \sum_{r \in \cR} \frac{1}{K} \sum_{\s \in r} \frac{1}{b} \sum_{i \in \s} f_{\etahat_{\stilde}}(W_i),
\end{equation}
where $\cR$ is a collection of $M$ random splits or cross-splits of the sample, $K$ is the number of folds ($K=1$ denotes sample-splitting), $b$ is the size of each subsample $\s$ (either the chosen subsample size when $K=1$ or the approximate fold size $n/K$ when $K>1$), and $\etahat = \etahat_{\cR} = \lrp{\lrp{\etahat_{\stilde}}_{\s \in r}}_{r \in \cR}$. 
I show in \Cref{th.clt_avg} that $\thetahat_{\etahat}$ is $\sqrt{n}$-Gaussian when centered around its marginal expectation 
$$\theta_{\etahat} = P \thetahat_{\etahat} = \frac{1}{M K} \sum_{r \in \cR} \sum_{\s \in r} P f_{\etahat_{\stilde}}.$$
In \Cref{example.classif_prob}, $\theta_{\etahat}$ is the fraction of individuals correctly classified under a rule that predicts $Y=1$ with probability 
$$\frac{1}{M K} \sum_{r \in \cR} \sum_{\s \in r} \etahat_{\stilde}(x)$$
for an individual with characteristics $X = x$. 

\Cref{as.avg} establishes sufficient conditions for the CLT in Theorem \ref{th.clt_avg}. 
\begin{assumption} \label{as.avg}
    \begin{assumptionenum}
        \item \label{as.2plusdelta} For some $\delta > 0$,
        $$\sup_{\eta \in H} \E[P]{\lrm{f_\eta(W)}^{2+\delta}} < \infty.$$
        \item \label{as.avg_etahat} For some $\etastar \in H$ and $\etatilde = \cA(D)$, 
        $$f_\etatilde(w) \Pto f_{\etastar}(w)$$
        pointwise for every $w$. 
    \end{assumptionenum}
\end{assumption}

\Cref{as.2plusdelta} is a standard moments condition for CLTs, uniformly over possible values of $\eta$. 
\Cref{as.avg_etahat} is a mild stability condition on $\etatilde$.
Importantly, $\etatilde$ is allowed to converge at any rate and to any limit $\etastar$. 
This condition is more interpretable but stronger than what I use for proving the more general CLTs in \Cref{section.general,section.z_estimators}. 
\Cref{as.avg_etahat} differs from the typical approach in the double machine learning literature where faster convergence rates (often $n^{-1/4}$) are required for nuisance functions, in a context where the target parameter does not depend on the estimated model $\etahat$ \citep[e.g.,][]{chernozhukov2018double}.
\begin{theorem} \label{th.clt_avg}
    Let \Cref{as.avg} hold. 
    Then, 
    $$\sqrt{n} \lrp{\thetahat_{\etahat} - \theta_{\etahat}} \leadsto \cN\lrp{0, V_{\Mbar, K} P \lrp{f_{\etastar} - P f_{\etastar}}^2},$$
    where 
    $$V_{\Mbar, K} = \begin{cases}
    \Mbar^{-1} \lrp{\pi^{-1} + \Mbar - 1}, & \text{if } K = 1 \text{ and } \Mbar < \infty \\
    1, & \text{otherwise}. 
    \end{cases}$$
\end{theorem}

\Cref{th.clt_avg} can be used to construct confidence intervals with the standard error 
$$\sigmahat_\etahat = \sqrt{V_{M, K}} \frac{1}{MK} \sum_{r \in \cR} \sum_{\s \in r} \sigmahat_{\etahat_{\stilde}},$$
where
$$\sigmahat_{\etahat_{\stilde}}^2 = \frac{1}{b} \sum_{i \in \s} \lrp{f_{\etahat_{\stilde}}(W_i) - \frac{1}{b} \sum_{i \in \s} f_{\etahat_{\stilde}}(W_i)}^2$$
and
$$V_{M, K} = \begin{cases}
    M^{-1} \lrp{n/b + M - 1}, & \text{if } K = 1 \\
    1, & \text{otherwise}. 
    \end{cases}$$
\begin{theorem} \label{th.avg_ci}
    Let \Cref{as.avg} hold and $P \lrp{f_{\etastar} - P f_{\etastar}}^2 > 0$. 
    Then, 
    $$P \lrp{\theta_\etahat \in \lrbk{\thetahat_\etahat - z_{1-\alpha/2} \frac{\sigmahat_\etahat}{\sqrt{n}}, \thetahat_\etahat + z_{1-\alpha/2} \frac{\sigmahat_\etahat}{\sqrt{n}}}} \to 1 - \alpha.$$
\end{theorem}

The proof of \Cref{th.clt_avg} relies on four main insights. 
I show them for the case of repeated cross-fitting, assuming that $n$ is a multiple of $K$ for simplicity. 
I provide a more detailed proof in \Cref{appendix.proofs_section_avg}, and a formal proof follows from the more general \Cref{th.clt_general,th.clt_z}. 
The first insight and main argument of the proof is to show that 
\begin{equation} \label{eq.avg_main_argument}
    \sqrt{n} \lrp{\thetahat_{\etahat} - \theta_{\etahat}} = \sqrt{n} \lrp{\frac{1}{n} \sum_{i=1}^n f_{\etastar}(W_i) - P f_{\etastar}} + o_P(1).
\end{equation}
Once this is established, the result follows from Lyapunov's CLT, since $\lrp{f_{\etastar}(W_i)}_{i=1}^n$ are iid. 
The second insight is that an application of Markov and H\"older inequalities gives that a sufficient condition for \cref{eq.avg_main_argument} is that 
\begin{equation} \label{eq.avg_2nd_insight}
    \Var[P]{\frac{1}{\sqrt{b}} \sum_{i \in \xi} \lrbk{f_{\etahat_{\xitilde}}(W_i) - P f_{\etahat_{\xitilde}}} - \lrbk{f_{\etastar}(W_i) - P f_{\etastar}}} \to 0,
\end{equation}
where $\xi$ is a random subset of $\bkn$ of size $b=n/K$ and $\xitilde$ is its complement. 
The third insight is that an application of the Law of Total Variance gives 
\begin{align}
        & \Var[P]{\frac{1}{\sqrt{b}} \sum_{i \in \xi} \lrbk{f_{\etahat_{\xitilde}}(W_i) - P f_{\etahat_{\xitilde}}} - \lrbk{f_{\etastar}(W_i) - P f_{\etastar}}} \nonumber \\
        & = \E[P]{\Var[P]{\frac{1}{\sqrt{b}} \sum_{i \in \xi} \lrp{f_{\etahat_{\xitilde}}(W_i) - P f_{\etahat_{\xitilde}}} - \lrp{f_{\etastar}(W_i) - P f_{\etastar}} \Bigm| D_{\xitilde} }} \label{eq.avg_3rd_1} \\
        & = \E[P]{\Var[P]{\lrp{f_{\etahat_{\xitilde}}(W) - P f_{\etahat_{\xitilde}}} - \lrp{f_{\etastar}(W) - P f_{\etastar}} \Bigm| D_{\xitilde}}}. \label{eq.avg_3rd_2}
\end{align}
Since the summands in \cref{eq.avg_3rd_1} are iid conditional on $D_{\xitilde}$, \cref{eq.avg_3rd_1} equals \cref{eq.avg_3rd_2}, which does not rely on the term $\sqrt{b}$. 
This is the crucial step that enables asymptotic normality without requiring an assumption on the rate at which $f_{\etahat_{\xitilde}}(W)$ converges to $f_{\etastar}(W)$. 

The final insight is that \Cref{as.avg} gives a sufficient condition for \cref{eq.avg_3rd_2} to converge to zero. 
For any $\varepsilon > 0$, 
\begin{align*}
    & \E[P]{\Var[P]{\lrp{f_{\etahat_{\xitilde}}(W) - P f_{\etahat_{\xitilde}}} - \lrp{f_{\etastar}(W) - P f_{\etastar}} \Bigm| D_{\xitilde}}} \\
    & \le \E[P]{\lrp{f_{\etahat_{\xitilde}}(W) - f_{\etastar}(W)}^2 } \\ 
    & = \E[P]{\lrp{f_{\etahat_{\xitilde}}(W) - f_{\etastar}(W)}^2 \I{\lrm{f_{\etahat_{\xitilde}}(W) - f_{\etastar}(W)} \le \varepsilon}} \\
    & \qquad + \E[P]{\lrp{f_{\etahat_{\xitilde}}(W) - f_{\etastar}(W)}^2 \I{\lrm{f_{\etahat_{\xitilde}}(W) - f_{\etastar}(W)} > \varepsilon}}, 
\end{align*}
where the first term is bounded by $\varepsilon^2$. 
By H\"older's inequality, the second term is bounded by 
$$\E[P]{\lrm{f_{\etahat_{\xitilde}}(W) - f_{\etastar}(W)}^{2+\delta}}^{\frac{2}{2+\delta}} P \lrp{\lrm{f_{\etahat_{\xitilde}}(W) - f_{\etastar}(W)} > \varepsilon}^{\frac{\delta}{2 + \delta}}.$$
The first term above is bounded by \Cref{as.2plusdelta}, and the second term can be made arbitrarily small since 
\begin{align*}
    P \lrp{\lrm{f_{\etahat_{\xitilde}}(W) - f_{\etastar}(W)} > \varepsilon} = \E[P]{P \lrp{\lrm{f_{\etahat_{\xitilde}}(W) - f_{\etastar}(W)} > \varepsilon \Bigm| W }}
\end{align*}
converges to zero by the dominated convergence theorem, since 
$$P \lrp{\lrm{f_{\etahat_{\xitilde}}(w) - f_{\etastar}(w)} > \varepsilon \Bigm| W = w} = P \lrp{\lrm{f_{\etahat_{\xitilde}}(w) - f_{\etastar}(w)} > \varepsilon } \to 0$$
from \Cref{as.avg_etahat} and independence of $W$ and $\etahat_{\xitilde}$. 
The result follows since $\varepsilon$ can be made arbitrarily small.

\subsection{Proofs} \label{appendix.proofs_section_avg}

\begin{proof}[\hypertarget{proof.th.clt_avg}{Proof of \Cref{th.clt_avg}}] 
    I provide a detailed proof for the repeated cross-fitting case discussed in \Cref{section.avg}, since that contains the main insights of the proof. 
    A complete and formal proof follows from the more general \Cref{th.clt_general}. 
    
    The argument consists of showing that 
    \begin{equation*} 
        \sqrt{n} \lrp{\thetahat_{\etahat} - \theta_{\etahat}} = \sqrt{n} \lrp{\frac{1}{n} \sum_{i=1}^n f_{\etastar}(W_i) - P f_{\etastar}} + o_P(1)
    \end{equation*}
    and applying Lyapunov's CLT to the first term on the right side of the equality. 

    Define $h(w, \eta) = \lrbk{f_{\eta}(w) - P f_{\eta}} - \lrbk{f_{\etastar}(w) - P f_{\etastar}}$ and note that 
    \begin{align*}
        & \sqrt{n} \lrp{\thetahat_{\etahat} - \theta_{\etahat}} - \sqrt{n} \lrp{\frac{1}{n} \sum_{i=1}^n f_{\etastar}(W_i) - P f_{\etastar}} \\
        & = \sqrt{b K} \frac{1}{M K} \sum_{r \in \cR} \sum_{\s \in r} \frac{1}{b} \sum_{i \in \s} h(W_i, \etahat_{\stilde}),
    \end{align*}
    since $b = n / K$ for cross-fitting. 
    For any $\varepsilon > 0$, it holds that
    \begin{align}
        & P \lrp{\lrm{\sqrt{b K} \frac{1}{M K} \sum_{r \in \cR} \sum_{\s \in r} \frac{1}{b} \sum_{i \in \s} h(W_i, \etahat_{\stilde})} > \varepsilon} \nonumber \\ 
        & \le P \lrp{\frac{\sqrt{K}}{M K} \sum_{r \in \cR} \sum_{\s \in r} \lrm{\frac{1}{\sqrt{b}} \sum_{i \in \s} h(W_i, \etahat_{\stilde})} > \varepsilon} \nonumber \\ 
        & \le \varepsilon^{-1} \frac{\sqrt{K}}{M K} \sum_{r \in \cR} \sum_{\s \in r} \E[P]{\lrm{\frac{1}{\sqrt{b}} \sum_{i \in \s} h(W_i, \etahat_{\stilde})}} \label{eq.proof_th.clt_avg.1} \\ 
        & = \varepsilon^{-1} \sqrt{K} \E[P]{\lrm{\frac{1}{\sqrt{b}} \sum_{i \in \xi} h(W_i, \etahat_{\xitilde})}} \label{eq.proof_th.clt_avg.2} \\ 
        & \le \varepsilon^{-1} \sqrt{K} \E[P]{\lrm{\frac{1}{\sqrt{b}} \sum_{i \in \xi} h(W_i, \etahat_{\xitilde})}^2}^{1/2} \label{eq.proof_th.clt_avg.3} \\ 
        & = \varepsilon^{-1} \sqrt{K} \Var[P]{\frac{1}{\sqrt{b}} \sum_{i \in \xi} h(W_i, \etahat_{\xitilde})}^{1/2}. \label{eq.proof_th.clt_avg.4}
    \end{align}
    \cref{eq.proof_th.clt_avg.1} follows from Markov's inequality. 
    \cref{eq.proof_th.clt_avg.2} defines $\xi$ as a random subset of $\bkn$ of size $b$, and uses the fact that the expected value does not depend on how the sample is (randomly) split. 
    \cref{eq.proof_th.clt_avg.3} follows from H\"older's inequality. 
    \cref{eq.proof_th.clt_avg.4} follows since \begin{align*}
        \E[P]{h(W, \etahat_{\xitilde})} & = \E[P]{\lrp{f_{\etahat_{\xitilde}}(W) - P f_{\etahat_{\xitilde}}} - \lrp{f_{\etastar}(W) - P f_{\etastar}}} \\
        & = \E[P]{ \E[P]{f_{\etahat_{\xitilde}}(W) - P f_{\etahat_{\xitilde}} \Bigm| D_{\xitilde}}} \\
        & = 0
    \end{align*}
    by definition.

    Since $K$ is assumed fixed, it is enough to show that 
    \begin{align}
        & \Var[P]{\frac{1}{\sqrt{b}} \sum_{i \in \xi} h(W_i, \etahat_{\xitilde})} \nonumber \\
        & = \Var[P]{\frac{1}{\sqrt{b}} \sum_{i \in \xi} \lrp{f_{\etahat_{\xitilde}}(W_i) - P f_{\etahat_{\xitilde}}} - \lrp{f_{\etastar}(W_i) - P f_{\etastar}}} \nonumber \\
        & = \E[P]{\Var[P]{\frac{1}{\sqrt{b}} \sum_{i \in \xi} \lrp{f_{\etahat_{\xitilde}}(W_i) - P f_{\etahat_{\xitilde}}} - \lrp{f_{\etastar}(W_i) - P f_{\etastar}} \Bigm| D_{\xitilde} }} \label{eq.proof_th.clt_avg.5} \\
        & = \E[P]{\Var[P]{\lrp{f_{\etahat_{\xitilde}}(W) - P f_{\etahat_{\xitilde}}} - \lrp{f_{\etastar}(W) - P f_{\etastar}} \Bigm| D_{\xitilde} }} \label{eq.proof_th.clt_avg.6} \\
        & = \E[P]{\Var[P]{f_{\etahat_{\xitilde}}(W) - f_{\etastar}(W) \Bigm| D_{\xitilde} }} \label{eq.proof_th.clt_avg.7}
    \end{align}
    converges to zero. 
    \cref{eq.proof_th.clt_avg.5} follows from the Law of Total Variance, since 
    $$\E[P]{\lrp{f_{\etahat_{\xitilde}}(W_i) - P f_{\etahat_{\xitilde}}} - \lrp{f_{\etastar}(W_i) - P f_{\etastar}} \Bigm| D_{\xitilde}} = 0.$$
    \cref{eq.proof_th.clt_avg.6} follows since the observations are iid conditional on $D_{\xitilde}$. 
    
    To show convergence to zero of \cref{eq.proof_th.clt_avg.7}, consider the inequality 
    \begin{align*}
        \E[P]{\Var[P]{f_{\etahat_{\xitilde}}(W) - f_{\etastar}(W) \Bigm| D_{\xitilde}}} & \le \E[P]{\E[P]{\lrp{f_{\etahat_{\xitilde}}(W) - f_{\etastar}(W)}^2 \Bigm| D_{\xitilde}}} \\
        & = \E[P]{\lrp{f_{\etahat_{\xitilde}}(W) - f_{\etastar}(W)}^2 }. 
    \end{align*}
    For any fixed $\varepsilon > 0$, 
    \begin{align*}
        & \E[P]{\lrp{f_{\etahat_{\xitilde}}(W) - f_{\etastar}(W)}^2} \\
        & = \E[P]{\lrp{f_{\etahat_{\xitilde}}(W) - f_{\etastar}(W)}^2 \I{\lrm{f_{\etahat_{\xitilde}}(W) - f_{\etastar}(W)} \le \varepsilon}} \\
        & \qquad + \E[P]{\lrp{f_{\etahat_{\xitilde}}(W) - f_{\etastar}(W)}^2 \I{\lrm{f_{\etahat_{\xitilde}}(W) - f_{\etastar}(W)} > \varepsilon}}. 
    \end{align*}
    The first term is bounded by $\varepsilon^2$. 
    By H\"older's inequality, 
    \begin{align*}
        & \E[P]{\lrp{f_{\etahat_{\xitilde}}(W) - f_{\etastar}(W)}^2 \I{\lrm{f_{\etahat_{\xitilde}}(W) - f_{\etastar}(W)} > \varepsilon}} \\
        & \le \E[P]{\lrm{f_{\etahat_{\xitilde}}(W) - f_{\etastar}(W)}^{2+\delta}}^{\frac{2}{2+\delta}} P \lrp{\lrm{f_{\etahat_{\xitilde}}(W) - f_{\etastar}(W)} > \varepsilon}^{\frac{\delta}{2 + \delta}}.
    \end{align*}
    The first term above is bounded by \Cref{as.2plusdelta}, and the second term can be made arbitrarily small since 
    \begin{align*}
        P \lrp{\lrm{f_{\etahat_{\xitilde}}(W) - f_{\etastar}(W)} > \varepsilon} = \E[P]{P \lrp{\lrm{f_{\etahat_{\xitilde}}(W) - f_{\etastar}(W)} > \varepsilon \Bigm| W }}
    \end{align*}
    converges to zero by the dominated convergence theorem, since 
    $$P \lrp{\lrm{f_{\etahat_{\xitilde}}(w) - f_{\etastar}(w)} > \varepsilon \Bigm| W = w} = P \lrp{\lrm{f_{\etahat_{\xitilde}}(w) - f_{\etastar}(w)} > \varepsilon } \to 0$$
    from \Cref{as.avg_etahat} and independence of $W$ and $\etahat_{\xitilde}$. 
    The result follows since $\varepsilon$ can be made arbitrarily small.     
\end{proof}

\begin{proof}[\hypertarget{proof.th.avg_ci}{Proof of \Cref{th.avg_ci}}] 
    Note 
    $$\sigmahat_{\etahat_{\stilde}}^2 = \frac{1}{b} \sum_{i \in \s} \lrp{f_{\etahat_{\stilde}}(W_i)^2} - \lrp{\frac{1}{b} \sum_{i \in \s} f_{\etahat_{\stilde}}(W_i)}^2.$$
    By a law of large numbers conditional on $\stilde$, 
    $$\frac{1}{b} \sum_{i \in \s} f_{\etahat_{\stilde}}(W_i)^2 - \E[P]{f_{\etahat_{\stilde}}(W)^2 | D_{\stilde}} \Pto 0,$$
    and similarly 
    $$\frac{1}{b} \sum_{i \in \s} f_{\etahat_{\stilde}}(W_i) - \E[P]{f_{\etahat_{\stilde}}(W) | D_{\stilde}} \Pto 0.$$
    Hence, 
    $$\sigmahat_{\etahat_{\stilde}}^2 - \lrp{\E[P]{f_{\etahat_{\stilde}}(W)^2 | D_{\stilde}} - \E[P]{f_{\etahat_{\stilde}}(W) | D_{\stilde}}^2} \Pto 0.$$
    Fix $\varepsilon > 0$ and define $h_{\etahat_{\stilde}}(w) = \lrm{f_{\etahat_{\stilde}}(W) - f_{\etastar}(W)}$. 
    \begin{align*}
        & \E[P]{h_{\etahat_{\stilde}}(W) | D_{\stilde}} \\
        & = \E[P]{h_{\etahat_{\stilde}}(W) \I{h_{\etahat_{\stilde}}(W) \le \varepsilon} | D_{\stilde}} + \E[P]{h_{\etahat_{\stilde}}(W) \I{h_{\etahat_{\stilde}}(W) > \varepsilon} | D_{\stilde}} \\
        & \le \varepsilon + \E[P]{h_{\etahat_{\stilde}}(W)^{1+\delta} | D_{\stilde}}^\frac{1}{1 + \delta} P \lrp{h_{\etahat_{\stilde}}(W) > \varepsilon | D_{\stilde}}^\frac{\delta}{1 + \delta}
    \end{align*}
    by H\"older's inequality. 
    The term $\E[P]{h_{\etahat_{\stilde}}(W)^{1+\delta} | D_{\stilde}}^\frac{1}{1 + \delta}$ is bounded by \Cref{as.2plusdelta}, and I show that $P \lrp{h_{\etahat_{\stilde}}(W) > \varepsilon | D_{\stilde}}^\frac{\delta}{1 + \delta}$ converges in probability to zero. 
    In the proof of \Cref{th.clt_avg}, I established that 
    $$\E[P]{P \lrp{\lrm{f_{\etahat_{\xitilde}}(w) - f_{\etastar}(w)} > \varepsilon \Bigm| D_{\stilde} }} = P \lrp{\lrm{f_{\etahat_{\xitilde}}(w) - f_{\etastar}(w)} > \varepsilon } \to 0.$$
    This implies that $P \lrp{\lrm{f_{\etahat_{\xitilde}}(w) - f_{\etastar}(w)} > \varepsilon \Bigm| D_{\stilde} } \Pto 0$
    since $L_1$ convergence implies convergence in probability. 
    Hence, 
    $$\E[P]{\lrm{f_{\etahat_{\stilde}}(W) - f_\etastar(W)} | D_{\stilde}} \Pto 0,$$
    which implies 
    $$\E[P]{f_{\etahat_{\stilde}}(W)| D_{\stilde}} - \E[P]{f_\etastar(W) } \Pto 0.$$
    A similar argument gives 
    $$\E[P]{f_{\etahat_{\stilde}}(W)^2| D_{\stilde}} - \E[P]{f_\etastar(W)^2 } \Pto 0.$$
    Combining results implies 
    $$\sigmahat_{\etahat_{\stilde}}^2 \Pto P \lrp{f_{\etastar} - P f_{\etastar}}^2.$$
    The result follows from \Cref{th.clt_avg}, since $V_{M, K}/V_{\Mbar, K} \to 1$.  
\end{proof}

%% file: general.tex
\section{CLT for Split-Sample Empirical Processes} \label{section.general}

I derive a CLT for empirical processes based on a broad class of split-sample procedures, uniformly over a large class of probability distributions. 
This section generalizes \Cref{section.avg}, which gives a more accessible exposition focusing on the particular case of sample averages. 
The CLT of this section can be used to prove asymptotic normality for a large class of estimators. 
That is the case for Z-estimators, which I develop in \Cref{section.z_estimators}. 
Moreover, this CLT can be used to establish asymptotic consistency of the bootstrap in several applications, following, for example, the arguments in Chapter 3.7 of \citet{van2023weak}. 

The notation follows \Cref{section.setup}. 
Let $\cP$ be a set of probability distributions, and $D = \{ W_i \}_{i \in \bkn}$, the dataset, be an iid sample of $W \sim P \in \cP$. 
I denote the expected value under $P \in \cP$ by $\E[P]$, and the variance by $\operatorname{Var}_{P}$. 
Given a set $T$, let $f_{t, \eta}: \cW \to \R$ be measurable functions for $t \in T$ and $\eta \in H$, with $H$ defined as in \Cref{section.setup}, and let $\cF_\eta = \lrbc{f_{t, \eta} : t \in T}$. 
$\etahat = \etahat_{\cR} = \lrp{\etahat_{\stilde_{m,k}}}_{m \in \lrbk{M}, k \in \lrbk{K}}$, $\norm{f}_{Q,r} = \lrp{\int |f|^r dQ}^{1/r}$, $L_r(Q) = \norm{\cdot}_{Q,r}$, and $\cQ$ denotes all finitely discrete probability distributions. 
I use $\lrm{ x }$ to denote cardinality when $x$ is a set and absolute value when $x$ is scalar. 
I denote by $N$ and $N_{[\, ]}$ respectively the covering and bracketing numbers, as in Definitions 2.1.5 and 2.1.6 of \citet{van2023weak}. 
For $\s \subseteq \bkn$, define the empirical measure
$$\bP_{\s} f_{t, \eta} = \frac{1}{|\s|} \sum_{i \in \s} f_{t, \eta}(W_i),$$
the marginal expectation
$$P f_{t, \eta} = \int_{w} f_{t, \eta}(w) dP(w),$$
and the empirical process
$$\G_{n, \etahat} (t) = \sqrt{n} \frac{1}{M} \sum_{r \in \cR} \frac{1}{K} \sum_{\s \in r} \lrp{\bP_{\s} f_{t, \etahat_{\stilde}} - P f_{t, \etahat_{\stilde}}}.$$

I establish below sufficient conditions for the CLT for split-sample empirical processes, presented in \Cref{th.clt_general}. 

\begin{assumption} \label{as.emp_proc} 
    The following conditions hold:
    \begin{assumptionenum}
        \item \label{as.total_bounded} $T$ is totally bounded for some semimetric $\rho$;
        \item \label{as.measurability} For every $\eta \in H$ and $t \in T$, $f_{t, \eta}$ is measurable;
        \item \label{as.envelope} For all $\eta \in H$, there exists a measurable envelope function $F_\eta$; 
        That is, $F_\eta : \cW \to \R$ is such that $|f_{t, \eta}(w)| \leq F_\eta(w) < \infty$ for all $t \in T$ and $w \in \cW$; 
        \item \label{as.unif_integrability} $\lim_{B \to \infty} \sup_{P \in \cP} \sup_{\eta \in H} \E[P]{F_\eta(W)^2 \I{F_\eta(W) > B}} = 0$; 
        \item \label{as.equicontinuity} For every $\delta_n \downarrow 0$, $$\sup_{P \in \cP} \sup_{\eta \in H} \sup_{\rho(t, t') < \delta_n} \E[P]{\lrp{f_{t, \eta}(W) - f_{t', \eta}(W)}^2} \to 0;$$
        \item \label{as.entropy} One of the following conditions holds for all $\delta_n \downarrow 0$:
            \begin{equation} \label{eq.entropy1}
                \sup_{\eta \in H} \sup_{Q \in \cQ} \int_{0}^{\delta_n} \sqrt{\log \covering{\varepsilon, \cF_\eta, L_2(Q)}} d \varepsilon \to 0,
            \end{equation}
            or 
            \begin{equation} \label{eq.entropy2}
                \sup_{P \in \cP} \sup_{\eta \in H} \int_{0}^{\delta_n} \sqrt{\log \bracketing{\varepsilon, \cF_\eta, L_2(P)}} d \varepsilon \to 0;
            \end{equation}
    \end{assumptionenum}
\end{assumption}

\begin{assumption} \label{as.etahat}
    There exists $\etastarP \in H$ such that for $\etatilde = \cA(D)$, $W \perp D$, and every $t \in T$, 
    $$\Var[P]{f_{t, \etatilde}(W) - f_{t, \etastarP}(W) \Bigm| D } \Pto 0$$ 
    uniformly in $P \in \cP$. 
\end{assumption}

Although technical, \Cref{as.emp_proc} is a weak condition that is satisfied in many applications. 
\Cref{as.total_bounded} through \ref{as.entropy} are standard Donsker conditions in the literature of weak convergence of empirical processes \citep[e.g.,][]{van2023weak}, generalized for the presence of the functions $\eta \in H$. 
In fact, if $T = \lrbc{t}$ and $\cP = \lrbc{P}$ are singletons, these conditions are implied by the ``$2 + \delta$'' moments condition in \Cref{as.2plusdelta} (\Cref{prop.avg_implies_empproc}). 
These assumptions are standard for proving functional CLTs by limiting the complexity of the sets $T$ and $\cF_\eta$. 
In addition to ensuring that each set $\cF_\eta$ is Donsker, \Cref{as.emp_proc} requires that the inequalities and convergences be uniform in $\eta \in H$. 
Importantly, \Cref{as.entropy} does not restrict the complexity of the class $H$, and it does not imply the much stronger condition that $\bigcup_{\eta \in H} \cF_\eta$ is Donsker. 
In applications, except for the restrictions on $\cP$, \Cref{as.total_bounded} through \Cref{as.entropy} are verifiable since they depend only on the choices of $T$ and $\cF_\eta$, and typically do not depend on how $\eta$ is calculated. 
The assumptions on $\cP$ involve the mild uniform square integrability condition \Cref{as.unif_integrability}, and the smoothness condition \Cref{as.equicontinuity}. 

Assumptions \Cref{as.total_bounded} through \Cref{as.entropy} give standard conditions for a CLT when $\cR$ consists of a single sample split. 
The proof for the case of multiple splits relies on the additional \Cref{as.etahat}. 
This is a weak stability condition that requires $\etatilde$ to converge at any rate to any function $\etastarP$, which is allowed to depend on $P$. 
If $T$ and $\cP$ are singletons, this is implied by \Cref{as.avg_etahat} (\Cref{prop.avg_implies_empproc}). 
Note that the requirement is pointwise in $t \in T$, and it holds, for example, if $f_{t, \etatilde}(w) \Pto f_{t, \etastarP}(w)$ for almost all $w \in \cW$. 

\begin{theorem} \label{th.clt_general} 
    \hyperlink{proof.th.clt_general}{(CLT for split-sample empirical processes)} \, \\
    Let \Cref{as.emp_proc,as.etahat} hold. 
    Then, the sequence $\G_{n, \etahat}$ is asymptotically $\rho$-equicontinuous uniformly in $P \in \cP$ and 
    $$\sup_{t \in T} \lrm{\G_{n, \etahat}(t) - \G_{n,\etastarP} (t)} \Pto 0$$
    uniformly in $P \in \cP$, where
    $$\G_{n,\etastarP} (t) = \sqrt{n} \frac{1}{M} \sum_{r \in \cR} \frac{1}{K} \sum_{\s \in r} \lrp{\bP_{\s} f_{t, \etastarP} - P f_{t, \etastarP}}.$$ 
    For any sequence $(P_n)_{n \ge 1} \subseteq \cP$ such that, for every $t, t' \in T$, 
    \begin{equation} \label{eq.cov_function}
        \E[P_n]{ \lrp{f_{t, \etastarPn}(W) - P_n f_{t, \etastarPn}} \lrp{f_{t', \etastarPn}(W) - P_n f_{t', \etastarPn}}} \to \sigma_{t,t'}, 
    \end{equation}
    for some $\sigma_{t,t'}$, 
    $$\G_{n, \etahat} \leadsto \G_{\etastar}$$
    in $\ell^{\infty}(T)$, where $\G_{\etastar}$ is a tight Gaussian process. 
    Moreover, the covariance function of $\G_{\etastar}$ is given by $V_{\Mbar, K} \sigma_{t,t'}$, where 
    $$V_{\Mbar, K} = \begin{cases}
    \Mbar^{-1} \lrp{\pi^{-1} + \Mbar - 1}, & \text{if } K = 1 \text{ and } \Mbar < \infty \\
    1, & \text{otherwise}. 
    \end{cases}$$
\end{theorem}

To the best of my knowledge, this appears to be the first central limit theorem for empirical processes that average over multiple splits of the sample. 
This result enables asymptotic inference for a large class of split-sample estimators. 
For example, combined with the functional delta method, it immediately implies asymptotic normality of Hadamard differentiable functionals of the split-sample empirical measure 
$$\sqrt{n} \frac{1}{M} \sum_{r \in \cR} \frac{1}{K} \sum_{\s \in r} \bP_{\s} f_{t, \etahat_{\stilde}}.$$
In \Cref{section.z_estimators}, I use \Cref{th.clt_general} as a building block to prove asymptotic normality of split-sample Z-estimators, a broad class that cover many if not most estimators used in practice, including the ones in \Cref{section.application_ghana}. 

\subsection{Proofs}

\begin{lemma} \label{lemma.asy_equic_seq_eta}
    Let Assumptions \Cref{as.total_bounded} through \Cref{as.entropy} hold, $(\eta_n)_{n \ge 1} \subseteq H$ be a deterministic sequence, $(P_n)_{n \ge 1} \subseteq \cP$, and $s \subseteq \bkn$ be a random (uniformly) subset of $\bkn$ such that $|\s| \to \infty$ as $n \to \infty$. 
    Define
    $$X_{n,s}(t) = \frac{1}{\sqrt{|\s|}} \sum_{i \in \s} \lrp{f_{t, \eta_n}(W_i) - P_n f_{t, \eta_n}}$$
    Then, the sequence $X_{n,s}$ is asymptotically $\rho$-equicontinuous. 
\end{lemma}

\begin{proof}[\hypertarget{proof.lemma.asy_equic_seq_eta}{Proof of \Cref{lemma.asy_equic_seq_eta}}] 
    \, \\ 
    The result follows from an application of Theorems 2.11.1 and 2.11.9 in \citet{van2023weak}, respectively for when conditions \cref{eq.entropy1,eq.entropy2} hold. 
    Their notation is adapted with $m_n = |\s|$, $\cF = T$, and $Z_{ni}(t) = |\s|^{-1/2} f_{t,\eta_n}(W_i)$, where it is implicit in the notation that $W_i \sim P_n$ (alternatively, one could denote $W_{ni}$ instead of $W_i$). 
    The presence of the suprema over $P \in \cP$ and $\eta \in H$ guarantee that the conditions in those theorems hold for any sequences $(\eta_n)_{n \ge 1}$ and $(P_n)_{n \ge 1}$. 
\end{proof}

\begin{lemma} \label{lemma.asy_equic_etahat}
    Let Assumptions \Cref{as.total_bounded} through \Cref{as.entropy} hold, and $s \subseteq \bkn$ be a random (uniformly) subset such that $|\s| \to \infty$ as $n \to \infty$. 
    Define 
    $$X_{n,s,\eta}(t) = \frac{1}{\sqrt{|\s|}} \sum_{i \in \s} \lrp{f_{t, \eta}(W_i) - P f_{t, \eta}}.$$
    Then, the sequence $X_{n,s,\etahat_{\stilde}}$ is asymptotically $\rho$-equicontinuous uniformly in $P \in \cP$. 
\end{lemma}

\begin{proof}[\hypertarget{proof.lemma.asy_equic_etahat}{Proof of \Cref{lemma.asy_equic_etahat}}] 
    Let $\cF_{\eta, \delta} = \lrbc{f - g: f,g \in \cF_{\eta}, \rho(f,g)< \delta}$ and $\varepsilon > 0$. 
    \begin{align}
        & \sup_{P \in \cP} P \lrp{\norm{X_{n,s,\etahat_{\stilde}}}_{\cF_{\etahat_{\stilde}, \delta}} > \varepsilon} \nonumber \\
        & = \sup_{P \in \cP} \int_{D_s, D_{\stilde}} \I{\norm{X_{n,s,\etahat_{\stilde}(D_{\stilde})}(D_{s})}_{\cF_{\etahat_{\stilde}(D_{\stilde}), \delta}} > \varepsilon} dP(D_s, D_{\stilde}) \label{eq.proof.lemma.asy_equic_etahat.integral} \\
        & = \sup_{P \in \cP} \int_{D_{\stilde}} \lrbk{\int_{D_s} \I{\norm{X_{n,s,\etahat_{\stilde}(D_{\stilde})}(D_{s})}_{\cF_{\etahat_{\stilde}(D_{\stilde}), \delta}} > \varepsilon} dP(D_s)} dP(D_{\stilde}) \label{eq.proof.lemma.asy_equic_etahat.independence} \\
        & \le \sup_{P \in \cP} \int_{D_{\stilde}} \sup_{\eta \in H} \lrbk{\int_{D_s} \I{\norm{X_{n,s,\eta}(D_{s})}_{\cF_{\eta, \delta}} > \varepsilon} dP(D_s)} dP(D_{\stilde}) \nonumber \\
        & = \sup_{P \in \cP} \sup_{\eta \in H} \lrbk{\int_{D_s} \I{\norm{X_{n,s,\eta}(D_{s})}_{\cF_{\eta, \delta}} > \varepsilon} dP(D_s)} \nonumber \\
        & = \sup_{P \in \cP} \sup_{\eta \in H} P \lrp{\norm{X_{n,s,\eta}(D_{s})}_{\cF_{\eta, \delta}} > \varepsilon}, \nonumber
    \end{align}
    where \cref{eq.proof.lemma.asy_equic_etahat.integral} makes explicit the dependence of $X_{n,s,\etahat_{\stilde}}$ on the subsample $D_s$ and of $\etahat_{\stilde}$ on $D_{\stilde}$, and \cref{eq.proof.lemma.asy_equic_etahat.independence} uses the fact that the split is random and $D_s, D_{\stilde}$ are independent. 

    Hence, for an arbitrary $\delta_n \downarrow 0$, $\sup_{P \in \cP} P \lrp{\norm{X_{n,s,\etahat_{\stilde}}}_{\cF_{\etahat_{\stilde}, \delta_n}} > \varepsilon} \to 0$ follows from 
    $$\sup_{P \in \cP} P \lrp{\norm{X_{n,s,\eta_n}(D_{s})}_{\cF_{\eta_n, \delta_n}} > \varepsilon} \to 0$$
    for any deterministic $(\eta_n)_{n \ge 1} \subseteq H$, which is established in \Cref{lemma.asy_equic_seq_eta}. 
\end{proof}

\begin{proof}[\hypertarget{proof.th.clt_general}{Proof of \Cref{th.clt_general}}] 
    \, \\ 
    The proof is divided into three main steps. 
    First, I show that 
    $$\sup_{t \in T} \lrm{\G_{n, \etahat}(t) - \G_{n,\etastarP}(t)} \Pto 0$$
    uniformly in $P \in \cP$.
    Second, I show that $\G_{n,\etastarP}$ is asymptotically $\rho$-equicontinuous. 
    Finally, I prove the Gaussian limit of $(\G_{n,\etastarP}(t))_{t \in T'}$ for any finite $T' \subseteq T$.  

    \paragraph{Step one.}
    Let $h_t(w, \eta)=\lrbk{f_{t, \eta}(w) - P f_{t, \eta}} - \lrbk{f_{t, \etastarP}(w) - P f_{t, \etastarP}}$, $\pi_n = b / n$, and fix $\varepsilon > 0$. 
    It follows that 
    \begin{align}
        & \sup_{P \in \cP} P \lrp{\sup_{t \in T} \lrm{\G_{n, \etahat}(t) - \G_{n,\etastarP}(t)} > \varepsilon} \nonumber \\
        & = \sup_{P \in \cP} P \lrp{\sup_{t \in T} \lrm{\sqrt{n} \frac{1}{M K} \sum_{r \in \cR} \sum_{\s \in r} \frac{1}{b} \sum_{i \in \s} h_t(W_i, \etahat_{\stilde})} > \varepsilon} \nonumber \\
        & \le \sup_{P \in \cP} P \lrp{ \frac{1}{M K} \sum_{r \in \cR} \sum_{\s \in r} \sup_{t \in T} \lrm{ \frac{\sqrt{n}}{b} \sum_{i \in \s} h_t(W_i, \etahat_{\stilde})} > \varepsilon} \nonumber \\
        & \le \varepsilon^{-1} \frac{1}{M K} \sum_{r \in \cR} \sum_{\s \in r} \sup_{P \in \cP} \E[P]{  \sup_{t \in T} \lrm{ \frac{\sqrt{n}}{b} \sum_{i \in \s} h_t(W_i, \etahat_{\stilde})}} \label{eq.proof.th.clt_general.markov} \\
        & = \varepsilon^{-1} \sup_{P \in \cP} \E[P]{  \sup_{t \in T} \lrm{ \frac{\sqrt{n}}{b} \sum_{i \in \xi} h_t(W_i, \etahat_{\xitilde})}} \label{eq.proof.th.clt_general.xi} \\
        & = \varepsilon^{-1} \pi_{n}^{-1/2} \sup_{P \in \cP} \E[P]{  \sup_{t \in T} \lrm{ \frac{1}{\sqrt{b}} \sum_{i \in \xi} h_t(W_i, \etahat_{\xitilde})}}, \label{eq.proof.th.clt_general.pi}
    \end{align}
    where \cref{eq.proof.th.clt_general.markov} follows from Markov's inequality, and \cref{eq.proof.th.clt_general.xi} defines $\xi$ as a random subset of $\bkn$ (uniformly over all subsets). 

    Since $\pi_n \to \pi \in (0,1)$, \cref{eq.proof.th.clt_general.pi} converges to zero if the term inside the expectation convergences in probability to zero uniformly in $P \in \cP$, since it is uniformly integrable (by \Cref{as.unif_integrability}). 
    This follows from stochastic equicontinuity of $\frac{1}{\sqrt{b}} \sum_{i \in \xi} h_t(W_i, \etahat_{\xitilde})$ (as a process indexed by $t \in T$) and pointwise convergence in $t$, by applying Theorem 22.9 in \citet{davidson2021stochastic}. 
    Stochastic equicontinuity follows since 
    \begin{align*}
        \frac{1}{\sqrt{b}} \sum_{i \in \xi} h_t(W_i, \etahat_{\xitilde}) & = \frac{1}{\sqrt{b}} \sum_{i \in \xi} \lrbk{f_{t, \etahat_{\xitilde}}(W_i) - P f_{t, \etahat_{\xitilde}}} - \frac{1}{\sqrt{b}} \sum_{i \in \xi} \lrbk{f_{t, \etastarP}(W_i) - P f_{t, \etastarP}}
    \end{align*}
    is a sum of two stochastically equicontinuous processes, respectively by \Cref{lemma.asy_equic_etahat} and \Cref{lemma.asy_equic_seq_eta}. 
    For pointwise convergence, I show that the variance converges to zero, and note $h_t(w,\eta)$ is mean zero by construction. 
    For an arbitrary $t \in T$, 
    \begin{align}
        & \sup_{P \in \cP} \Var[P]{\frac{1}{\sqrt{b}} \sum_{i \in \xi} h_t(W_i, \etahat_{\xitilde})} \nonumber \\ 
        & = \sup_{P \in \cP} \Var[P]{\frac{1}{\sqrt{b}} \sum_{i \in \xi} \lrp{f_{t, \etahat_{\xitilde}}(W_i) - P f_{t, \etahat_{\xitilde}}} - \lrp{f_{t, \etastarP}(W_i) - P f_{t, \etastarP}}} \nonumber \\ 
        & = \sup_{P \in \cP} \E[P]{\Var[P]{\frac{1}{\sqrt{b}} \sum_{i \in \xi} \lrp{f_{t, \etahat_{\xitilde}}(W_i) - P f_{t, \etahat_{\xitilde}}} - \lrp{f_{t, \etastarP}(W_i) - P f_{t, \etastarP}} \Biggm| D_{\xitilde}}} \label{eq.proof.th.clt_general.ltv} \\ 
        & = \sup_{P \in \cP} \E[P]{\Var[P]{ \lrp{f_{t, \etahat_{\xitilde}}(W) - P f_{t, \etahat_{\xitilde}}} - \lrp{f_{t, \etastarP}(W) - P f_{t, \etastarP}} \Biggm| D_{\xitilde}}} \label{eq.proof.th.clt_general.indep} \\ 
        & = \sup_{P \in \cP} \E[P]{\Var[P]{ f_{t, \etahat_{\xitilde}}(W) - f_{t, \etastarP}(W) \Biggm| D_{\xitilde}}}, \nonumber
    \end{align}
    where \cref{eq.proof.th.clt_general.ltv} uses the Law of Total Variance and the fact that $\E[P]{h_t(W_i, \etahat_{\xitilde}) | D_{\xitilde}} = 0$, and \cref{eq.proof.th.clt_general.indep} follows since the summands are iid conditional on $D_{\xitilde}$. 
    Finally, the last term converges to zero from \Cref{as.etahat}. 
    Note that since $f_{t,\eta}$ are uniformly square integrable by \Cref{as.unif_integrability}, convergence in probability of the conditional variance implies its convergence in $L_1$. 

    \paragraph{Step two.}
    Let $\lambda_i = (\pi_n M K)^{-1} \lrm{\lrbc{s \in \lrbc{s_{m,k}}_{m \in \lrbk{M}, k \in \lrbk{K}} : i \in \s}}$ and note that
    $$\G_{n,\etastarP}(t) = \frac{1}{\sqrt{n}} \sum_{i = 1}^n \lambda_i \lrp{f_{t, \etastarP}(W_i) - P f_{t, \etastarP}}.$$
    
    Let $\lambda = \lrp{\lambda_i}_{i \in \bkn}$, $\cF_{\etastarP, \delta} = \lrbc{f - g: f,g \in \cF_{\etastarP}, \rho(f,g)< \delta}$, $\varepsilon > 0$, and $\delta_n \downarrow 0$.  
    \begin{align*}
        \sup_{P \in \cP} P \lrp{\norm{\G_{n,\etastarP}}_{\cF_{\etastarP, \delta_n}} > \varepsilon} & = \sup_{P \in \cP} \E[P]{ P \lrp{\norm{\G_{n,\etastarP}}_{\cF_{\etastarP, \delta_n}} > \varepsilon \Bigm| \lambda }} \\
        & \le \sup_{\lambda: \lambda_i \le \pi_n^{-1}} \sup_{P \in \cP} P \lrp{\norm{\G_{n,\etastarP}}_{\cF_{\etastarP, \delta_n}} > \varepsilon \Bigm| \lambda }.
    \end{align*}
    The last term converges to zero from asymptotic equicontinuity of $\frac{1}{\sqrt{n}} \sum_{i = 1}^n \lambda_{n,i} \lrp{f_{t, \etastarPn}(W_i) - P f_{t, \etastarPn}}$ for arbitrary sequences $(P_n)_{n \ge 1} \subseteq \cP$ and $(\lambda_{n,i})_{n \ge 1, i \in \bkn}$ satisfying $\lambda_{n,i} \le \pi_n^{-1}$ for all $n,i$. 
    Asymptotic equicontinuity can be verified under \Cref{as.emp_proc}, for example, by applying Theorem 2.11.1 (when \cref{eq.entropy1} holds) and Theorem 2.11.9 (when \cref{eq.entropy2} holds) of \citet{van2023weak}. 
    Their notation is adapted with $m_n = n$, $\cF = T$, and $Z_{ni}(t) = n^{-1/2} \lambda_{n,i} f_{t, \etastarPn}(W_i)$, where it is implicit in the notation that $W_i \sim P_n$ (alternatively, one could denote $W_{ni}$ instead of $W_i$). 
    For $\gamma > 0$, note that 
    \begin{align*}
        & \E[P]{\sup_{t \in T} \lrp{n^{-1/2} \lambda_{n,i} f_{t, \etastarPn}(W_i)}^2 \I{\sup_{t \in T} \lrm{n^{-1/2} \lambda_{n,i} f_{t, \etastarPn}(W_i)} > \gamma }} \\
        & \le \pi_n^{-2} \E[P]{\sup_{t \in T} \lrp{n^{-1/2} f_{t, \etastarPn}(W_i)}^2 \I{\sup_{t \in T} \lrm{n^{-1/2} \pi_n^{-1} f_{t, \etastarPn}(W_i)} > \gamma }},
    \end{align*}
    and 
    \begin{align*}
        & \lrp{n^{-1/2} \lambda_{n,i} f_{t, \etastarPn}(W_i) - n^{-1/2} \lambda_{n,i} f_{t', \etastarPn}(W_i)}^2 \\
        & \le \pi_n^{-2} \lrp{n^{-1/2} f_{t, \etastarPn}(W_i) - n^{-1/2} f_{t', \etastarPn}(W_i)}^2,
    \end{align*}
    for any $t, t' \in T$, $n$, and $i \in \bkn$. 

    \paragraph{Step three.} 
    If $K > 1$, $\lambda_i=1$ for all $i$, and the Gaussian limit follows from Lindeberg's CLT and the Cram\'er-Wold device, using \Cref{as.unif_integrability}. 
    
    For $K=1$ and $\Mbar < \infty$, let
    $$M_i = \lrm{\lrbc{s \in \lrbc{s_{m,k}}_{m \in \lrbk{M}, k \in \lrbk{K}} : i \in \s}},$$
    so $\lambda_i = (\pi_n M)^{-1} M_i$. 
    $$\Var[P_n]{\G_{n,\etastarP}(t)}[\lambda] = \E[P_n]{\lrp{f_{t,\etastarPn}(W) - P_n f_{t,\etastarPn}}^2} \frac{1}{n} \sum_{i=1}^{n} \lambda_i^2.$$
    Without loss of generality, let $M = \Mbar$. 
    $$\frac{1}{n} \sum_{i=1}^{n} \lambda_i^2 = \frac{1}{\pi_n^2 \Mbar^2} \sum_{j=1}^{\Mbar} j^2 \frac{1}{n} \sum_{i=1}^{n} \I{M_i = j}.$$
    In \Cref{lemma.combinatorial}, I show that 
    $$\frac{1}{n} \sum_{i=1}^{n} \I{M_i = j} \Pnto \binom{\Mbar}{j} \pi^j (1 - \pi)^{\Mbar - j}.$$
    Hence, 
    $$\sum_{j=1}^{\Mbar} j^2 \frac{1}{n} \sum_{i=1}^{n} \I{M_i = j} \Pnto \sum_{j=1}^{\Mbar} j^2 \binom{\Mbar}{j} \pi^j (1 - \pi)^{\Mbar - j} = \pi (1 - \pi) \Mbar + (\pi \Mbar)^2,$$
    since the sum in the right is the second moment of a binomial distribution with parameters $\Mbar$ and $\pi$. 
    Collecting the results,
    $$\frac{1}{n} \sum_{i=1}^{n} \lambda_i^2 \Pnto 1 + (1 - \pi) \pi^{-1} M^{-1}.$$
    The Gaussian limit follows from Lindeberg's CLT conditional on $\lambda$ and the dominated convergence theorem, and the Cram\'er-Wold device. 

    Finally, let $K = 1$ and $\Mbar = \infty$. 
    I show that 
    \begin{equation} \label{eq.rep_Minf}
        \frac{1}{\sqrt{n}} \sum_{i = 1}^n \lambda_i \lrp{f_{t, \etastarPn}(W_i) - P_n f_{t, \etastarPn}} - \frac{1}{\sqrt{n}} \sum_{i = 1}^n \lrp{f_{t, \etastarPn}(W_i) - P_n f_{t, \etastarPn}}
    \end{equation}
    converges to zero in $L_2$. 
    For the mean, 
    \begin{align*}
        & \E[P_n]{\frac{1}{\sqrt{n}} \sum_{i = 1}^n \lrp{\lambda_i - 1} \lrp{f_{t, \etastarPn}(W_i) - P_n f_{t, \etastarPn}}} \\
        & = \sqrt{n} \E[P_n]{\lrp{\lambda_1 - 1} \lrp{f_{t, \etastarPn}(W_1) - P_n f_{t, \etastarPn}}} \\
        & = \sqrt{n} \E[P_n] \biggl[ \lrp{\lambda_1 - 1} \underbrace{\E[P_n]{f_{t, \etastarPn}(W_1) - P_n f_{t, \etastarPn} \Bigm| \lambda_1}}_{=0} \biggr]. \\
    \end{align*}
    For the variance, 
    \begin{align*}
        & \Var[P_n]{\frac{1}{\sqrt{n}} \sum_{i = 1}^n \lrp{\lambda_i - 1} \lrp{f_{t, \etastarPn}(W_i) - P_n f_{t, \etastarPn}}}  \\
        & = \E[P_n]{\Var[P_n]{\frac{1}{\sqrt{n}} \sum_{i = 1}^n \lrp{\lambda_i - 1} \lrp{f_{t, \etastarPn}(W_i) - P_n f_{t, \etastarPn}} \Biggm| \lambda }} \\
        & = \Var[P_n]{f_{t, \etastarPn}(W_i) - P_n f_{t, \etastarPn}} \E[P_n]{\frac{1}{n} \sum_{i=1}^{n} \lrp{\lambda_i - 1}^2} \\
        & = \Var[P_n]{f_{t, \etastarPn}(W_i) - P_n f_{t, \etastarPn}} \E[P_n]{\lrp{\lambda_1 - 1}^2},
    \end{align*}
    where the first equality follows since $\E[P_n]{f_{t, \etastarPn}(W) - P_n f_{t, \etastarPn} \Bigm| \lambda} = 0$ by the Law of Total Variance, and the second equality since the summands are iid conditional on $\lambda$. 
    Since $\lambda_1 - 1$ is bounded, $\E[P_n]{\lrp{\lambda_1 - 1}^2} \to 0$ if $\lambda_1 \Pnto 1$, which follows from 
    \begin{align*}
        \lambda_1 & = \lrp{\pi_n M}^{-1} M_1 \\
        & = \lrp{\pi_n}^{-1} \frac{1}{M} \sum_{m = 1}^M \I{1 \in s_{m,1}} \Pnto 1
    \end{align*}
    by a law of large numbers, since $\E[P_n]{\I{1 \in s_{m,1}}} = P_n(1 \in s_{m,1}) = \pi_n$ and splits are independent. 
    Finally, the Gaussian limit follows from Lindeberg's CLT and the Cram\'er-Wold device. 
\end{proof}

%% file: fast_convergence.tex
\section{Inference with Fast Converging Moments} \label{section.fast_convergence}

Consider the setting of \Cref{section.inference_normal}. 
The normal approximation CI \cref{eq.ci_h} may not cover $h(\theta_\etahat)$ with nominal probability when the variance of any moment function evaluated at the limit parameter $\theta_{\etastarP}$ is $0$, that is, 
\begin{equation} \label{eq.fast_moments}
    \Var[P]{\psi_{\theta_{\etastarP},\etastarP,j}(W)} = 0
\end{equation}
for any $j \in \lrbk{1,\dots,d}$. 
If that happens, either $\sigma^2_{\etastarP} = 0$, $\Psidot_{\etastarP}$ is not invertible, or both. 
If $\sigma^2_{\etastarP} = 0$, \cref{eq.clt_h} implies that the centered estimator multiplied by $\sqrt{n}$ converges in probability to zero, and the normal approximation in \cref{eq.ci_h} may not be accurate. 
Similarly, if $\Psidot_{\etastarP}$ is not invertible, $V^*_{\etastarP}$ is not well-defined, and the normal approximation may be inaccurate. 
In this subsection, I provide an approach to inference on $\theta_\etahat$ that is general in considering the class of Z-estimators in \Cref{section.z_estimators}.

I explore the fact that \cref{eq.fast_moments} implies that the empirical moment equation evaluated at $\theta_\etahat$ converges faster than the typical $\sqrt{n}$ rate to construct a confidence interval for $\theta_\etahat$ that is uniformly asymptotically valid regardless of whether \cref{eq.fast_moments} happens or not. 
The issue discussed in this section is not important for every application. 
First, I discuss examples of when one may or may not comfortably assume that \cref{eq.fast_moments} does not hold. 
Then, I propose a confidence interval, prove its uniform asymptotic validity, and characterize its power properties. 
I focus on the estimator $\thetahat_{\etahat} = \thetahat_{\etahat}^{(2)}$ from \Cref{section.z_estimators}, and the results can be extended to $\thetahat_{\etahat}^{(1)}$ and $\thetahat_{\etahat}^{(3)}$ using similar techniques.

\subsection{Examples} \label{section.fast_convergence.examples}

In many cases, the researcher can safely assume that \cref{eq.fast_moments} does not happen, depending on the setup and definition of $\psi_{\theta, \eta}$. 
In other cases, as in Section \ref{section.application_hte}, \cref{eq.fast_moments} can happen under one of the main hypotheses of interest. 
I present examples of both cases below.

\begin{RevExample}{example.classif_binary}
    In \Cref{example.classif_binary}, $W = (Y,X)$, $Y$ is binary, and $\etahat : \cX \to \lrbc{0,1}$ is a predictor of $Y$ using covariates $X$. 
    The parameter of interest is a split-sample Z-estimand with $\psi_{\theta, \eta}(y, x) = \I{y = \eta(x)} - \theta$: 
    $$\theta_\etahat = \frac{1}{M K} \sum_{r \in \cR} \sum_{\s \in r}  \int \I{y = \etahat_{\stilde}(x)} dP(y,x).$$
    The variance 
    $$\Var[P]{\psi_{\theta_{\etastarP}, \etastarP}(Y, X)} = P \lrp{Y = \etastarP(X)} \lrbk{1 - P \lrp{Y = \etastarP(X)}}$$
    is positive unless $\etastarP(X)$ always predicts $Y$ correctly or always incorrectly. 
    In practice, predictive algorithms rarely have a near perfect (or imperfect) performance, and in many cases the researcher can confidently assume $\operatorname{Var}_P [\psi_{\theta_{\etastarP}, \etastarP}(Y, X)] > 0$. 
\end{RevExample}

\begin{example} 
    Consider a dataset with covariates $X$, a mean zero continuous outcome $Y \in \R$, and the goal of assessing whether a predictor $\etahat(X)$ has predictive power for $Y$. 
    One way of assessing predictive power for $Y$ is by conducting inference on the covariance 
    $$\theta_\etahat = \frac{1}{M K} \sum_{r \in \cR} \sum_{\s \in r}  \int y \etahat_{\stilde}(x) dP(y,x).$$ 
    $\theta_\etahat$ is a Z-estimand with moment function $\psi_{\theta, \eta}(y, x) = y \eta(x) - \theta$, and its limit variance is 
    $$\Var[P]{\psi_{\theta_{\etastarP}, \etastarP}(Y, X)} = \Var[P]{Y \etastarP(X)}.$$
    Let $\etastar(x) = \E{Y | X = x}$ be the limit of $\etahat(x)$. 
    If $X$ has no predictive power for $Y$, for example because $Y$ and $X$ are independent, $\etastar(x) = 0$, and $\Var[P]{Y \etastarP(X)} = 0$. 
    Hence, the CI in \cref{eq.ci_h} may fail to achieve nominal coverage asymptotically. 
\end{example}

\begin{remark} \label{remark.ybar}
    \sloppy
    When $\operatorname{Var}_P [\psi_{\theta_{\etastarP}, \etastarP}(Y, X)] = 0$, the asymptotic distribution of $\thetahat_\etahat$ may depend on the specific structure of how $\etahat$ is calculated. 
    Let $Y$ be a mean zero scalar random variable, $K=2$, $M=1$, and $(\s, \stilde)$ be a 2-fold cross-split of the data of equal sizes. 
    Let $\psi_{\theta,\etahat_{\stilde}}(y) = y \bar{y}_{\stilde}^d - \theta$ for some odd positive $d$, where $\bar{y}_{\stilde} = \frac{1}{|\stilde|} \sum_{i \in \stilde} Y_i$. 
    Then, $\thetahat_\etahat = \frac{1}{2} \lrp{\bar{y}_{\s} \bar{y}_{\stilde}^d + \bar{y}_{\stilde} \bar{y}_{\s}^d}$ and $n^{1/2 + d/2} \thetahat_\etahat = \frac{1}{2} \lrp{(\sqrt{n} \bar{y}_{\s}) (\sqrt{n} \bar{y}_{\stilde})^d + (\sqrt{n} \bar{y}_{\stilde}) (\sqrt{n} \bar{y}_{\s})^d}$, which follows a non-trivial distribution that depends on $d$. 
    If, for example, $d=3$, $(\sqrt{n} \bar{y}_{\s})^d$ is approximately distributed as the cube of a standard normal distribution, and $d = 5$ leads to a different distribution. 
    Moreover, the dependence between $(\sqrt{n} \bar{y}_{\s}) (\sqrt{n} \bar{y}_{\stilde})^d$ and $(\sqrt{n} \bar{y}_{\stilde}) (\sqrt{n} \bar{y}_{\s})^d$ is not trivial. 
\end{remark}

\subsection{An Adaptive Confidence Interval} \label{section.fast_convergence.approach}

I show how to construct a confidence interval $\hat{C}_{1-\alpha}$ that satisfies 
$$\lim_{n \to \infty} \inf_{P \in \cP} P(\theta_{\etahat} \in \hat{C}_{1-\alpha}) = 1 - \alpha,$$
regardless of whether \cref{eq.fast_moments} may hold or not, by introducing a tuning parameter. 
In \Cref{section.diff_performance,appendix.hte}, I propose a different approach for the particular cases of inference on comparisons between models and in the Generic ML context of \citet{chernozhukov2025generic}, which explicitly account for the dependence across splits. 

I construct $\hat{C}_{1-\alpha}$ by inverting the test 
\begin{equation} \label{eq.fast_test}
    \begin{cases}
    H_{0,\etahat}: & h(\theta_{\etahat}) = \tau \\
    H_{A,\etahat}: & h(\theta_{\etahat}) \neq \tau,
\end{cases} 
\end{equation}
that is, $\hat{C}_{1-\alpha}$ contains all values of $\tau$ for which the null hypothesis is not rejected at significance level $\alpha$. 
My approach consists of a data-driven procedure to choose one of two p-values for testing \cref{eq.fast_test}: $p_c(\tau)$ or $p_e(\tau)$. 
$p_c(\tau)$ is a \textit{conservative} p-value, meant to be valid when \cref{eq.fast_moments} holds, that is, the p-value a researcher would use if they knew \cref{eq.fast_moments} were true. 
$p_e(\tau)$ is an \textit{exact} p-value, coming from the normal approximation \cref{eq.ci_h}, as it achieves exact nominal coverage in large samples when \cref{eq.fast_moments} does not hold. 
Hence, I test \cref{eq.fast_test} with the p-value $p_e(\tau)$ when the data suggest that the empirical moment equations are away from zero, and with $p_c(\tau)$ otherwise. 
The idea of using different tests based on pre-testing some condition (in this case, whether the empirical moment equations are away from zero), is similar to \citet{shi2015model}, in the context of moment inequalities. 
Specifically, 
$$p_e(\tau) = 2 \Phi \lrp{- \lrm{\frac{\sqrt{n} \lrp{h(\thetahat_\etahat) - \tau}}{\sigmahat_{\etahat}}}},$$
$$\Psihatmin(\tau) = \min_{j \in \lrbk{d}} \lrm{\frac{1}{M K} \sum_{r \in \cR} \sum_{\s \in r} \Psihat_{\s,\etahat_{\stilde},j}(\tau)},$$
$$\Psihat(\tau) = \norm{\frac{1}{M K} \sum_{r \in \cR} \sum_{\s \in r} \Psihat_{\s,\etahat_{\stilde}}(\tau)},$$
$$a_n(\tau) = \I{\Psihatmin(\tau) \Psihat(\tau) > \gamma_{n}},$$
$$p(\tau) = a_n(\tau) p_e(\tau) + \lrbk{1 - a_n(\tau)} p_c(\tau),$$
$$\hat{C}_{1-\alpha} = \lrbc{\tau \in \R : p(\tau) \le \alpha}.$$
The idea behind $a_n$ is that $\sqrt{n} \Psihatmin(\tau) \Pto 0$ when $\operatorname{Var}_P [\psi_{\theta_{\etastarP}, \etastarP}(W)] = 0$. 
The sequence $\gamma_{n}$ is a tuning parameter that should ideally be specified before the data analysis. 
The properties of $\gamma_{n}$ and $p_c(\tau)$ are specified in \Cref{as.fast_rate} below.

\begin{assumption} \label{as.fast_rate}
    The following conditions hold: 
    \begin{assumptionenum}
        \item \label{as.fast_pc} $\sup_{P \in \cP} P (p_c(\theta_\etahat) \le \alpha) \le \alpha$;
        \item \label{as.fast_gamma} $n \gamma_{n} \to \gamma \in (0, \infty)$;
        \item \label{as.fast_variance} The set $\cP$ can be decomposed as $\cP = \cP_{+} \bigcup \cP_0$, where 
        \begin{itemize}
            \item[(a)] For every $\varepsilon > 0$, 
            $$\sup_{P \in \cP_{+}} \sup_{\norm{\theta - \theta_{\etastarP}} > \varepsilon} - \norm{\Psi_{\etastarP}(\theta)} < 0 = \norm{\Psi_{\etastarP}(\theta_{\etastarP})},$$
            $\Psi_{\eta}$ is differentiable at $\theta_{\eta}$ for $\eta \in H$, and for some $\bar{c}_1 > 0$, 
            $$\inf_{P \in \cP_{+}} \lrm{\det \lrp{\Psidot_{\etastarP}}} \ge \bar{c}_1;$$
            \item[(b)] $\sup_{P \in \cP_0} \norm{\Psi_{\etastarP}(\theta_{\etastarP})} = 0$, $\theta_\etahat \Pto \theta_{\etastarP}$ for some $\theta_{\etastarP} \in \Theta'$ uniformly in $P \in \cP_0$, and 
            $\sup_{P \in \cP_0} \min_{j \in \lrbk{d}} \Var[P]{\psi_{\theta_\etastarP,\etastarP,j}(W)} = 0.$
        \end{itemize}
    \end{assumptionenum}
\end{assumption}
\Cref{as.fast_pc} requires the p-value $p_c(\tau)$ to be valid, even if conservative, including when $\operatorname{Var}_P [\psi_{\theta_{\etastarP}, \etastarP}(W)] = 0$. 
Constructing $p_c(\tau)$ is context-specific, but a conservative, trivially valid option is $p_c(\tau)=1$. 
Note that this option does not lead to an unbounded CI since $a_n(\tau) = 1$ with probability approaching one for values of $\tau$ far from $\theta_\etahat$. 
\Cref{as.fast_gamma} requires $\gamma_{n}$ to converge to zero at the $n^{-1}$ rate. 
\Cref{as.fast_variance} substitutes and weakens \Cref{as.z_unique_0} and \Cref{as.z_jacobian}. 
It allows $\Psidot_{\etastarP}$ to be singular and $\lVert \Psi_{\etastarP}(\theta) \rVert = 0$
to have multiple solutions for $\theta$ when the variance of $\psi_{\theta_\etastarP,\etastarP,j}(W)$ is zero for some $j$. Valid inference is achieved in these cases since $a_n = 0$ with probability approaching one. 
Note that \Cref{as.z_unique_0} and \Cref{as.z_jacobian} imply \Cref{as.fast_variance} since $\cP = \cP_{+}$, and if 
$$\inf_{P \in \cP} \min_{j \in \lrbk{d}} \Var[P]{\psi_{\theta_\etastarP,\etastarP,j}(W)} > 0,$$
\Cref{as.fast_variance} implies both \Cref{as.z_unique_0} and \Cref{as.z_jacobian}. 
I establish the uniform asymptotic validity of $\hat{C}_{1-\alpha}$, and explore its power properties.

\begin{theorem} \label{th.fast_size}
    \hyperlink{proof.th.fast_ci1}{(Uniform Asymptotic Validity of $\hat{C}_{1-\alpha}$)} \, \\
    Let Assumptions \ref{as.z_donsker}-\ref{as.z_psidot_etastar} and \ref{as.fast_rate} hold. 
    Then, 
    $$\liminf_{n \to \infty} \inf_{P \in \cP} P(\theta_{\etahat} \in \hat{C}_{1-\alpha}) \ge 1 - \alpha.$$
\end{theorem}
I show that the hypothesis test \cref{eq.fast_test}, where $H_{0,\etahat}$ is rejected if $p(\tau) > \alpha$, has power approaching $1$ for fixed alternatives and non-trivial power for some sequences of local alternative hypotheses. 
I compare my test with an oracle test that correctly picks $p_e(\tau)$ or $p_c(\tau)$ depending on the asymptotic behavior of $\thetahat_\etahat$. 
In order to study local power, I consider sequences $(P_n)_{n \ge 1} \subseteq \cP$ under different regimes for the limit behavior of $\sqrt{n} \norm{\theta_\etahat - \tau}$ and the variance of $\psi_{\theta_\etastarPn,\etastarPn}(W)$. 
Let 
$$v_n^2 = \min_{j \in \lrbk{d}} \Var[P_n]{\psi_{\theta_\etastarPn,\etastarPn,j}(W)}.$$
The oracle test is defined by 
$$p^*(\tau) = 
\begin{cases}
    p_c(\tau) \text{, if } v_n \to 0 \text{ and } \sqrt{n} \norm{\theta_\etahat - \tau} = O_{P_n}(1), \\ 
    p_e(\tau) \text{, otherwise}.
\end{cases}$$
This test is infeasible since it depends on the sequence of DGPs $(P_n)_{n \ge 1}$. 
For the different regimes, I compare the limits
\begin{align*}
    \pi_\alpha & = \lim_{n \to \infty} P_n \lrp{p(\tau) \le \alpha}, \\
    \pi^*_\alpha & = \lim_{n \to \infty} P_n \lrp{p^*(\tau) \le \alpha}.
\end{align*}

\begin{theorem} \label{th.fast_power}
    Let \Cref{as.z_donsker}-\Cref{as.z_psidot_etastar} and \Cref{as.fast_rate} hold, $\tau \in \R$, $\alpha \in (0,1)$, and $(P_n)_{n \ge 1}$ be a sequence such that the limits $v$, $\pi_\alpha$ and $\pi^*_\alpha$ exist. 
    Assume $p_c(\tau)$ is an independent Bernoulli random variable taking value $0$ with probability $\alpha$ and $1$ with probability $1-\alpha$ (that is, it rejects the null with probability $\alpha$). 
    Then, the relationships in \Cref{tab.power} hold, where each row defines a separate regime for $\sqrt{n} \norm{\theta_\etahat - \tau}$. 
    \input{table_power.tex}
\end{theorem}

\subsection{Proofs and Extra Definitions} \label{section.proofs_section.fast_convergence}

\begin{proof}[\hypertarget{proof.th.fast_size}{Proof of \Cref{th.fast_size}}] 

Let $(P_n)_{n \ge 1} \subseteq \cP$ be such that 
$$v = \lim_{n \to \infty} \min_{j \in \lrbk{d}} \Var[P_n]{\psi_{\theta_\etastarPn,\etastarPn,j}(W)}$$
exists. 

If $v > 0$, 
$p(\tau) \ge p_e(\tau)$, and 
$$P_n(\theta_{\etahat} \in \hat{C}_{1-\alpha}) \ge 1 - \alpha + o(1)$$
by \Cref{th.clt_z}. 

For $v = 0$, note that by \Cref{th.clt_general},
\begin{align*}
    & \sqrt{n} \frac{1}{M K} \sum_{r \in \cR} \sum_{\s \in r} \Psihat_{\s,\etahat_{\stilde}}(\theta_\etahat) \\
    & = \sqrt{n} \frac{1}{M K} \sum_{r \in \cR} \sum_{\s \in r} \lrp{\Psihat_{\s,\etahat_{\stilde}}(\theta_\etahat) - \Psi_{\etahat_{\stilde}}(\theta_\etahat)} \\
    & = \sqrt{n} \frac{1}{M K} \sum_{r \in \cR} \sum_{\s \in r} \lrp{\Psihat_{\s,\etastarPn}(\theta_\etahat) - \Psi_{\etastarPn}(\theta_\etahat)} + o_{P_n}(1) \\
    & = \sqrt{n} \frac{1}{M K} \sum_{r \in \cR} \sum_{\s \in r} \lrp{\Psihat_{\s,\etastarPn}(\theta_\etastarPn) - \Psi_{\etastarPn}(\theta_\etastarPn)} + o_{P_n}(1) \\
    & = \sqrt{n} \frac{1}{M K} \sum_{r \in \cR} \sum_{\s \in r} \Psihat_{\s,\etastarPn}(\theta_\etastarPn) + o_{P_n}(1), 
\end{align*}
and, for any $j \in \lrbk{d}$, 
$$\Var[P_n]{\frac{\sqrt{n}}{M K} \sum_{r \in \cR} \sum_{\s \in r} \Psihat_{\s,\etahat_{\stilde},j}(\theta_\etahat)} \le \Var[P_n]{\psi_{\theta_\etastarPn,\etastarPn,j}(W)} + o(1).$$
If $v=0$, 
$$\Var[P_n]{\frac{\sqrt{n}}{M K} \sum_{r \in \cR} \sum_{\s \in r} \Psihat_{\s,\etahat_{\stilde},{j_n}}(\theta_\etahat)} \to 0$$
for some $(j_n)_{n \ge 1}$, and hence 
$$\Var[P_n]{\min_{j \in \lrbk{d}} \lrm{\frac{\sqrt{n}}{M K} \sum_{r \in \cR} \sum_{\s \in r} \Psihat_{\s,\etahat_{\stilde},j}(\theta_\etahat)}} \to 0.$$
As a consequence, 
$$P_n \Bigl( \underbrace{\sqrt{n} \Psihatmin(\theta_\etahat)}_{=o_{P_n}(1)} \underbrace{\sqrt{n} \Psihat(\tau)}_{=O_{P_n}(1)} > \underbrace{n \gamma_n}_{=\gamma + o(1)} \Bigr) \to 0,$$
and $a_n(\theta_\etahat) \Pnto 0$, which concludes the proof. 
\end{proof}

\begin{proof}[\hypertarget{proof.th.fast_power}{Proof of \Cref{th.fast_power}}] 

    Define 
    $$\Psi_n(\tau) = \norm{P_n \frac{1}{M K} \sum_{r \in \cR} \sum_{\s \in r} \Psihat_{\s,\etahat_{\stilde}}(\tau)}.$$
    
    First, let $\sqrt{n} \norm{\theta_\etahat - \tau} \Pnto \infty^*$. 
    If $n v_n \to \infty$,
    $$P_n \Bigl( \underbrace{\sqrt{n} \Psihatmin(\theta_\etahat)}_{=O_{P_n}(1)} \Bigl[ \underbrace{\sqrt{n} \Psihat(\tau) - \sqrt{n} \Psi_n(\tau)}_{=O_{P_n}(1)} + \underbrace{\sqrt{n} \Psi_n(\tau)}_{\Pnto \infty} \Bigr] \underbrace{\sqrt{n} \Psihat(\tau)}_{=O_{P_n}(1)} > \underbrace{n \gamma_n}_{=\gamma + o(1)} \Bigr) \Pnto 1,$$
    since 
    $$\sqrt{n} \lrp{\Psi_{\etahat_{\xitilde}}(\tau) - \Psi_{\etahat_{\xitilde}}(\theta_\etahat)} = \lrp{\Psidot_\etastarPn + o_{P_n}(1)} \sqrt{n} \lrp{\tau - \theta_\etahat}.$$
    If $n v_n = O(1)$,
    $$P_n \Bigl( \underbrace{v_n^{-1} \sqrt{n} \Psihatmin(\theta_\etahat)}_{=O_{P_n}(1)} \Bigl[ \underbrace{v_n \sqrt{n} (\Psihat(\tau) - \Psi_n(\tau))}_{=o_{P_n}(1)} + \underbrace{v_n \sqrt{n} \Psi_n(\tau)}_{=O_{P_n}(1)} \Bigr] \underbrace{\sqrt{n} \Psihat(\tau)}_{=O_{P_n}(1)} > \underbrace{n \gamma_n}_{=\gamma + o(1)} \Bigr)$$
    is $O_{P_n}(1)$, and $a_n(\tau) =O_{P_n}(1)$. 
    If $n v_n = o(1)$, $v_n \sqrt{n} \Psi_n(\tau) = o_{P_n}(1)$, and $a_n \Pnto 0$. 

    Second, let $\sqrt{n} \norm{\theta_\etahat - \tau} = O_{P_n}(1)$, and hence $\sqrt{n} \Psihat(\tau) = O_{P_n}(1)$. 
    It follows that 
    $$P_n \Bigl( \underbrace{v_n^{-1} \sqrt{n} \Psihatmin(\theta_\etahat)}_{=O_{P_n}(1)} \underbrace{\sqrt{n} \Psihat(\tau)}_{=O_{P_n}(1)} > v_n^{-1} \underbrace{n \gamma_n}_{=\gamma + o(1)} \Bigr)$$
    is $O_{P_n}(1)$ if $v_n^2 \to v^2 > 0$, and converges to zero if $v_n \to 0$. 

    Finally, the last row follows since $p_e(\theta_\etahat) - p_e(\tau) = o_{P_n}(1)$. 
    Note that $p_e(\theta_\etahat) - p_e(\tau) > o_{P_n}(1)$ for the second row. 
\end{proof}

%% file: table_power.tex
\begin{table}[!ht]
    \centering
    \setlength{\arrayrulewidth}{1.0pt} 
    \caption{Power Comparison by Regime}
    \renewcommand{\arraystretch}{1.5} 
    \label{tab.power}
    \begin{tabular}{|l|c|c|c|c|}
        \cline{2-5}
        \multicolumn{1}{c|}{} & \multicolumn{1}{c|}{$n v_n^2 = o(1)$} & \multicolumn{1}{c|}{$n v_n^2 = O(1)^*$} & \multicolumn{1}{c|}{$n v_n^2 \to \infty^{**}$} & \multicolumn{1}{c|}{$v_n^2 \to v^2 > 0$} \\
        \cline{2-5}
        \hline
        $\sqrt{n} \norm{\theta_\etahat - \tau} \Pnto \infty$ 
        & $\alpha = \pi_\alpha < \pi^*_\alpha$
        & $\alpha < \pi_\alpha < \pi^*_\alpha$
        & $\pi_\alpha = \pi^*_\alpha = 1$
        & $\pi_\alpha = \pi^*_\alpha = 1$ \\
        \hline
        $\sqrt{n} \norm{\theta_\etahat - \tau} = O_{P_n}(1)^{***}$ 
        & $\alpha = \pi_\alpha < \pi^*_\alpha$
        & $\alpha = \pi_\alpha < \pi^*_\alpha$
        & $\alpha = \pi_\alpha < \pi^*_\alpha$
        & $\alpha < \pi_\alpha < \pi^*_\alpha$ \\
        \hline
        $\sqrt{n} \norm{\theta_\etahat - \tau} \Pnto 0$ 
        & $\alpha = \pi_\alpha = \pi^*_\alpha$
        & $\alpha = \pi_\alpha = \pi^*_\alpha$
        & $\alpha = \pi_\alpha = \pi^*_\alpha$
        & $\alpha = \pi_\alpha = \pi^*_\alpha$ \\
        \hline
    \end{tabular}\\[6pt]
    \raggedright \footnotesize * Assumes $n v_n^2 \nrightarrow 0$; ** Assumes $v_n^2 \to 0$; *** Assumes $\sqrt{n} \norm{\theta_\etahat - \tau} \neq o_{P_n}(1)$. 
\end{table}

%% file: appendix_comparison.tex
\section{Note on Comparing Two Nonparametric Models} \label{appendix.comparison}

I discuss an extension of the setting of \Cref{section.diff_performance} for comparing $\theta_\etahat$ to the performance of another model $\etahat'$, computed with the same split-sample approach as $\etahat$. 

Let 
$$\cS = \lrp{\s_{m,k}}_{m \in \lrbk{M}, k \in \lrbk{K}}$$
and denote the split-specific models $\etahat = (\etahat_{\stilde})_{\s \in \cS}$ and $\etahat' = (\etahat'_{\stilde})_{\s \in \cS}$, where $\etahat_{\stilde} = \cA(D_{\stilde})$ and $\etahat'_{\stilde} = \cA'(D_{\stilde})$, that is, the two models are trained using the same sample but different algorithms. 
For example, $\etahat$ could be estimated with random forests while $\etahat'$ could be estimated with neural networks. 
Denote 
$$\deltahat_{\etahat, \etahat'}^{(1)} = \lrp{\thetahat_{\etahat_{\stilde}} - \thetahat_{\etahat'}}_{\s \in \cS}$$
and 
$$\deltahat_{\etahat, \etahat'}^{(2)} = \lrp{\thetahat_{\etahat'_{\stilde}} - \thetahat_{\etahat}}_{\s \in \cS}.$$

$\deltahat_{\etahat, \etahat'}^{(1)}$ can be used for testing whether $\thetahat_{\etahat_{\stilde}} \ge \thetahat_{\etahat'}$ for all $\s \in \cS$ versus the alternative that $\thetahat_{\etahat_{\stilde}} < \thetahat_{\etahat'}$ for at least one $\s \in \cS$, similarly to \Cref{section.diff_onesided_test} and \Cref{th.diff_onesided_test}. 
Note that the Donsker and rate conditions in \Cref{as.diff_onesided_test_bP} are not required for \Cref{th.diff_onesided_test}. 
They are used only for the pointwise \Cref{th.diff_pointwise} to cover the case $\theta_\etastarP = \theta_{b_P}$. 
Similarly, $\deltahat_{\etahat, \etahat'}^{(2)}$ can be used to test whether $\thetahat_{\etahat'_{\stilde}} \ge \thetahat_{\etahat}$ for all $\s \in \cS$ versus the alternative that $\thetahat_{\etahat'_{\stilde}} < \thetahat_{\etahat}$ for at least one $\s \in \cS$.

%% file: appendix_additional_results.tex
\section{Additional Results} \label{app.additional_results}

\begin{proposition} \label{prop.avg_implies_empproc}
    In the context of \Cref{as.emp_proc}, let $T=\lrbc{t}$ and $\cP=\lrbc{P}$ be singletons, and let \Cref{as.avg} hold for $P$ and $f_{\eta,t}$ measurable with $f_{\eta,t}(w) < \infty$ for all $w$. 
    Then, \Cref{as.emp_proc} holds. 
\end{proposition}

\begin{proof}[\hypertarget{proof.prop.avg_implies_empproc}{Proof of \Cref{prop.avg_implies_empproc}}]  
    \Cref{as.total_bounded} and \Cref{as.measurability} hold trivially. 
    \Cref{as.envelope} holds by taking $f_{\eta,t}$ as its own envelope. 
    The uniform integrability condition \Cref{as.unif_integrability} is implied by the $2+\delta$ assumption \Cref{as.2plusdelta}. 
    \Cref{as.equicontinuity} holds trivially. 
    \Cref{as.entropy} holds since both covering and bracketing numbers are equal to $1$ with singleton $T$. 
    Finally, \Cref{as.etahat} follows since 
    $$\E[P]{\Var[P]{f_{\etahat_{\xitilde}}(W) - f_{\etastar}(W) \Bigm| D_{\xitilde} }} \to 0,$$
    as established under \Cref{as.avg_etahat} in the proof of \Cref{th.clt_avg} \cref{eq.proof_th.clt_avg.7}, since convergence in $L_1$ implies convergence in probability.
\end{proof}

\begin{lemma} \label{lemma.combinatorial}
    In the context of \Cref{th.clt_general}, 
    $$Z_n = n^{-1} \sum_{i = 1}^n \I{M_i = c} \Pnto \dbinom{M}{c} \pi^c (1-\pi)^{M-c}.$$
\end{lemma}

\begin{proof}
    I show that $\E[P_n]{Z_n} = \dbinom{M}{c} {\pi_n}^c (1-{\pi_n})^{M-c}$ and $\Var[P_n]{Z_n} \to 0$ as $n \to \infty$. 
    By definition, $M_i = \lrm{\lrbc{s \in \lrbc{s_{m,1}}_{m \in \lrbk{M}} : i \in s}}$. 
    $\E[P_n]{Z_n} = \E[P_n]{\I{M_1 = c}} = P_n \lrp{M_1 = c}$ since all $M_i$ are equally distributed for any $i$. 
    The event $\lrbc{M_i = c}$ is equivalent to the event that observation $i$ is chosen in exactly $c$ of the $M$ splits of the sample. 
    Since the splits are independent, $M_i$ follows a binomial distribution with parameters $M$ and $\pi_n$. 
    Hence, the probability of this event is $\dbinom{M}{c} {\pi_n}^c (1-{\pi_n})^{M-c}$. 

    To show that $\Var[P_n]{Z_n} \to 0$, I use the fact that 
    \begin{align*}
        \Var[P_n]{Z_n} & = n^{-2} \sum_{i = 1}^n \Var[P_n]{\I{M_i = c}} + n^{-2} \sum_{i \neq j} \Cov[P_n]{\I{M_i = c}, \I{M_j = c}} \\
        & = n^{-1} \Var[P_n]{\I{M_1 = c}} + n^{-2} n (n - 1) \Cov[P_n]{\I{M_1 = c}, \I{M_2 = c}}.
    \end{align*}
    Hence, it's enough to show that
    $$\Cov[P_n]{\I{M_1 = c}, \I{M_2 = c}} = P_n \lrp{M_1 = c, M_2 = c} - P_n \lrp{M_1 = c}^2 \to 0.$$
    I show that $P_n \lrp{M_1 = c \Bigm| M_2 = c} \to P_n \lrp{M_1 = c}$. 
    Note $b = \pi_n n$ is the number of draws in each split. 
    Using combinatorial arguments, the conditional probability is given by
    \begin{align*}
        P_n \lrp{M_1 = c \Bigm| M_2 = c} = \sum_{t=0}^{c} & \dbinom{c}{t} \dbinom{M - c}{c - t} \left( \frac{\dbinom{n - 2}{b - 2}}{\dbinom{n - 1}{b - 1}} \right)^t \left( 1 - \frac{\dbinom{n - 2}{b - 2}}{\dbinom{n - 1}{b - 1}} \right)^{c - t} \\
        & \times \left( \frac{\dbinom{n - 2}{b - 1}}{\dbinom{n - 1}{b}} \right)^{c - t} \left( 1 - \frac{\dbinom{n - 2}{b - 1}}{\dbinom{n - 1}{b}} \right)^{M - 2c + t}.
    \end{align*}
    $t$ represents the number of splits that contain both observations 1 and 2. 
    Since observation 2 is chosen in $c$ splits, $0 \le t \le c$. 
    There are $\dbinom{c}{t}$ ways of choosing among the $c$ splits that contain observation 2, which $t$ will also contain observation 1. 
    There are $\dbinom{M - c}{c - t}$ ways of choosing the remaining $c - t$ splits that contain observation 1 but not 2. 
     $\left( \frac{\dbinom{n - 2}{b - 2}}{\dbinom{n - 1}{b - 1}} \right)^t$ is the probability of choosing observation 1 in the $t$ splits that contain both observations. 
    $\left( 1 - \frac{\dbinom{n - 2}{b - 2}}{\dbinom{n - 1}{b - 1}} \right)^{c - t}$ is the probability of not choosing observation 1 in the remaining $c - t$ splits that contain observation 2. 
    $\left( \frac{\dbinom{n - 2}{b - 1}}{\dbinom{n - 1}{b}} \right)^{c - t}$ is the probability of choosing observation 1 in the $c - t$ splits that contain observation 1 but not 2. 
    Finally, $\left( 1 - \frac{\dbinom{n - 2}{b - 1}}{\dbinom{n - 1}{b}} \right)^{M - 2c + t}$ is the probability of not choosing observation 1 in the remaining $M - 2c + t$ splits that contain neither observation.

    For large $n$, we can approximate the combinatorial terms:
    $$\frac{\dbinom{n - 2}{b - 2}}{\dbinom{n - 1}{b - 1}} = \frac{(n - 2)!}{(b - 2)! (n - b)!} \lrp{\frac{(n - 1)!}{(b - 1)! (n - b)!}}^{-1} = \frac{b - 1}{n - 1} = {\pi_n} + o(1).$$
    Similarly,
    $$\frac{\dbinom{n - 2}{b - 1}}{\dbinom{n - 1}{b}} = \frac{(n - 2)!}{(b - 1)! (n - b - 1)!} \lrp{\frac{(n - 1)!}{b! (n - b - 1)!}}^{-1} = \frac{b}{n - 1} = {\pi_n} + o(1).$$
    It follows that 
    \begin{align*}
        P_n \lrp{M_1 = c \Bigm| M_2 = c} & = \sum_{t=0}^{c} \dbinom{c}{t} \dbinom{M - c}{c - t} {\pi_n}^t (1 - {\pi_n})^{c - t} {\pi_n}^{c - t} (1 - {\pi_n})^{M - 2c + t} + o(1) \\
        & = {\pi_n}^c (1 - {\pi_n})^{M - c} \sum_{t=0}^{c} \dbinom{c}{t} \dbinom{M - c}{c - t} + o(1) \\
        & = {\pi_n}^c (1 - {\pi_n})^{M - c} \dbinom{M}{c} + o(1) \\
        & = P_n \lrp{M_1 = c} + o(1),
    \end{align*}
    where the third equality uses Vandermonde's Identity. 
\end{proof}

%% file: application_rct.tex
\section{Covariate Adjustment in Randomized Trials} \label{appendix.application_rct}

Let $W = (Y, A, X)$, where $Y \in \R$ is an observed outcome, $A$ is a binary (randomized) treatment assignment indicator, and $X \in \cX \subseteq \R^d$ is a set of covariates, for some $d \ge 1$. 
Let $Y(1), Y(0)$ denote potential outcomes respectively under treatment and control, and $Y = A Y(1) + (1 - A) Y(0)$. 
In the simplest form of an RCT, $A \perp (X, Y(1), Y(0))$. 
In this setting, the ATE $\theta$ can be identified from the regression 
\begin{equation} \label{eq.rct_basic}
    Y = \alpha + \theta A + \varepsilon.
\end{equation}

The covariates are not necessary for identification of $\theta$. 
However, adding regressors in (\ref{eq.rct_basic}) can lead to power improvement by reducing the variance of the error term $\varepsilon$ and thus the asymptotic variance of the least squares estimator of $\theta$. 
One approach to incorporating covariates is through a covariate-adjustment term $\eta(X)$: 
\begin{equation} \label{eq.rct_eta}
    Y = \alpha_\eta + \thetaeta D + \beta_\eta \eta(X) + \varepsilon.
\end{equation}
If $A \perp (X, Y(1), Y(0))$, $\thetaeta = \theta$ does not depend on $\eta$. Still, its OLS estimator $\thetahat_\eta$ does depend on $\eta$. 
In practice, one needs to estimate $\eta$ with a model $\etahat$. 
Inference becomes challenging if the same data is used to estimate both $\etahat$ and $\thetahatetahat$ because the observations in \cref{eq.rct_eta} become no longer iid. 
The asymptotic distribution of $\sqrt{n} (\thetahat_\etahat - \theta_\etahat)$ can be characterized following \Cref{section.z_estimators}, specifically \Cref{th.clt_z}.